\documentclass{article}

\usepackage[flushleft]{threeparttable}
\usepackage[utf8]{inputenc} % allow utf-8 input
\usepackage[T1]{fontenc}    % use 8-bit T1 fonts
\usepackage{hyperref}       % hyperlinks
\usepackage{url}            % simple URL typesetting
\usepackage{booktabs}       % professional-quality tables
\usepackage{amsfonts}       % blackboard math symbols
\usepackage{nicefrac}       % compact symbols for 1/2, etc.
\usepackage{microtype}      % microtypography
\usepackage{xcolor}         % colors
\usepackage{multicol}
\usepackage{multirow}
\usepackage{subcaption}
\usepackage{fullpage}
\usepackage{comment}
\usepackage{natbib}
\usepackage{amsmath}

\allowdisplaybreaks

\title{Differentially Private Continual Releases of Streaming Frequency Moment Estimations}
\author{Alessandro Epasto <aepasto@google.com>,\\
Jieming Mao <maojm@google.com>,\\
Andres Munoz Medina <ammedina@google.com>,\\
Vahab Mirrokni <mirrokni@google.com>,\\
Sergei Vassilvitskii <sergeiv@google.com>\\
Peilin Zhong <peilinz@google.com>
}
\usepackage{geometry}
\usepackage{mathrsfs}
\usepackage{amssymb, amsmath, latexsym, amsfonts, graphicx, amsthm}

\usepackage[ruled]{algorithm2e}
\newcommand{\jm}[1]{{\color{blue}{Jieming: #1}}}
\newcommand{\pz}[1]{{\color{red}{Peilin: #1}}}
\newcommand{\lap}{\text{Lap}}

\newcommand{\sm}{\textbf{Sum}}

\newcommand{\dis}{\mathrm{dist}}

\newcommand{\hh}{\textbf{-HH}}
\newcommand{\poly}{\mathrm{poly}}
\newcommand{\xhdr}[1]{\vspace{2mm} \noindent{\bf #1}}

\newtheorem{theorem}{Theorem}[section]
\newtheorem{observation}[theorem]{Observation}
\newtheorem{lemma}[theorem]{Lemma}
\newtheorem*{theorem*}{Theorem}
\newtheorem*{lemma*}{Lemma}
\newtheorem{definition}[theorem]{Definition}
\newtheorem{corollary}[theorem]{Corollary}
\newtheorem{remark}[theorem]{Remark}
\newtheorem{claim}[theorem]{Claim}

\date{April 2022}

\SetKwInOut{Parameter}{Parameters}

\begin{document}

\maketitle

\begin{abstract}
The streaming model of computation is a popular approach for working with large-scale data. In this setting, there is a stream of items and the goal is to compute the desired quantities (usually data statistics) while making a single pass through the stream and using as little space as possible.

Motivated by the importance of data privacy, we develop differentially private streaming algorithms under the continual release setting, where the union of outputs of the algorithm at every timestamp must be differentially private. Specifically, we study the fundamental $\ell_p$ $(p\in [0,+\infty))$ frequency moment estimation problem under this setting, and give an $\varepsilon$-DP algorithm that achieves $(1+\eta)$-relative approximation $(\forall \eta\in(0,1))$ with $\mathrm{poly}\log(Tn)$ additive error and uses $\mathrm{poly}\log(Tn)\cdot \max(1, n^{1-2/p})$ space, where $T$ is the length of the stream and $n$ is the size of the universe of elements.
Our space is near optimal up to poly-logarithmic factors even in the non-private setting.

To obtain our results, we first reduce several primitives under the differentially private continual release model, such as counting distinct elements, heavy hitters and counting low frequency elements, to the simpler, counting/summing problems in the same setting. 
Based on these primitives, we develop a differentially private continual release level set estimation approach to address the $\ell_p$ frequency moment estimation problem.

We also provide a simple extension of our results to the harder sliding window model, where the statistics must be maintained over the past $W$ data items.

\end{abstract}

\section{Introduction}

Data privacy is a central concern in the deployment of real-world computational systems. In the vast literature on privacy in computation~\cite{SURVEY}, the notion of differential privacy (DP)~\cite{dwork2008differential,dwork2014algorithmic} has remained the {\it de facto} standard for more than a decade. The classical formulation of differential privacy assumes that the data is static~\cite{dwork2008differential}, and that the data curator is interested obtaining answers to a \emph{predetermined} number of queries on the dataset.

Real world applications, however, often require the analysis of user datasets that are organically and rapidly growing. This is illustrated by the popular {\it streaming model} of computation~\cite{Morris78a,alon1996space}, where data arrives over time, and at each update, a new solution is output by the algorithm. In such streaming applications, the {\it continual release} model of differential privacy~\cite{DworkNPR10,ChanSS11} promises a rigours guarantee of privacy: an observer of {\it all} the outputs of the algorithm is information-theoretically bounded in the ability to learn about the {\it existence} of an individual data point in the stream.

In this paper, we focus on two fundamental challenges in the field of private streaming algorithms: the insertion only, or streaming, model and the sliding window model. In the former model, a data curator receives a stream of data $a_1,a_2, \ldots$ and at each time $t$ releases a statistical query depending on all data received up to that point. In the latter, the computation depends only on the last $W$ items observed by the data curator. 

The sliding window model may be generally more practically relevant compared to the streaming model as it allows to account for information freshness and in some cases it can be a legal or privacy requirement as well. For instance, in some situations, data privacy laws such as the General Data Protection Regulation (GDPR)\footnote{https://gdpr-info.eu/art-17-gdpr/} do not allow unlimited retention of user data.

Our main contribution is to provide private algorithms for a series of foundational streaming problems under both the streaming and sliding window model. 

\xhdr{Motivating example: Privacy Sandbox.} We present a concrete practical application of our results. As part of the {\em Privacy Sandbox} initiative, Chrome has developed a series of APIs to reduce cross site tracking while supporting the digital advertising ecosystem. A key part of one of the proposals is a {\em k-Anonymity Server}.\footnote{\url{https://github.com/WICG/turtledove/blob/main/FLEDGE_k_anonymity_server.md}} The server ensures that each ad creative that is reported to advertisers has won its respective auction at least $k$ times over a particular time window. Abstracting the specifics, this problem requires computing the number of distinct elements over a sliding window.  
%
%In a nutshel, the k-anonymity server acts as safeguard in the FLEDGE advertisement system~\cite{ADDCITATION} in Chrome by restricting the use of interest groups that have not been reported, by at least $k$ users, within a given time to live. This is done for obvious privacy reasons and requires the computation of the cardinality of distinct elements (i.e. users) associated with a given interest group within a recent time frame. To compute such set of k-anonymous groups the server needs to process a large user datasets (i.e., the presence of users in a group) in a streaming setting (as data arrives over time), considering only recent activity. 
Moreover, to further strengthen privacy protections, the computation itself should be made differentially private, which is precisely the setting we consider in this work. 
%has announced that to protect such user data in the server, the system will consider a notion of differential privacy. 

The previous example elucidates a concrete motivation for the study of sliding window algorithms for counting distinct element problems with differential privacy in the continual release setting. 
The rest of the paper proceeds by formalizing this model and our results in this space.

In particular, we study a more general class of statistics of the input data than the problem of counting distinct elements: the $\ell_p$ frequency moments problem.
The $\ell_p$ frequency moment is the sum of the $p$-th power of the frequencies of the elements. 
The number of distinct elements is a special case for $p=0$.
The $\ell_p$ frequency moment problem is one of the most fundamental problems in the streaming literature. 
The first non-trivial algorithm for $p=1$ is \cite{morris1978counting}. 
Later \cite{flajolet1985probabilistic} is the first to study the case for $p=0$.
\cite{alon1996space}  initiates the study for $p=2$ and other $p\in [0,\infty)$.
After the developments over several decades, the spaces of the current best $\ell_p$ frequency moment estimation algorithms are near optimal for all $p\in[0,\infty)$, i.e. they almost match the proven space lower bounds (see e.g.,~\cite{ganguly2011polynomial,kane2011fast,kane2010exact,flajolet2007hyperloglog,li2013tight}).

However the landscape of DP streaming $\ell_p$ frequency moment is mysterious even in the non-continual release setting.
Most existing work only studied for $p=0,1,2$ (see more discussion in Section~\ref{sec:related}).
The work of~\cite{WangPS22} considered general $p\in (0,1]$.
A recent independent work~\cite{blocki2022make} studied all $p\in [0,\infty)$.
%The only known work that studied $p$ in other range is~\cite{WangPS22} which considered $p\in(0,1]$ but still not in the continual release setting.
But none of~\cite{blocki2022make,WangPS22} considered continual release setting.
In the DP continual release setting, \cite{Bolot13} studied the count of distinct elements but not in the low space streaming setting.
In the DP streaming continual release setting, existing work~\cite{dwork2008differential,ChanSS11} only studied the case for $p=1$.
No previous algorithm for $p\not=1$ is known in the DP low space streaming continual release setting.

For $\ell_p$ frequency moment in the sliding window model, there are known techniques~\cite{datar2002maintaining,BravermanO07} which can convert the streaming algorithm into sliding window algorithm using some additional small space.
We show how to extend these techniques to convert our DP streaming continual release streaming algorithms to DP sliding window continual release algorithms.

\subsection{Computational Model}
In this paper, we consider a streaming setting with $T$ timestamps. 
At each timestamp $t \in [T]$, we get an input $a_t \in \mathcal{U} \cup \{\perp\}$, where $\mathcal{U}$ represents the universe of all possible input elements, and $\perp$ represents empty. 
We sometimes also consider an input stream of integers, i.e., at each timestamp $t\in[T]$, we get an input $a_t\in\mathbb{Z}$.
The goal is to compute some function $g(\cdot)$ based on the inputs.
Instead of only computing $g$ at the end of the stream, we consider the continual release model throughout the paper:
\begin{itemize}
\item \textbf{Continual Release Model}: At every timestamp $t\in [T]$, we want to output $g(\cdot)$ based on the data $a_1,a_2,\cdots,a_t$.
\end{itemize}

We consider two different streaming models depending on the range of inputs that $g$ is based on.

\begin{itemize}
\item \textbf{(Insertion-only) Streaming Model}: In this model, $g$ depends on all past inputs, i.e. at timestamp $t \in [T]$, we want to compute $g(a_1,...,a_t)$.
Unless otherwise specified, we use streaming model to refer to insertion-only streaming model throughout the paper.
\item \textbf{Sliding Window Model}: In this model, we have a parameter $W \in \mathbb{Z}_{\geq 1}$ for window size. 
$g$ depends on the last $W$ inputs, i.e. at timestamp $t \in [T]$, we want to compute $g(a_{\max(t-W+1, 1)},...,a_t)$.
\end{itemize}

In this paper, we are particularly interested in algorithms that use space %and per-round running time 
sub-linear in $T$.
For a (sub-)stream $\mathcal{S}$, we use $\|\mathcal{S}\|_p^p$ to denote the $\ell_p$ frequency moment of $\mathcal{S}$.
In particular, for $p=0$, $\|\mathcal{S}\|_0$ denotes the number of distinct elements.
For $p=1$, $\|\mathcal{S}\|_1$ denotes the number of non-empty elements. 
We refer readers to preliminaries section (Section~\ref{sec:preli}) for detailed notation and definitions.

\subsection{Our Results and Comparison to Prior Work}
In this section, we give a brief overview of our results and comparison with prior work.
We use $n$ to denote the size of the universe $\mathcal{U}$.
We use $(\alpha,\gamma)$-approximation to specify the approximation guarantee with multiplicative factor $\alpha$ and additive error $\gamma$.
In all of our results, we use $\varepsilon\geq 0$ for the DP parameter, and use $\eta\in(0,0.5)$ in the relative error.

\subsubsection{Differentially Private Streaming Continual Release Algorithms}

We developed a series of DP algorithms for solving frequency moments estimation and its related problems in the streaming continual release model. 

\xhdr{$\ell_p$ Frequency moment estimation $(p\in[0,\infty))$.}
Our main result is a general $\ell_p$ frequency moment estimation algorithm which works for all $p\in [0,\infty)$.

\begin{theorem}[$\ell_p$ Frequency moment, informal version of Theorem~\ref{thm:lp_moment}]
There is an $\varepsilon$-DP algorithm in the streaming continual release model such that with probability at least $0.9$, it always outputs an $\left(1+\eta, \left(\frac{\log(Tn)}{\eta\varepsilon}\right)^{O(\max(1,p))}\right)$-approximation to $\|\mathcal{S}\|_p^p$ for every timestamp $t$, where $\mathcal{S}$ denotes the stream up to timestamp $t$.
The algorithm uses space
$
\max(1,n^{1-2/p})\cdot \left(\frac{\log(Tn)}{\eta\varepsilon}\right)^{O(\max(1,p))}.
$
\end{theorem}
To the best of our knowledge, we are the first to study the general $\ell_p$ frequency moment estimation problem in the differentially private streaming continual release setting.
\cite{DworkNPR10} and \cite{ChanSS11} studies the summing problem in the same setting, where the summing problem can be seen as a special case for $p=1$. 
\cite{WangPS22} studies the streaming $\ell_p$ frequency moment estimation for $p\in (0,1]$ based on the $p$-stable distribution, and a concurrent independent work~\cite{blocki2022make} studies the case for $p\in(0,1]$, but it is not clear how to generalize their techniques to the continual release model, i.e., their approach only provides the differential privacy guarantee of the output at the end of the stream.
In addition, the approach of \cite{WangPS22} does not achieve $\varepsilon$-DP for an arbitrarily small $\varepsilon>0$, and they also mention that their technique might not be easily extended to the case for $p>1$. 

Our space usage is near optimal up to poly-logarithmic factors even when comparing with the non-private streaming $\ell_p$ frequency moment estimation algorithms: for $p\leq 2$, the space needed for both our algorithm and previous non-private algorithm (see e.g.,~\cite{kane2010exact}) is poly-logarithmic, for $p>2$, the space needed for both our algorithm and previous non-private algorithm~\cite{indyk2005optimal} is $\tilde{O}(n^{1-2/p})$\footnote{We use $\tilde{O}(g)$ to denote $g\cdot\poly(\log(g))$}. 
Note that $\Omega(n^{1-2/p})$ space is a proven lower bound for $p>2$ even in the non-private case~\cite{saks2002space,bar2004information}.

\xhdr{Summing ($\ell_1$ frequency moment estimation).}
The easiest problem that is related to the $\ell_p$ frequency moment estimation problem would be the summing problem: the goal is to compute the summation of the input numbers.
Note that $\ell_1$ frequency moment estimation is a special case of the summing of a binary stream, i.e., we regard $\perp$ as $0$ and all other elements as $1$.

\begin{theorem}[Summing of a non-negative stream, informal version of Theorem~\ref{thm:summing_non_negative}]
There is an $\varepsilon$-DP algorithm for summing problem in the streaming continual release model.
If the input numbers are guaranteed to be non-negative, with probability at least $0.9$, the output is always a $\left(1+\eta, O_{\varepsilon,\eta}(\log T)\right)$-approximation to the sum of all input numbers at any timestamp $t\in [T]$.
The algorithm uses space $O(1)$.
\end{theorem}

The summing problem was studied by~\cite{DworkNPR10,ChanSS11} in the differentially private streaming continual release model. 
Their approximation has $O(\log^{2.5} T)$ additive error.
In our work, we show that if we allow $(1+\eta)$ relative error and work on the stream with non-negative numbers only, we can reduce the additive error to $O(\log T)$.
This is useful when we cannot avoid the relative error for some problem (such as the number of distinct elements) in the streaming model but we still need summing as a subroutine.

\xhdr{Counting distinct elements ($\ell_0$ frequency moment estimation).}
Counting distinct elements if one of the fundamental problems in the streaming literature.
The goal is to estimate the number of distinct elements that appeared in the stream.
We provide a DP streaming continual release algorithm for counting distinct elements.

\begin{theorem}[Number of distinct elements, informal version of Corollary~\ref{cor:distinct_large_universe_better_additive}]
There is an $\varepsilon$-DP algorithm for the number of distinct elements in the streaming continual release model.
With probability at least $0.9$, the output is always a $\left(1+\eta,O_{\varepsilon,\eta}\left(\log^2(T)\right))\right)$-approximation for every timestamp $t\in [T]$. 
The algorithm uses $\poly\left(\frac{\log(T)}{\eta\min(\varepsilon,1)}\right)$ space.
\end{theorem}

In the non-private streaming setting, counting distinct element can be solved via sketching algorithms of~\cite{flajolet1985probabilistic} and its variants e.g.,~\cite{flajolet2007hyperloglog}.
Some recent work~\cite{Choi20,Smith0T20} extends these sketching techniques for counting distinct element in a DP streaming setting.
However, it is not clear how to extend these techniques to the continual release setting.
Continual release of counts of distinct elements is studied by \cite{Bolot13}. However, \cite{Bolot13} is not in the low space streaming setting.

\xhdr{Estimation of frequencies and $\ell_2$ frequency moments.}
The goal of $\ell_2$ frequency moment estimation is to estimate the sum of square of frequencies of elements.
We present a DP streaming continual release CountSketch~\cite{charikar2002finding} algorithm and use it for estimating $\ell_2$ frequency moments and the frequency of each element.

\begin{theorem}[Frequency and $\ell_2$ frequency moments, informal version of Theorem~\ref{thm:count_sketch}]
There is an $\varepsilon$-DP algorithm in the streaming continual release model such that with probability at least $0.9$, it always outputs for every timestamp $t\in [T]$:
\begin{enumerate}
    \item $\hat{f}_a$ for every $a\in\mathcal{U}$ such that $|f_a-\hat{f}_a|\leq \eta \|\mathcal{S}\|_2+\tilde{O}_{\varepsilon,\eta}\left(\log^{3.5}(Tn)\right)$, where $\mathcal{S}$ denotes the stream up to timestamp $t$ and $f_a$ denotes the frequency of $a$ in $\mathcal{S}$,
    \item $\hat{F}_2$ such that $|\hat{F}_2-\|\mathcal{S}\|_2^2|\leq \eta \|\mathcal{S}\|_2^2+\tilde{O}_{\varepsilon,\eta}\left(\log^7(Tn)\right)$
\end{enumerate}
The algorithm uses $O\left(\frac{\log(Tn)}{\eta^2}\cdot \log(T)\right)$ space.
\end{theorem}
Although DP $\ell_2$ frequency moment was studied by a line of work (see e.g., \cite{BlockiBDS12,Sheffet17,BuGKLST21}), none of them considers the streaming continual release setting, and it is not clear how to extend previous techniques to the the continual release setting.

\xhdr{$\ell_p$ Heavy hitters.}
In the $\ell_p$ heavy hitters problem, we are given a parameter $k$, and the goal is to find elements whose frequency to the $p$-th power is at least $1/k$ fraction of the $\ell_p$ frequency moment.
By extending our DP streaming continual release CountSketch algorithm, we obtain a DP streaming continual release $\ell_p$ heavy hitters algorithm.
\begin{theorem}[$\ell_p$ Heavy hitters for all $p\in[0,\infty)$, informal version of Theorem~\ref{thm:lp_heavyhitters}]
%Let $\phi = \max(1,|\mathcal{U}|^{1-2/p})$.
There is an $\varepsilon$-DP algorithm in the streaming continual release model such that with probability at least $0.9$, it always outputs a set $H\subseteq \mathcal{U}$ and a function $\hat{f}: H\rightarrow \mathbb{R}$ for every timestamp $t\in [T]$ satisfying
\begin{enumerate}
    \item $\forall a\in H$, $\hat{f}(a)\in (1\pm \eta)\cdot f_a$ where $f_a$ is the frequency of $a$ in the stream $\mathcal{S}$ up to timestamp $t$,
    \item $\forall a\in \mathcal{U},$ if $f_a\geq \frac{1}{\varepsilon\eta}\cdot \poly\left(\log\left(\frac{T\cdot k\cdot n}{\eta}\right)\right)$ and $f_a^p\geq \|\mathcal{S}\|_p^p/k$ then $a\in H$,
    \item The size of $H$ is at most $O\left(\log(Tn)\cdot2^p\cdot k \right)$.
\end{enumerate}
The algorithm uses $\max(1,n^{1-2/p})\cdot\frac{k^3}{\eta^2}\cdot \poly\left(\log\left(T\cdot k\cdot n\right)\right)$ space.
\end{theorem}
To the best of our knowledge, though DP streaming continual release $\ell_1$ heavy hitters problem is studied by~\cite{chan2012differentially}, $\ell_p$ (for $p\not=1$) heavy hitters problem has not been studied in the DP streaming continual release setting before.
Note that $\Omega(n^{1-2/p})$ for $p>2$ is a lower bound of space needed for $\ell_p$ heavy hitters even in the non-private setting~\cite{saks2002space,bar2004information}.

\subsubsection{Differentially Private Sliding Window Continual Release Algorithms}
Smooth histogram~\cite{BravermanO07} is a general algorithmic framework which can convert a relative-approximate streaming algorithm into a relative-approximate sliding window algorithm if the objective function that we want to compute has some nice properties.

We generalize the smooth histogram to make it support converting an approximate streaming algorithm with both relative error and additive error into an approximate sliding window algorithm with both relative error and additive error if the objective function has good properties.
In addition, we show that if the streaming algorithm is DP in the continual release setting, then the converted sliding window algorithm is also DP in the continual release setting.

By applying our generalized smooth histogram approach and paying a $\poly\left(\frac{\log T}{\eta}\right)$ more factor than our DP streaming continual release algorithms in both additive error and space usage, we show DP sliding window continual release algorithms for 
\begin{enumerate}
\item $\ell_p$ Moment estimation (see Corollary~\ref{cor:sliding_window_lp_moment}),
\item Summing (see Corollary~\ref{cor:sliding_window_summing}),
\item Counting distinct elements (see Corollary~\ref{cor:sliding_window_distinct_elements}),
\item $\ell_2$ Moment estimation (see Corollary~\ref{cor:sliding_window_l2}).
\end{enumerate}

\vspace{-0.1in}
\subsection{Our Techniques}
\vspace{-0.1in}
In this section, we briefly discuss the high level ideas of our algorithms.
We present a set of techniques to reduce almost all problems that we considered in the DP streaming continual release setting to the summing problem in the DP streaming continual release setting.

\vspace{-0.1in}
\xhdr{Summing with better additive error via grouping.}
To illustrate the intuition of using grouping for differentially private streaming continual release algorithms, we start with the following simple problem: given a stream of numbers $c_1,c_2,\cdots,c_T$ where each $c_i$ is at least $10\cdot\ln(10\cdot T)/\varepsilon$, the goal is to output a $(1\pm 0.1)$-approximation to $\sum_{j=1}^t c_j$ for every prefix $t$ with probability at least $0.9$, and we want the set of all outputs to be $\varepsilon$-DP, i.e., the continual released results to be $\varepsilon$-DP.
A simple way to solve the above problem is that we release a stream of noisy numbers $\hat{c}_1,\hat{c}_2,\cdots,\hat{c}_T$ where $\forall t\in[T],\hat{c}_t=c_t+\lap(1/\varepsilon)$, and we report $\sum_{j=1}^t \hat{c}_j$ for every prefix $t\in [T]$.
It is easy to see that $(\hat{c}_1,\hat{c}_2,\cdots,\hat{c}_T)$ is $\varepsilon$-DP.
Since the reported approximate prefix sums only depend on $(\hat{c}_1,\hat{c}_2,\cdots,\hat{c}_T)$, the continual released results are $\varepsilon$-DP.
Furthermore, with probability at least $0.9$, $\forall t\in[T], |c_t-\hat{c}_t|\leq \ln(10\cdot T)/\varepsilon$.
Since $c_t$ is at least $10\ln(10\cdot N)/\varepsilon$, we have $0.9c_t\leq \hat{c}_t\leq 1.1c_t$ which implies that every reported approximate prefix sum is a $(1\pm 0.1)$-approximation.

To generalize the above idea, we propose a grouping approach in to group the consecutive numbers in the stream in a differentially private way such that the total count of each group is large enough.
To implement grouping, we need to apply the sparse vector technique (see e.g.,~\cite{dwork2014algorithmic}) iteratively. 
The similar idea also appeared in \cite{DworkNRR15} which shows a better additive error guarantee for the summing problem than \cite{DworkNPR10} when the stream is sparse.
In contrast, our additive error guarantee is always better than \cite{DworkNPR10} and \cite{DworkNRR15} while we allow an additional $(1+\varepsilon)$ relative approximation.
%\jm{This idea is similar to the partition algorithm in \cite{DworkNRR15}, which gives a better guarantee for summing when the stream is sparse.}

\vspace{-0.1in}
\xhdr{Counting distinct elements.} 
We explain how to reduce counting distinct elements problem to the summing problem.
Suppose the element universe is small, we are able to track the set of elements that already appeared during the stream.
Then, we can create a binary stream of $\{0,1\}$ where $1$ denotes that we see a new element and $0$ denotes that the input element already appeared or it is empty.
Therefore, the sum of the binary stream at timestamp $t$ is exactly the number of distinct elements.
Furthermore, if we change an element in the input stream from $a $ to $b$, there are only constant number of positions of the binary stream will flip: consider the change $a\rightarrow \perp \rightarrow b$.
If $a$ is not its first appearance in the input stream, changing $a$ to $\perp$ does not cause any change in the binary stream.
If $a$ is its first appearance in the input stream, changing $a$ to $\perp$ will make the corresponding $1$ in the binary stream be $0$ and make the $0$ corresponding to the original second appearance of $a$ in the input stream to be $1$.
Thus, it will affect at most $2$ entries of the binary stream.
Similarly, changing $\perp$ to $b$ will cause the change of at most $2$ entries of the binary stream.
Thus, the binary stream has low sensitivity which implies that a DP streaming continual release summing algorithm gives a good approximation to the number of distinct elements with a small additive error.
Next, we discuss how to handle the large universe.
For large universe, we can try different sampling rate $1/2,1/4,1/8\cdots,1/T$.
There should be a sampling rate such that (1) if we hash the sampled elements into hashing buckets, there is no collision with a good probability, (2) the number of samples is much larger than the additive error caused by the summing subroutine so we can have a good relative approximation of the number of distinct sampled elements.
Then we can use the number of distinct sampled elements to estimate the number of distinct elements in the input stream.

\vspace{-0.1in}
\xhdr{CountSketch and $\ell_p$ heavy hitters.}
Let $h:\mathcal{U}\rightarrow[k]$ be a hash function which uniformly hash elements into $k$ hash buckets.
Let $g:\mathcal{U}\rightarrow\{-1,1\}$ randomly map each element to $-1$ or $1$ with equal probability.
The CountSketch is a tuple of $k$ numbers $(z_1,z_2,\cdots,z_k)$ where $z_i$ is the sum of weighted frequencies of elements hashed to the bucket $i$, and the weight of the frequency of element $a$ is $g(a)$.
 Changing $a$ to $b$ in the input stream will change at most $2$ buckets: the bucket $i$ contains $a$ and the bucket $j$ contains $b$.
 Since $|g(a)|\leq 1$, $z_i$ and $z_j$ will be changed by at most $1$. 
 We can use DP streaming continual release summing algorithm of~\cite{DworkNPR10,ChanSS11} to estimate $(z_1,z_2,\cdots,z_k)$ such that each estimation $\hat{z}_i$ of $z_i$ only has $\poly(\log T)$ additive error, and $(\hat{z}_1,\hat{z}_2,\cdots,\hat{z}_k)$ is DP under the streaming continual release model.
 Suppose the $\ell_2$ frequency moment is much larger than $\poly(\log T)$, then the additive error becomes the relative error, and we can use $\hat{z}_1,\hat{z}_2,\cdots,\hat{z}_k$ to obtain a good relative approximation of the $\ell_2$ frequency moment.
 Similarly, if an element has frequency much larger than $\poly(\log(T))$, then $\poly(\log(T))$ becomes small relative error of the frequency and we are able to check whether it is an $\ell_2$ heavy hitter by the standard analysis of CountSketch.
 Thus, we can use this DP streaming continual release CountSketch to estimate $\ell_2$ frequency moment with $(1+\eta)$-relative error and $\poly(\log T)$-additive error, and we can use such CountSketch to find all elements which are at least $\poly(\log T)$ and are $\ell_2$ heavy hitters.
 
 Note that for $p\leq 2$, if $a$ has the largest frequency and it is an $(1/k)$-$\ell_p$ heavy hitter, then $a$ must be an $(1/k)$-$\ell_2$ heavy hitter.
 For $p>2$, if $a$ is an $1/k$-$\ell_p$ heavy hitter, than $a$ must be an $1/(kn^{1-2/p})$-$\ell_2$ heavy hitter.
 Therefore, by some hashing technique, we can use $\ell_2$ heavy hitters algorithm to construct $\ell_p$ heavy hitters algorithm.
 But since $\ell_2$ heavy hitters can only report the elements with frequency larger than $\poly(\log T)$, the obtained $\ell_p$ heavy hitters algorithm can only report the elements with frequency larger than $\poly(\log T)$ as well.
 
\vspace{-0.1in}
 \xhdr{$\ell_p$ Frequency moment estimation.}
 In high level we want to simulate the level set estimation idea of~\cite{indyk2005optimal} in the DP streaming continual release setting.
 In particular, let $\alpha=1+\eta$, let $f_a$ denote the frequency of $a$ and let $G_i=\{a\mid f_a\in(\alpha^i,\alpha^{i+1}]\}$. 
 Then $\sum_{i} |G_i|\cdot (\alpha^i)^p$ is a good approximation to the $\ell_p$ frequency moment.
 We say $G_i$ is contributing, if $|G_i|\cdot (\alpha^i)^p$ is at least $\Omega_{\alpha}(1/\log(T))$ fraction of the $\ell_p$ moment.
 Since non-contributing elements only contributes a small total amount to the $\ell_p$ frequency moment, it is easy to see that $\sum_{\text{contributing } G_i} |G_i|\cdot (\alpha^i)^p$ is still a good approximation to the $\ell_p$ frequency moment.
 Thus, we only need to estimate the size of each contributing $G_i$.
 Due to the definition of contributing, it is easy to see that if $G_i$ is contributing, either $\alpha^i$ is large or $|G_i|$ is large.
 In fact, as observed by~\cite{indyk2005optimal}, for each contributing level set $G_i$, there must be a proper sampling probability such that after sampling, there are at least $\poly(\log T)$ elements from $G_i$ sampled and all of the sampled elements from $G_i$ are at least $1/\poly(\log T)$-$\ell_p$ heavy hitters among the set of all sampled elements from the universe $\mathcal{U}$.
 Ideally, we can try different sampling rate $1,1/2,1/4,1/8,\cdots,1/T$ and use our $\ell_p$ heavy hitters algorithm to report the heavy hitters and estimate $|G_i|$ for each $i$.
 However, for $i$ with $\alpha^i \ll \poly(\log T)$, our DP streaming continual release $\ell_p$ heavy hitters algorithm does not report any element from $G_i$.
 We must find another way to estimate $|G_i|$ instead of using heavy hitters.

Similar to counting distinct elements, let us start with the case that the universe size is small so we can track the sets $A_1,A_2,\cdots,A_k$ for some $k=\poly(\log T)$, where $A_i$ is the set of all elements whose frequency is exactly $i$.
We can construct streams $\mathcal{S}_1,\mathcal{S}_2,\cdots,\mathcal{S}_k$ with numbers in $\{-1,0,1\}$.
During the stream, when we see an input element $a$, and if $a$ is in the set $A_i$, then we move $a$ to the set $A_{i+1}$ due to the increase of the frequency of $a$.
At the same time, we append $-1$ to the stream $\mathcal{S}_i$, append $1$ to the stream $\mathcal{S}_{i+1}$ and append $0$ to $\mathcal{S}_j$ for $j\not=i,i+1$.
It is easy to check that the sum of $\mathcal{S}_l$ is always the same as $|A_l|$, i.e., the number of elements which have frequency exactly $l$.
Furthermore, similar to the analysis for counting distinct elements, if we change an element in the input stream, each $\mathcal{S}_l$ might be affected by at most $4$ entries.
Thus, the total sensitivity of $(\mathcal{S}_1,\mathcal{S}_2,\cdots,\mathcal{S}_k)$ is at most $O(k)$.
Therefore, we can use the DP continual release summing to estimate the sum of each $\mathcal{S}_l$ with additive error $\poly(\log T)$.
Thus, we can estimate $|G_i|$ for each $i$ with $\alpha^i\ll \poly(\log T)$ with additive error $\poly(\log T)$ and will only introduce at most $\poly(\log T)$ additive error in approximating the $\ell_p$ frequency moment.

Now let us go back to the case that the size of the universe is large.
In this case, we can use the similar hashing and subsampling technique discussed for counting distinct elements to estimate $|A_l|$ for each $l\in [k]$.
 
\subsection{Related Work}

\label{sec:related}

\cite{DworkNPR10} and \cite{ChanSS11} initiated the study of differential privacy in the continual release model, and proposed the binary tree mechanism for computing summations. \cite{Bolot13} and \cite{PerrierAK19} generalized their results to decayed summations, counting distinct elements without space constraints and summations with real-valued data. \cite{Song18,FichtenbergerHO21} studied graph problems under the differentially private continual release model. \cite{Jain12,AdamA13,Agarwal17a} studied differentially private online learning. \cite{Jain21} gave the first polynomial separation in terms of error between the continual release model and the batch model under differential privacy. \cite{Upadhyay19a} studied heavy hitters in the differentially private sliding window model.

Differentially private frequency moment estimation for $p = 0,1,2$ (without continual releases) has been well-studied \cite{mir2011pan, dwork2010pan, BlockiBDS12,Sheffet17,Choi20,Smith0T20,BuGKLST21}. \cite{WangPS22} studied frequency moment estimation (without continual releases) for $p\in (0,1]$ with low space complexity.
Recent concurrent independent work~\cite{blocki2022make} studies $p\in[0,\infty)$ with low space complexity but not in continual release setting as well.
The differentially private $\ell_1$ heavy hitters problem is studied by~\cite{mir2011pan, dwork2010pan} in the low space streaming setting but not in the continual release setting.
\cite{chan2012differentially} studied differentially private $\ell_1$ heavy hitters problem in the low space continual release streaming setting. 
But it is not clear how to extend their techniques to $l_p$ case for  $p\not =1$.

$\ell_p$ Frequency moment estimation and $\ell_p$ heavy hitters are heavily studied in the non-private streaming literature.
For $\ell_p$ frequency moment estimation, the problem can be solved by e.g.~\cite{flajolet1985probabilistic,flajolet2007hyperloglog,durand2003loglog} for $p=0$, ~\cite{alon1996space,charikar2002finding,thorup2004tabulation} for $p=2$, ~\cite{kane2010exact,indyk2006stable,li2008estimators,kane2011fast} for $p\in(0,2)$ and \cite{indyk2005optimal, andoni2011streaming,andoni2017high} for $p>2$.
For $\ell_p$ heavy hitters, the problem can be solved by e.g.,~\cite{cormode2005improved,misra1982finding} for $p=1$, ~\cite{charikar2002finding} for $p=2$, ~\cite{jowhari2011tight} for $p\in (0,2)$, and \cite{indyk2005optimal,andoni2011streaming} for $p>2$.

\section{Preliminaries}\label{sec:preli}
\subsection{Notation}
In this paper, for $n\geq 1$, we use $[n]$ to denote the set $\{1,2,\cdots,n\}$.
If there is no ambiguity, for $i\leq j\in \mathbb{Z}$, we sometimes use $[i,j]$ to denote the set of integers $\{i,i+1,\cdots,j\}$ instead of the set of real numbers $\{a\in\mathbb{R}\mid i\leq a \leq j\}$.
We use $f_a(a_1,a_2,\cdots,a_k)$ to denote the frequency of $a$ in the sequence $(a_1,a_2,\cdots,a_k)$, i.e., $f_a(a_1,a_2,\cdots,a_k)=|\{i\in[k]\mid a_i = a\}|$.
If the sequence $(a_1,a_2,\cdots,a_k)$ is clear in the context, we use $f_a$ to denote the frequency of $a$ for short.
We use $\mathbf{1}_{\mathcal{E}}$ to denote a indicator, i.e., $\mathbf{1}_{\mathcal{E}}=1$ if condition $\mathcal{E}$ holds and $\mathbf{1}_{\mathcal{E}}=0$ if condition $\mathcal{E}$ does not hold.

For $\alpha\geq 1,\gamma\geq 0$, if $\frac{1}{\alpha}\cdot x-\gamma\leq y \leq \alpha\cdot x + \gamma$, then $y$ is an $(\alpha,\gamma)$-approximation to $x$.
If $y$ is a $(1,\gamma)$-approximation to $x$, we say $y$ is an approximation to $x$ with additive error $\gamma$.
If $y$ is an $(\alpha,0)$-approximation to $x$, we say $y$ is an $\alpha$-approximation to $x$.
We use $a\pm b$ to denote the real number interval $[a-|b|,a+|b|]$.
For a set $S$ of real numbers, we use $S\pm b$ to denote the set $\bigcup_{a\in S} a\pm b$, and use $S\cdot c$ to denote the set $\bigcup_{a\in S}a \cdot c$.
We use $\lap(b)$ to denote the Laplace distribution with scale $b$, i.e., $\lap(b)$ has density function given by  $\frac{1}{2b}\exp(|x|/b)$

\subsection{Functions to Compute}
We study several fundamental functions in the streaming literature.
When the inputs are integers, we consider the summing problem over a (sub-)stream $(a_i,\cdots,a_j)$:
\begin{itemize}
    \item Sum of numbers: $\sm(a_i,\cdots,a_j) := \sum_{k=i}^j a_k$
\end{itemize}

When the inputs are from $\mathcal{U}\cup \{\perp\}$, we consider the functions $g(a_i,\cdots,a_j)$ that are based on the frequencies of the elements in a (sub-)stream $(a_i,\cdots,a_j)$.
\begin{itemize}
    \item Count of non-empty elements: $\|(a_i,...,a_j)\|_1 := \sum_{a\in \mathcal{U}} f_a(a_i,...,a_j)$.
 
%    \item Check if the number of non-empty rounds is at least some threshold $T$ \jm{maybe we don't need to mention here in this paper. The interesting part about threshold is that we can obtain an additive error result, even if the approximation to the underlying function is multiplicative.}: $\thr_T(x_i,...,x_j) = 1_{ \cnt(x_i,...,x_j) \geq T}$.

    \item The number of distinct elements: $\|(a_i,...,a_j)\|_0 := \sum_{a\in \mathcal{U} } \mathbf{1}_{f_a(a_i,...,a_j) > 0}$.
    \item $\ell_p$-Frequency moment:  $\|(a_i,...,a_j)\|_p^p := \sum_{a\in \mathcal{U}} f_a(a_i,...,a_j)^p$.
    \item $\ell_p$-Heavy hitters: $(1/k)$-$\ell_p\hh(a_i,...,a_j) := \{a\in\mathcal{U}\mid f_a(a_i,\cdots,a_j)^p\geq \|(a_i,\cdots,a_j)\|_p^p/k\}$.

%\jm{
%\begin{remark}
%\end{remark}
%}    
\end{itemize}
Note that $\|(a_i,...,a_j)\|_1$ is a special case of $\sm(a_i,\cdots,a_j)$ with binary inputs.

\subsection{Differential Privacy}

\xhdr{Neighboring streams:} Consider two streams $\mathcal{S}=(a_1,a_2,\cdots,a_T)$ and $\mathcal{S}'=(a_1',a_2',\cdots,a_T')$. 
If there is at most one timestamp $t\in[T]$ such that (1). $|a_t-a_t'|\leq 1$ (only required when the inputs are treated as integers) (2). $\forall i\not= t,a_i=a_i'$, then we say $\mathcal{S}$ and $\mathcal{S}'$ are neighboring streams.

\begin{definition} [Differential Privacy]~\label{def:dp}
We say algorithm $A$ is $\varepsilon$-DP, if for any two neighboring streams $\mathcal{S}, \mathcal{S}'$, and any output set $\mathcal{O}$,
\[
\Pr[A(\mathcal{S}) \in \mathcal{O}] \leq e^{\varepsilon} \cdot \Pr[A(\mathcal{S}') \in \mathcal{O}].
\]
\end{definition}
%For pure privacy, we omit the $\delta$, i.e. we use $\varepsilon$-DP to represent $(\varepsilon, 0)$-DP.
Note that in the continual release model, the output $A(\mathcal{S})$ mentioned in Definition~\ref{def:dp} is the entire output history of the algorithm $A$ over stream $\mathcal{S}$ at every timestamp.

\begin{definition}[Distance between streams]
Consider two streams $\mathcal{S}$ and $\mathcal{S}'$.
If $d$ is the minimum number such that there exists a sequence of streams $\mathcal{S}_0,\mathcal{S}_1,\cdots,\mathcal{S}_d$ where $\mathcal{S}_0=\mathcal{S},\mathcal{S}_d=\mathcal{S}'$ and $\forall i\in[d]$, $\mathcal{S}_i$ and $\mathcal{S}_{i-1}$ are neighboring streams, then the distance between $\mathcal{S}$ and $\mathcal{S}'$ is $\dis(\mathcal{S},\mathcal{S}')=d$.
\end{definition}

\begin{definition}[Sensitivity of a stream mapping]\label{def:sensitivity}
Let $\mathcal{F}$ be a mapping which maps a given input stream $\mathcal{S}$ to a tuple of streams $(\mathcal{F}_1(\mathcal{S}),\mathcal{F}_2(\mathcal{S}),\cdots,\mathcal{F}_k(\mathcal{S}))$.
The sensitivity of $\mathcal{F}$ is the minimum value $s$ such that for any two neighboring streams $\mathcal{S}$ and $\mathcal{S}'$, $\sum_{i\in[k]} \dis(\mathcal{F}_i(S),\mathcal{F}_i(S'))\leq s$.
\end{definition}

\begin{theorem}[Composition~\cite{dwork2014algorithmic}]\label{thm:composition}
Let $\mathcal{F}$ be a mapping which maps a given input stream $\mathcal{S}$ to a tuple of streams $(\mathcal{F}_1(\mathcal{S}),\mathcal{F}_2(\mathcal{S}),\cdots,\mathcal{F}_k(\mathcal{S}))$.
Let $A_1,A_2,\cdots,A_k$ be $k$ $\varepsilon$-DP algorithms.
Let $A$ be an algorithm such that $A(\mathcal{S})=M(A_1(\mathcal{F}_1(\mathcal{S})),A_2(\mathcal{F}_2(\mathcal{S})),\cdots,A_k(\mathcal{F}_k(\mathcal{S})))$ for some function $M(\cdot)$.
Then the algorithm $A$ is $(s\varepsilon)$-DP.
\end{theorem}

%We say $\mathcal{S}$ and $\mathcal{S}'$ are \emph{neighboring} streams if they differ at exactly one round.
\xhdr{Example usage of Theorem~\ref{thm:composition}:} 
Some example usage of Theorem~\ref{thm:composition} in our paper are presented as the following:
\begin{itemize}
\item \textbf{Composition of multiple algorithms over the input stream:} Suppose each of $A_1,A_2,\cdots,A_k$ is $\varepsilon$-DP, then for any function $M(\cdot)$, $M(A_1(\mathcal{S}),A_2(\mathcal{S}),\cdots,A_k(\mathcal{S}))$ is $(k\varepsilon)$-DP.
\item \textbf{Composition of algorithms on disjoint sub-streams:} For an input stream $\mathcal{S}=(a_1,a_2,\cdots,a_T)$, we partition $\mathcal{S}$ into $\mathcal{S}_1=(a_{1,1},a_{1,2},\cdots,a_{1,T}),\mathcal{S}_2=(a_{2,1},a_{2,2},\cdots,a_{2,T}),\cdots,\mathcal{S}_k=(a_{k,1},a_{k,2},\cdots,a_{k,T})$, i.e., $\forall t\in [T]$, there is only one $i\in[k]$ such that $a_{i,t}=a_t$, and $\forall i'\not=i,a_{i,t}=\perp\text{ (or $0$)}$.
It is clear that such partitioning has sensitivity $1$.
Thus, if each of $A_1,A_2,\cdots,A_k$ is an $\varepsilon$-DP algorithms, then for any function $M(\cdot)$, $M(A_1(\mathcal{S}_1),A_2(\mathcal{S}_2),\cdots,A_k(\mathcal{S}_k))$ is $\varepsilon$-DP.
\end{itemize}

\subsection{Streaming Continual Release Summing and Counting}
For summing problem in the streaming continual release model, the binary tree mechanism was proposed in \cite{DworkNPR10,ChanSS11}. It gets poly-logarithmic additive error and uses logarithmic space. 
Furthermore, it can handle negative numbers in the stream.

\begin{theorem}[\cite{DworkNPR10,ChanSS11}]\label{thm:dp_summing}
Let $\varepsilon\geq 0,\xi\in(0,0.5)$, there is an $\varepsilon$-DP algorithm for summing in the streaming continual release model. 
With probability $1-\xi$, the additive error of the output for every timestamp $t\in [T]$ is always at most $O\left(\frac{1}{\varepsilon} \log^{2.5}(T)\log\left(\frac{1}{\xi}\right)\right)$. 
The algorithm uses $O(\log(T))$ space.
\end{theorem}
%Note that above theorem directly implies a streaming continual release algorithm for count of non-empty elements by a simple reduction: for a stream $\mathcal{S}=(a_1,a_2,\cdots,a_T)$ of elements in $\mathcal{U}\cup\{\perp\}$, construct a stream $\mathcal{C}=(c_1,c_2,\cdots,c_T)$ where $\forall t\in [T]$, $c_t=1$ if $a_t\not=\perp$ and $c_t=0$ otherwise.
%Then $\forall t\in [T]$, the sum of $c_1,c_2,\cdots,c_t$ is exactly the count of non-empty elements in $a_1,a_2,\cdots,a_t$.

\subsection{Probability Tools}
\begin{lemma}[\cite{bellare1994randomness}]\label{lem:concentration}
Let $\lambda\geq 4$ be an even integer. 
Let $X$ be the sum of $n$ $\lambda$-wise independent random variables which take values in $[0, 1]$.
Let $\mu = E[X]$ and $A > 0$. 
Then we have
\begin{align*}
\Pr\left[\left|X-\mu\right| > A\right] \leq 8 \cdot \left(\frac{\lambda\mu + \lambda^2}{A^2}\right)^{\lambda/2}
\end{align*}
\end{lemma}

\begin{lemma}[Median trick to boost success probability]\label{lem:median_trick}
Suppose $X$ is a random estimator of value $v$ such that with probability at least $2/3$, $X$ is an $(\alpha,\gamma)$-approximation to $v$.
Then, for $\xi\in(0,0.5)$, if we draw $k=\lceil50\log(1/\xi)\rceil$ independent copies $X_1,X_2,\cdots,X_k$ of $X$, with probability at least $1-\xi$, the median of $X_1,X_2,\cdots,X_k$ is an $(\alpha,\gamma)$-approximation to $v$.
\end{lemma}
\begin{proof}
We say $X_i$ is good if $X_i$ is an $(\alpha,\gamma)$-approximation to $v$.
If the median is not good, then $\sum_{i\in[k]}\mathbf{1}_{X_i\text{ is good}}< \frac{k}{2}$.
By Chernoff bound, $\Pr[\text{median is not good}]\leq \Pr[\sum_{i\in[k]}\mathbf{1}_{X_i\text{ is good}}< \frac{k}{2}]\leq \Pr[E[\sum_{i\in[k]}\mathbf{1}_{X_i\text{ is good}}] - \frac{k}{2} > \frac{k}{6}]\leq e^{-k/48}\leq \xi.$
\end{proof}

\section{Continual Released Summing with Better Additive Error}\label{sec:count}
In this section, we show that if we allow a relative approximation and the input stream only contains non-negative numbers, then we can have a continual released summing algorithm with additive error better than Theorem~\ref{thm:dp_summing}.

%\subsection{Differentially Private Grouping of the Stream}
%\jm{It seems you only describe a streaming algorithm, do you still need a smooth histogram to make it work in the sliding window setting?}
%\pz{It is a reduction of the stream itself. We can apply any non-private sliding window to the reduced stream not necessarily to be a smooth histogram.}

\begin{algorithm}[H]
	\KwIn{A stream of non-negative numbers $c_1,c_2,\cdots,c_T$ DP parameter $\varepsilon>0$, approximation parameter $\eta\in(0,0.5)$, and failure probability $\xi\in(0,1)$.}
	\KwOut{A stream of groups with grouped noisy counts 
	$(\hat{c}_1,\hat{c}_2,\cdots,\hat{c}_T)$
	%$(G_1,\hat{c}_1),(G_2,\hat{c}_2),\cdots,(G_m,\hat{c}_m)$.
	}
	Let $\varepsilon_0\gets \varepsilon / 2$.~\\
	Initialize a group index $i\gets1$, current group $G_1\gets \emptyset$, and threshold $\tau_1\gets \left(\frac{1}{\eta}+1\right)\cdot \frac{7}{\varepsilon_0}\cdot\ln\left(3\cdot T/\gamma\right) + \lap(2/\varepsilon_0)$. \hfill{//$G_i$ is used for analysis only.}\\
	\For{$t = 1$ to $T$} {
	    $G_i\gets G_i\cup \{t\}$.~\\
	    Let $\nu_t\gets \lap(4/\varepsilon_0)$.~\\
	    \uIf {$\nu_t+\sum_{j\in G_i} c_j\geq \tau_i$} {
	        $\hat{c}_t\gets \lap(1/\varepsilon_0)+\sum_{j\in G_i}c_j$. \\
	        %Append $(G_i,\hat{c}_i)$ to the output stream.~\\
	        $i\gets i + 1$.~\\
	        $\tau_i\gets \left(\frac{1}{\eta}+1\right)\cdot \frac{7}{\varepsilon_0}\cdot\ln\left(3\cdot T/\gamma\right) + \lap(2/\varepsilon_0)$.~\\
	        $G_i\gets \emptyset$.
	    }
	    \Else {
	        $\hat{c}_t\gets 0$.
	    }
  	}
  	%Let $\hat{c}_i\gets \lap(1/\varepsilon_0)+\sum_{j\in G_i}c_j$.
  	%Append $(G_i,\hat{c}_i)$ to the output stream.~\\
    \caption{Grouping Stream of Counts.}
    \label{alg:DPgroup}
\end{algorithm}

\begin{lemma}\label{lem:dp_output_stream}
The output stream $\hat{c}_1,\hat{c}_2,\cdots,\hat{c}_T$ of Algorithm~\ref{alg:DPgroup} is $\varepsilon$-DP.
\end{lemma}
The proof idea is to iteratively apply sparse vector technique.
We put the proof into Appendix~\ref{sec:proof_of_grouping_dp}.

\begin{lemma}\label{lem:acc_output_stream}
Let $\hat{c}_1,\hat{c}_2,\cdots,\hat{c}_T$ be the output stream of Algorithm~\ref{alg:DPgroup}.
Then with probability at least $1-\xi$, $\forall l,r$ satisfying $1\leq l\leq r\leq T$,
\begin{align*}
(1-\eta)    \sum_{j=l}^r c_j - \left(\frac{1}{\eta}+4\right)\cdot \frac{7}{\varepsilon_0}\cdot \ln\left(3\cdot T/\xi\right)\leq \sum_{j=l}^{r} \hat{c}_{j} \leq (1+\eta) \sum_{j=l}^r c_j + \left(\frac{1}{\eta}+4\right)\cdot \frac{7}{\varepsilon_0}\cdot \ln\left(3\cdot T/\xi\right).
\end{align*}
\end{lemma}
We put the proof of Lemma~\ref{lem:acc_output_stream} into Appendix~\ref{sec:proof_of_acc_grouping}

\begin{theorem}[Summing of a non-negative stream]\label{thm:summing_non_negative}
Let $\varepsilon\geq 0,\xi\in(0,0.5)$, there is an $\varepsilon$-DP algorithm for summing in the streaming continual release model.
If the input numbers are guaranteed to be non-negative, with probability at least $1-\xi$, the output is always a $\left(1+\eta, O\left(\frac{\log(T/\xi)}{\varepsilon\eta}\right)\right)$-approximation to the summing problem at any timestamp $t\in [T]$.
The algorithm uses space $O(1)$.
\end{theorem}
\begin{proof}
According to Lemma~\ref{lem:dp_output_stream}, the output stream of Algorithm~\ref{alg:DPgroup} is $\varepsilon$-DP, thus, we only need to solve non-private summing over $(\hat{c}_1,\hat{c}_2,\cdots,\hat{c}_T)$.
The approximation guarantee is given by Lemma~\ref{lem:acc_output_stream}.
Note that we do not need to store $G_i$, we only need to maintain the sum of numbers in $G_i$ at any timestamp.
Thus, the total space needed is $O(1)$.
\end{proof}

\section{Continual Released Number of Distinct Elements}
In this section, we show how to use $\varepsilon$-DP streaming continual release summing to solve $\varepsilon$-DP streaming continual release number of distinct elements.
In Section~\ref{sec:dist_small_universe}, we show how to estimate the number of distinct elements if the universe is small.
In Section~\ref{sec:dist_large_universe}, we reduce the number of distinct elements of a large universe to the number of distinct elements of a small universe via subsampling.

\subsection{Number of Distinct Elements for Small Universe}\label{sec:dist_small_universe}

\begin{algorithm}[H]\label{alg:distinct_small_universe}
\small
	\KwIn{A stream $\mathcal{S}$ of elements $a_1,a_2,\cdots,a_T\in \mathcal{U}\cup\{\perp\}$ with guarantee that $|\mathcal{U}|\leq m$.}
	\Parameter{Relative approximation factor $\alpha\geq 1$ and additive approximation factor $\gamma\geq 0$ depending on the streaming continual release summing algorithm.\\ \hfill{//See Theorem~\ref{thm:dp_summing}}. }
	\KwOut{Estimation of the number of distinct elements at every timestamp $t$.}
	Initialize an empty stream $\mathcal{C}$.\\
	Let $S\gets \emptyset$.\\
	\For {each $a_t$ in the stream $\mathcal{S}$}{
	    \If{$a_t\not\in S$ and $a_t\not=\perp$}{
	        $S\gets S\cup \{a_t\}$.\\
	        Append $1$ to the end of the stream $\mathcal{C}$.\\
	    }
	    \Else{
	        Append $0$ to the end of the stream $\mathcal{C}$.\\
	    }
	    Output an $(\alpha,\gamma)$-approximation to the total counts of $\mathcal{C}$.
	}
    \caption{Number of Distinct Elements for Small Universe}
\end{algorithm}

\begin{lemma}\label{lem:approx_distinct_small_universe}
At the end of any time $t\in[T]$, the output of Algorithm~\ref{alg:distinct_small_universe} is an $(\alpha,\gamma)$-approximation to the number of distinct elements.
%The algorithm uses $O(m)$ space.
\end{lemma}
\begin{proof}
Since we append $1$ to the stream $\mathcal{C}$ if and only if we see a new non-empty element, the total counts in $\mathcal{C}$ is always equal to the number of distinct elements at the end of any time $t\in[T]$.
Thus, an $(\alpha,\gamma)$-approximation to the total counts of $\mathcal{C}$ is an $(\alpha,\gamma)$-approximation to the number of distinct elements.

%Since each element is in $\mathcal{U}$, we have $|S|\leq |\mathcal{U}|\leq m$ at any time. 
%Thus, the space usage is at most $O(m)$.
\end{proof}

\begin{lemma}\label{lem:DP_distinct_small_universe}
If the algorithm to continually release the approximate total counts of $\mathcal{C}$ in Algorithm~\ref{alg:distinct_small_universe} is $\varepsilon$-DP, Algorithm~\ref{alg:distinct_small_universe} is $5\varepsilon$-DP in the continual release model.
\end{lemma}
\begin{proof}
Consider two neighboring stream $\mathcal{S}=(a_1,a_2,\cdots,a_T)$ and $\mathcal{S}'=(a_1',a_2',\cdots,a_T')$ of elements in $\mathcal{U}$ where they only differ at timestamp $t$, i.e., $a_t\not=a_{t}'$.
Let us consider the difference between the generated count stream $\mathcal{C}=(c_1,c_2,\cdots,c_T)$ and $\mathcal{C}'=(c_1',c_2',\cdots,c'_T)$.

Consider any timestamp $i\in [T]$, if $a_i\not=a_t$ and $a_i\not=a_t'$, it is easy to verify that $c_i=c_i'$.
Suppose $i\not=t$.
If $a_i=a_t=u\in\mathcal{U}$ but $a_i$ is at least the third appearance of $u$ in $\mathcal{S}$, then $a_i'=a_i=u$ is at least the second appearance of $u$ in $\mathcal{S}'$ which implies that $c_i=c_i'=0$.
Similarly, if $a_i=a_t'=u\in\mathcal{U}$ is at least third appearance of $u$ in $\mathcal{S}$, we can show $c_i=c_i'=0$ as well.
Thus the sensitivity of the stream $\mathcal{C}$ is at most $5$: only when $i=t$ or $a_i$ is the first/second appearance of $a_t,a_{t'}$, $c_i$ might be different from $c_i'$.
Therefore, if we use an $\varepsilon$-DP algorithm to continually release the total counts of $\mathcal{C}$, the continually released output of Algorithm~\ref{alg:distinct_small_universe} is $(5\cdot\varepsilon)$-DP.
\end{proof}

\begin{theorem}\label{thm:dp_distinct_small_universe}
Let $\varepsilon\geq 0,\xi\in(0,0.5)$, suppose there is an $\varepsilon$-DP streaming continual release summing algorithm (for stream of non-negative numbers) which uses space $J$ and with probability at least $1-\xi$ always outputs an $(\alpha,\gamma)$-approximation for every timestamp.
There is a $(5\varepsilon)$-DP algorithm for the number of distinct elements of streams with universe size at most $m$ in the streaming continual release model.
With probability at least $1-\xi$, the algorithm always outputs an $(\alpha,\gamma)$-approximation for every timestamp $t\in [T]$. 
The algorithm uses $O(m+J)$ space.
\end{theorem}
\begin{proof}
Consider Algorithm~\ref{alg:distinct_small_universe}.
The approximation guarantee is proven by Lemma~\ref{lem:approx_distinct_small_universe}.
The DP guarantee is proven by Lemma~\ref{lem:DP_distinct_small_universe}.
In the remaining of the proof, we only need to prove the space usage.
Since $|\mathcal{U}|\leq m$, the space needed to maintain set $S$ is at most $m$.
The space needed to continually release an $(\alpha,\gamma)$-approximation to the summing problem over $\mathcal{C}$ is at most $\mathcal{J}$.
Thus, the total space needed is at most $O(m+J)$.
\end{proof}

By combining the above theorem with Theorem~\ref{thm:dp_summing}, we obtain the following corollary.
\begin{corollary}[Streaming continual release distinct elements for small universe]\label{cor:distinct_small_universe}
There is an $\varepsilon$-DP algorithm for the number of distinct elements of streams with universe size at most $m$ in the streaming continual release model.
With probability at least $1-\xi$, the additive error of the output is always at most $O\left(\frac{1}{\varepsilon}\log^{2.5}(T)\log\left(\frac{1}{\xi}\right)\right)$ for every timestamp $t\in [T]$. 
The algorithm uses $O(m+\log(T))$ space.
\end{corollary}

By combining Theorem~\ref{thm:dp_distinct_small_universe} with Theorem~\ref{thm:summing_non_negative}, we obtain the following corollary:
\begin{corollary}[Streaming continual release distinct elements for small universe, better additive error]\label{cor:distinct_small_universe_better_additive}
There is an $\varepsilon$-DP algorithm for the number of distinct elements of streams with universe size at most $m$ in the streaming continual release model.
With probability at least $1-\xi$, the additive error of the output is always an $\left(1+\eta,O\left(\frac{\log(T/\xi)}{\varepsilon\eta}\right)\right)$-approximation to the number of distinct elements for every timestamp $t\in [T]$. 
The algorithm uses $O(m)$ space.
\end{corollary}

\subsection{Number of Distinct Elements for General Universe}\label{sec:dist_large_universe}

\begin{algorithm}[H]\label{alg:distinct_large_universe}
\small
	\KwIn{A stream $\mathcal{S}$ of elements $a_1,a_2,\cdots,a_T\in \mathcal{U}\cup \{\perp\}$, and a error parameter $\eta\in (0,0.5)$.}
	\Parameter{Relative approximation factor $\alpha\geq 1$ and additive approximation factor $\gamma\geq 0$ depending on the streaming continual release algorithm for number of distinct elements of streams with small universe of elements. \hfill{//See Theorem~\ref{thm:dp_distinct_small_universe}}.}
	\KwOut{Estimation of the number of distinct elements $\|\mathcal{S}\|_0$.}
	$L\gets\lceil \log \min(|\mathcal{U}|,T) \rceil,\lambda\gets 2\log (1000 L), m\gets 100L\cdot \left(16\alpha\max\left(\gamma/\eta, 32\alpha\lambda/\eta^2\right)\right)^2$.\\
	Let $h:\mathcal{U}\rightarrow [m]$ be a pairwise independent hash function. \\\hfill{//Here we treat $[m]$ as a universe of elements with size $m$ instead of a set of integers.}\\
	Let $g:\mathcal{U}\rightarrow [L]\cup \{\perp\}$ be a $\lambda$-wise independent hash function and $\forall a\in \mathcal{U},i\in[L],\Pr[g(a)=i]=2^{-i},\Pr[g(a)=\perp]=2^{-L}$.\\
	Initialize empty streams $\mathcal{S}_1,\mathcal{S}_2,\cdots,\mathcal{S}_L$.\\
	\For{ each $a_t$ in the stream $\mathcal{S}$}
	{
	    \For {$i\in[L]$} 
	    {
	    \If{$a_t\not=\perp$ and $g(a_t)=i$}{
	        Append $h(a_t)$ to the end of the stream $\mathcal{S}_{i}$.
	    }
	    \Else {
	        Append $\perp$ to the end of the stream $\mathcal{S}_{i}$.
	    }
	    }
	    $\forall i\in [L],$ compute $\hat{s}_i$ which is an $(\alpha,\gamma)$-approximation to  $\|\mathcal{S}_i\|_0$.\\
	Find the largest $i\in [L]$ such that $\hat{s}_i\geq \max\left(\gamma/\eta, 32\alpha\lambda/\eta^2\right)$, and output $\hat{s}_i\cdot 2^i$.\\
	If such $i$ does not exist, output $0$.
	}
    \caption{Number of Distinct Elements via Subsampling}
\end{algorithm}

\begin{lemma}\label{lem:correctness_distinct_large}
%If $\|\mathcal{S}\|_0\geq ?$, with probability at least $0.9$, the output of Algorithm~\ref{alg:distinct_large_universe} is an $?$-approximation to $\|\mathcal{S}\|_0$.
%Otherwise, with probability at least $0.9$, the output of Algorithm~\ref{alg:distinct_large_universe} is either $0$ or an $?$-approximation to $\|\mathcal{S}\|_0$.
Consider any timestamp $t\in[T]$.
Let $v$ be the output of Algorithm~\ref{alg:distinct_large_universe}.
With probability at least $0.9$, $v$ is a $((1+O(\eta))\alpha,O(\alpha^2\max(\gamma/\eta,\alpha\log(L)/\eta^2)))$-approximation to $\|(a_1,a_2,\cdots,a_t)\|_0$.
\end{lemma}

%\begin{remark}
%To boost the success probability of Algorithm~\ref{alg:distinct_large_universe} to $1-\xi$, we can run $O(\log(1/\xi))$ number of parallel copies of Algorithm~\ref{alg:distinct_large_universe} and take the median of the outputs.
%\end{remark}

To prove Lemma~\ref{lem:correctness_distinct_large}, we need following intermediate statements.
We consider a timestamp $t\in [T]$.
Let $\mathcal{S}$ denote the input stream at timestamp $t$, i.e., $\mathcal{S}=(a_1,a_2,\cdots,a_t)$.
Let $G_i = \{a_j \mid g(a_j) = i, j\leq t \}$ for $i \in [L]$.
\begin{claim}\label{cla:good_concentration}
$\forall i \in [L]$, if $\|\mathcal{S}\|_0\geq 2^i\cdot 4\lambda/\eta^2, \Pr[|G_i| \in (1\pm \eta) \cdot\|\mathcal{S}\|_0/2^i]\geq 1 - 0.01/L$.
Otherwise, $\Pr[||G_i|-\|\mathcal{S}\|_0/2^i|\leq4\lambda/\eta] \geq 1 - 0.01/L$.
\end{claim}
\begin{proof}
Suppose $\|\mathcal{S}\|_0\geq 2^i\cdot 4\lambda/\eta^2$.
Due to Lemma~\ref{lem:concentration}, we have:
\begin{align*}
&\Pr\left[\left||G_i| - \|\mathcal{S}\|_0 / 2^i\right| > \eta\cdot \|\mathcal{S}\|_0 / 2^i\right]\\
\leq & 8\cdot \left(\frac{\lambda\cdot \|\mathcal{S}\|_0/2^i + \lambda^2}{\left(\eta\cdot \|\mathcal{S}\|_0/2^i\right)^2}\right)^{\lambda/2}\\
\leq & 0.01/L,
\end{align*}
where the last inequality follows from that $\lambda = 2\cdot \log(1000L)$ and $\|\mathcal{S}\|_0\geq 2^i\cdot 4\lambda / \eta^2$.

Suppose $\|\mathcal{S}\|_0 \leq 2^i \cdot4\lambda/\eta^2$, By applying Lemma~\ref{lem:concentration} again, we have:
\begin{align*}
&\Pr\left[\left||G_i|- \|\mathcal{S}\|_0/2^i\right| >4\lambda/\eta\right]\\
\leq & 8\cdot \left(\frac{\lambda\cdot \|\mathcal{S}\|_0/2^i + \lambda^2}{(4\lambda/\eta)^2}\right)^{\lambda/2}\\
\leq & 0.01/L,
\end{align*}
where the last inequality follows from $\|\mathcal{S}\|_0/2^i \leq 4\lambda / \eta^2$ and $\lambda = 2\cdot \log (1000 L)$.
\end{proof}

\begin{claim}\label{cla:collision}
For $i\in [L]$, conditioning on $|G_i|\leq 16\alpha\max\left(\gamma/\eta, 32\alpha\lambda/\eta^2\right)$, the probability that $|G_i|=\|\mathcal{S}_i\|_0$ is at least $1-0.01/L$.
\end{claim}
\begin{proof}
Since $h$ is pairwise independent, $\forall a,b\in G_i,\Pr[h(a)=h(b)]=1/m$.
Since $m=100L\cdot \left(16\alpha\max\left(\gamma/\eta,32\alpha\lambda/\eta^2\right)\right)^2$, $\Pr[\exists a\not=b\in G_i, h(a)=h(b)]\leq |G_i|^2/m\leq 0.01/L$ by a union bound.
\end{proof}

Let $\mathcal{E}$ be the event that both of the following hold:
\begin{enumerate}
    \item $\forall i \in [L]$ with $2^i\cdot 4\lambda/\eta^2\leq \|\mathcal{S}\|_0$, $|G_i| \in (1\pm \eta)\cdot \|\mathcal{S}\|_0/ 2^i$.
    \item $\forall i \in [L]$ with $2^i\cdot 4\lambda/\eta^2> \|\mathcal{S}\|_0$, $|G_i|\in \|\mathcal{S}\|_0/2^i\pm 4\lambda/\eta$.
\end{enumerate}
According to Claim~\ref{cla:good_concentration}, $\mathcal{E}$ happens with probability at least $0.99$.
Let $\mathcal{E}'$ be the event that $\forall i\in[L]$ with $|G_i|\leq 16\alpha \max(\gamma/\eta,32\alpha\lambda/\eta^2)$, $|G_i|=\|\mathcal{S}_i\|_0$.
According to Claim~\ref{cla:collision}, $\mathcal{E}'$ happens with probability at least $0.99$.

Next, we are going to prove Lemma~\ref{lem:correctness_distinct_large}.
\begin{proof}[Proof of Lemma~\ref{lem:correctness_distinct_large}]
In this proof, we condition on both events $\mathcal{E}$ and $\mathcal{E}'$.
Note that the probability that both $\mathcal{E}$ and $\mathcal{E}'$ happen is at least $0.98$.

Consider the case that $\|\mathcal{S}\|_0\geq 8\alpha\cdot \max\left(\gamma/\eta,32\alpha\lambda/\eta^2\right)$.
Let $i^*\in[L]$ be the largest value such that $\|\mathcal{S}\|_0/2^{i^*}\geq 4\alpha\cdot \max(\gamma/\eta,32\alpha\lambda/\eta^2)$.
According to event $\mathcal{E}$, we have $|G_{i^*}|\in (1\pm \eta)\cdot \|\mathcal{S}\|_0/2^{i^*}$.
Due to our choice of $i^*$, we have $\|\mathcal{S}\|_0/2^{i^*}\leq 8\alpha\cdot \max\left(\gamma/\eta,32\alpha\lambda/\eta^2\right)$.
Thus, $|G_{i^*}|\leq 16\alpha \max\left(\gamma/\eta,32\alpha\lambda/\eta^2\right)$.
According to event $\mathcal{E}'$, we have $\|\mathcal{S}_{i^*}\|_0=|G_{i^*}|$.
Therefore, we have
$
\hat{s}_{i^*} \geq \|\mathcal{S}_{i^*}\|_0/\alpha-\gamma\geq |G_{i^*}|/\alpha - \gamma.
$
Since $|G_{i^*}|\geq 2\alpha \cdot \max\left(\gamma/\eta,32\alpha\lambda/\eta^2\right)$, we have $\hat{s}_{i^*}\geq \max\left(\gamma/\eta,32\alpha\lambda/\eta^2\right)$.
Therefore, Algorithm~\ref{alg:distinct_large_universe} will output $\hat{s}_{i'}\cdot 2^{i'}$ for some $i'\geq i^*$.
Due to event $\mathcal{E}$, we know $|G_{i'}|\leq \max(2\cdot \|\mathcal{S}\|_0/2^{i^*}, \|\mathcal{S}\|_0/2^{i^*}+4\lambda/\eta)\leq 16\alpha\max(\gamma/\eta,32\alpha\lambda/\eta^2)$.
According to event $\mathcal{E}'$, we have $\|\mathcal{S}_{i'}\|_0=|G_{i'}|$.
Since the algorithm outputs $\hat{s}_{i'}\cdot 2^{i'}$, we know that $\hat{s}_{i'}\geq \max(\gamma/\eta,32\alpha\lambda/\eta^2)$ which implies that
\begin{align*}
|G_{i'}|&=\|\mathcal{S}_{i'}\|_0\\
&\geq (\hat{s}_{i'}-\gamma)/\alpha\\
&\geq (1-\eta)\hat{s}_{i'}/\alpha\\
&\geq 16\lambda/\eta^2.
\end{align*}
According to event $\mathcal{E}$, we have $|G_{i'}|\in (1\pm \eta)\cdot \|\mathcal{S}\|_0/2^{i'}$.
Thus, we have
\begin{align*}
\hat{s}_{i'}&\leq \alpha \|\mathcal{S}_{i'}\|_0 +\gamma\\
&=\alpha |G_{i'}| +\gamma\\
&\leq \frac{\alpha}{1-\eta}\cdot |G_{i'}|\\
&\leq \frac{1+\eta}{1-\eta}\cdot \alpha \cdot \|\mathcal{S}\|_0/2^{i'}\\
&\leq (1+4\eta)\alpha\|\mathcal{S}\|_0/2^{i'},
\end{align*}
where the second inequality follows from $\hat{s}_{i'}\geq \gamma/\eta$ and the last inequality follows from $\eta\leq 0.5$.
Similarly, we have:
\begin{align*}
\hat{s}_{i'}&\geq \|\mathcal{S}_{i'}\|_0/\alpha - \gamma\\
&=|G_{i'}|/\alpha - \gamma\\
&\geq |G_{i'}|/((1+\eta)\alpha)\\
&\geq \frac{1-\eta}{1+\eta}\cdot \frac{1}{\alpha}\cdot \|\mathcal{S}\|_0/2^{i'}\\
&\geq (1-4\eta)/\alpha\cdot \|\mathcal{S}\|_0 / 2^{i'},
\end{align*}
where the second inequality follows from $\hat{s}_{i'}\geq \gamma/\eta$.

Next, consider the case that $\|\mathcal{S}\|_0<8\alpha\cdot \max(\gamma/\eta,32\alpha\lambda/\eta^2)$.
Algorithm~\ref{alg:distinct_large_universe} either outputs $0$ or outputs $\hat{s}_{i'}\cdot 2^{i'}$ for some $i'\in [L]$.
Suppose it outputs $\hat{s}_{i'}\cdot 2^{i'}$.
We have $\hat{s}_{i'}\geq \max\left(\gamma/\eta,32\alpha\lambda/\eta^2\right)$, which implies that
\begin{align*}
|G_{i'}|\geq \|\mathcal{S}_{i'}\|_0 \geq (\hat{s}_{i'}-\gamma)/\alpha\geq (1-\eta)\hat{s}_{i'}/\alpha\geq 16\lambda/\eta^2.
\end{align*}
According to event $\mathcal{E}$, we have $|G_{i'}|\leq 2\cdot \|\mathcal{S}\|_0/2^{i'}$.
Therefore
\begin{align*}
\hat{s}_{i'}\cdot 2^{i'} &\leq (\alpha\|\mathcal{S}_{i'}\|_0+\gamma)\cdot 2^{i'}\\
&\leq \frac{\alpha}{1-\eta} \cdot \|\mathcal{S}_{i'}\|_0 \cdot 2^{i'}\\
&\leq \frac{\alpha}{1-\eta} \cdot |G_{i'}| \cdot 2^{i'}\\
&\leq 4\alpha\|\mathcal{S}\|_0\\
&\leq 32\alpha^2\max(\gamma/\eta,32\alpha\lambda/\eta^2).
\end{align*}
\end{proof}

\begin{theorem}\label{thm:distinct_elements}
Let $\varepsilon\geq 0,\xi,\xi'\in(0,0.5),\eta\in(0,0.5)$, suppose there is an $\varepsilon$-DP algorithm for the number of distinct elements of streams with element universe size at most $ 100\log(\min(|\mathcal{U}|,T))\cdot (16\alpha\max(\gamma/\eta, 32\alpha\cdot 2\log(1000\log(\min(|\mathcal{U}|,T)))/\eta^2))^2$ in the streaming continual release model which uses space $J$ and with probability at least $1-\xi$ always outputs an $(\alpha,\gamma)$-approximation for every timestamp. 
There is an $(\varepsilon'=\lceil50\log(T/\xi')\rceil\varepsilon)$-DP algorithm for the number of distinct elements of streams with element universe $\mathcal{U}$ in the streaming continual release model.
With probability at least $1-\xi' -\log(\min(|\mathcal{U}|,T))\cdot \lceil50\log(T/\xi')\rceil\cdot \xi$, the algorithm always outputs an $((1+O(\eta))\alpha,O(\alpha^2\max(\gamma/\eta,\alpha\log\log(\min(|\mathcal{U}|,T))/\eta^2)))$-approximation for every timestamp $t\in [T]$. 
The algorithm uses $O(J\cdot \log(\min(|\mathcal{U}|,T))\cdot \log(T/\xi'))$ space.
\end{theorem}
\begin{proof}
According to our construction of $\mathcal{S}_1,\mathcal{S}_2,\cdots,\mathcal{S}_L$ and the definition of the sensitivity (Definition~\ref{def:sensitivity}), $(\mathcal{S}_1,\mathcal{S}_2,\cdots,\mathcal{S}_L)$ has sensitivity $1$.
Since the algorithm to report $\hat{s}_i$ for each $i\in [L]$ is $\varepsilon$-DP, Algorithm~\ref{alg:distinct_large_universe} is $\varepsilon$-DP.
According to Lemma~\ref{lem:correctness_distinct_large}, for any $t\in[T]$, condition on that $\hat{s}_i$ is an $(\alpha,\gamma)$-approximation to $\|\mathcal{S}\|_0$, with probability at least $0.9$, the output is a $((1+O(\eta))\alpha,O(\alpha^2\max(\gamma/\eta,\alpha\log\log(\min(|\mathcal{U}|,T))/\eta^2)))$-approximation.
To boost the probability to $1-\xi'$ such that $\forall t\in[T]$ the approximation guarantee always holds, we need to run $\lceil50\log(T/\xi')\rceil$ independent copies of Algorithm~\ref{alg:distinct_large_universe} and take the median of the outputs at every timestamp.
Thus, the overall algorithm is $(\lceil50\log(T/\xi')\rceil\cdot \varepsilon)$-DP.

Next, consider the success probability that every $\hat{s}_i$ is an $(\alpha,\beta)$-approximation to $\|\mathcal{S}_i\|_0$.
By taking a union bound over all $i$ and all independent copies of Algorithm~\ref{alg:distinct_large_universe}, the success probability is at least $1-\log(\min(|\mathcal{U}|,T))\cdot \lceil50\log(T/\xi')\rceil\cdot \xi$.

Finally, consider the space usage.
Consider each running copy of Algorithm~\ref{alg:distinct_large_universe}. 
Hashing function $h(\cdot)$ takes $O(1)$ space. 
Hashing function $g(\cdot)$ takes $O(\lambda) = O(\log\log(\min(|\mathcal{U}|, T)))$ space.
To continually release $\hat{s}_i$ for all $i$, we need to use $O(J\cdot \log(\min(|\mathcal{U}|,T)))$ space.
Thus, the total space needed for all copies is at most $O(J\cdot \log(\min(|\mathcal{U}|,T))\cdot \log(T/\xi'))$.
\end{proof}

By plugging Corollary~\ref{cor:distinct_small_universe} into above theorem with $\xi=\frac{\xi'/2}{\log(\min(|\mathcal{U}|,T))\cdot \lceil 50\log(2T/\xi')\rceil},\varepsilon = \frac{\varepsilon'}{\lceil50\log(2T/\xi')\rceil}$, $\alpha = 1,\gamma = O(\frac{1}{\varepsilon}\log^{2.5}(T)\log(1/\xi))$ and $J=O(\log(T)+100\log(\min(|\mathcal{U}|,T))\cdot (16\alpha\max(\gamma/\eta, 32\alpha\cdot 2\log(1000\log(\min(|\mathcal{U}|,T)))/\eta^2))^2)$, we get the following corollary.
\begin{corollary}[Streaming continual release distinct elements]\label{cor:distinct_large_universe}
For $\eta\in(0,0.5),$ there is an $\varepsilon$-DP algorithm for the number of distinct elements of streams with element universe $\mathcal{U}$ in the streaming continual release model.
With probability at least $1-\xi$, the output is always a $(1+\eta,O\left(\max\left(\frac{\log(T/\xi)\log^{2.5}(T)\log(1/\xi)\log\log(T/\xi)}{\eta\varepsilon},\frac{\log\log T}{\eta^2}\right)\right))$-approximation for every timestamp $t\in [T]$. 
The algorithm uses $\poly\left(\frac{\log(T/\xi)}{\eta\min(\varepsilon,1)}\right)$ space.
\end{corollary}

By plugging Corollary~\ref{cor:distinct_small_universe_better_additive} into Theorem~\ref{thm:distinct_elements} with $\xi=\frac{\xi'/2}{\log(\min(|\mathcal{U}|,T))\cdot \lceil 50\log(2T/\xi')\rceil},\varepsilon = \frac{\varepsilon'}{\lceil50\log(2T/\xi')\rceil}$, $\alpha = 1+\eta,\gamma = O\left(\frac{\log\left(T/\xi\right)}{\varepsilon\eta}\right)$ and 
\begin{align*}
J=O(100\log(\min(|\mathcal{U}|,T))\cdot (16\alpha\max(\gamma/\eta, 32\alpha\cdot 2\log(1000\log(\min(|\mathcal{U}|,T)))/\eta^2))^2),
\end{align*}
we get the following corollary.
\begin{corollary}[Streaming continual release distinct elements, better dependence in $\log(T)$]\label{cor:distinct_large_universe_better_additive}
For $\eta\in(0,0.5),$ there is an $\varepsilon$-DP algorithm for the number of distinct elements of streams with element universe $\mathcal{U}$ in the streaming continual release model.
With probability at least $1-\xi$, the output is always a $(1+O(\eta),O\left(\frac{\log^2(T/\xi)}{\eta^2\varepsilon}\right))$-approximation for every timestamp $t\in [T]$. 
The algorithm uses $\poly\left(\frac{\log(T/\xi)}{\eta\min(\varepsilon,1)}\right)$ space.
\end{corollary}
\section{Continual Released $\ell_p$ Heavy Hitters and Frequency Moment Estimation}
In this section, we present $\varepsilon$-DP streaming continual release algorithms for $\ell_p$ heavy hitters and frequency moment estimation.
In Section~\ref{sec:count_sketch}, we present an algorithm for $\varepsilon$-DP CountSketch~\cite{charikar2002finding} in the streaming continual release model.
The CountSketch is used for $\ell_2$ heavy hitters and $\ell_2$ moment estimation.
In Section~\ref{sec:lp_heavyhitters}, we show how to use $\ell_2$ heavy hitters to solve $\ell_p$ heavy hitters.
In Section~\ref{sec:low_freq_elements}, we show how to estimate the number of elements which have low frequencies. 
In Section~\ref{sec:lp_moment}, we show how to use $\ell_p$ heavy hitters and the estimator of low frequency elements to estimate the $\ell_p$ frequency moment.

\subsection{Continual Released CountSketch}\label{sec:count_sketch}

\begin{algorithm}[H]\label{alg:countsketch}
\small
	\KwIn{A stream $\mathcal{S}$ of elements $a_1,a_2,\cdots,a_T\in \mathcal{U}\cup \{\perp\}$, a parameter $k\in\mathbb{Z}_{\geq 1}$.}
	\Parameter{Relative approximation factor $\alpha\geq 1$ and additive approximation factor $\gamma\geq 0$ depending on the streaming continual release summing algorithm.\\ \hfill{//See Theorem~\ref{thm:dp_summing}}. }
	\KwOut{A tuple $(z_1,z_2,\cdots,z_k)$ at every timestamp $t$.}
	Let $h:\mathcal{U}\rightarrow [k]$ be a $4$-wise independent hash function, s.t., $\forall a\in\mathcal{U},i\in[k],\Pr[h(a)=i]=\frac1k.$\\
	Let $g:\mathcal{U}\rightarrow \{-1,1\}$ be a $4$-wise independent hash function, s.t., $\forall a\in \mathcal{U},\Pr[g(a)=1]=\frac12$.\\
	Initialize empty streams $\mathcal{S}_1,\mathcal{S}_2,\cdots,\mathcal{S}_k$.\\
	\For{ each $a_t$ in the stream $\mathcal{S}$}
	{
	    \If{$a_t=\perp$}{
	        Append $0$ to the end of every stream $\mathcal{S}_1,\mathcal{S}_2,\cdots,\mathcal{S}_k$.
	    }
	    \Else{
	        Append $g(a_t)$ to the end of the stream $\mathcal{S}_{h(a_t)}$
	        and append $0$ to the end of every stream $\mathcal{S}_i$ for $i\not = h(a_t)$.
	    }
	    Output a tuple $(z_1,z_2,\cdots,z_k)$ where $z_i$ is an estimation of the total counts of $\mathcal{S}_i$ with additive error at most $\gamma$.
	}
    \caption{Continual Released CountSketch}
\end{algorithm}

\begin{lemma}[DP guarantee]\label{lem:countsketch_dp_guarantee}
If the subroutine of continually releasing the approximate total counts of $\mathcal{S}_i$ for every $i\in[k]$ in Algorithm~\ref{alg:countsketch} is $\varepsilon$-DP, Algorithm~\ref{alg:countsketch} is $2\varepsilon$-DP.
\end{lemma}
\begin{proof}
Consider two neighboring streams $\mathcal{S}$ and $\mathcal{S}'$ where the only difference is the $t$-th element, i.e., $a_t\not = a_t'$.
Consider their corresponding streams $\mathcal{S}_1,\mathcal{S}_2,\cdots,\mathcal{S}_k$ and $\mathcal{S}'_1,\mathcal{S}'_2,\cdots,\mathcal{S}'_k$ in Algorithm~\ref{alg:countsketch}.
For $i\in[k]$ and $j\not=t$, the $j$-th number in $\mathcal{S}_i$ should be the same as the $j$-th number in $\mathcal{S}'_i$.
Thus, we only need to consider the $t$-th number in $\mathcal{S}_i$ and $\mathcal{S}'_i$ for every $i\in [k]$.
Since $a_t$ can only make the $t$-th number of $\mathcal{S}_{h(a_t)}$ be non-zero, $a'_t$ can only make the $t$-th number of $\mathcal{S}'_{h(a'_t)}$ be non-zero, and the non-zero number can be only $\pm 1$, the sensitivity is at most $2$.
Thus, if we use $\varepsilon$-DP algorithm to continually release the total counts of $\mathcal{S}_1,\mathcal{S}_2,\cdots,\mathcal{S}_k$, the continually released output of Algorithm~\ref{alg:countsketch} is $2\varepsilon$-DP.
\end{proof}

\begin{lemma}[Good approximation for frequent elements]\label{lem:countsketch_frequency_elements}
Consider any $a\in \mathcal{U}$ and any timestamp $t\in [T]$.
Let $f_a$ be the frequency of $a$ in $a_1,a_2,\cdots,a_t$.
Let $(z_1,z_2,\cdots,z_k)$ be the output of Algorithm~\ref{alg:countsketch} at timestamp $t$.
Then $\forall \eta\in(0,0.5)$, with probability at least $1-1/(k\eta^2)$, $|f_a-g(a)\cdot z_{h(a)}|\leq \eta\cdot\sqrt{\sum_{b\in\mathcal{U}}f_b^2} + \gamma$.
\end{lemma}
\begin{proof}
Let $\hat{z}_{h(a)}$ be the true total counts of stream $\mathcal{S}_{h(a)}$ at timestamp $t$.
According to the original CountSketch~\cite{charikar2002finding}, we have 
$\Pr\left[|f_a-g(a)\cdot \hat{z}_{h(a)}|\leq \eta \sqrt{\sum_{b\in\mathcal{U}}f_b^2}\right]\geq 1-1/(k\eta^2)$.
Since
\begin{align*}
|f_a-g(a)\cdot z_{h(a)}|\leq |f_a-g(a)\cdot \hat{z}_{h(a)}|+|g(a)|\cdot| z_{h(a)}-\hat{z}_{h(a)}|\leq |f_a-g(a)\cdot \hat{z}_{h(a)}|+\gamma,
\end{align*}
we have with probability at least $1-1/(k\eta^2)$, $|f_a-g(a)\cdot z_{h(a)}|\leq \eta\cdot\sqrt{\sum_{b\in\mathcal{U}}f_b^2} + \gamma$.
\end{proof}

\begin{comment}
\jm{Maybe we can put this kind of remarks as one theorem in preliminaries. It looks like we have at least three of them now.}
\pz{Agreed}
\
\begin{remark}
To boost the success probability of the estimation in Lemma~\ref{lem:countsketch_frequency_elements} to $1-\xi$, we can set $k=3/\eta^2$ and repeat Algorithm~\ref{alg:countsketch} $O(\log(1/\xi))$ times and take the median of the estimations.
\end{remark}
\end{comment}

\begin{lemma}[$\ell_2$ Frequency moment estimation]\label{lem:countsketch_frequency_momement}
Consider any timestamp $t\in[T]$. 
For $a\in\mathcal{U}$, let $f_a$ be the frequency of $a$ in $a_1,a_2,\cdots,a_t$.
Let $(z_1,z_2,\cdots,z_k)$ be the output of Algorithm~\ref{alg:countsketch} at timestamp $t$.
Then $\forall \eta \in (0,0.5)$, with probability at least $1-100/(k\eta^2)$, $|\sum_{i=1}^k z_i^2- \sum_{a\in \mathcal{U}} f_a^2|\leq \eta \sum_{a\in \mathcal{U}} f_a^2 + 4k\gamma^2/\eta$
\end{lemma}
\begin{proof}
Let $F_2=\sum_{a\in\mathcal{U}}f_a^2$.
Let $Z=\sum_{i\in[k]} z_i^2$.
For $i\in[k]$, let $\hat{z}_i$ be the true total counts of stream $\mathcal{S}_i$ at timestamp $t$.
Let $\hat{Z}=\sum_{i\in[k]}\hat{z}_i^2$.
According to the analysis of CountSketch~\cite{charikar2002finding,thorup2004tabulation},
We have $\Pr\left[|F_2-\hat{Z}|\leq \eta/4\cdot F_2\right]\geq 1-100/(k\eta^2)$.
In the following, we condition on $|F_2-\hat{Z}|\leq \eta/4\cdot F_2$.

We have
\begin{align*}
&|F_2-Z|\\
\leq& |F_2-\hat{Z}| + |\hat{Z}-Z|\\
\leq& \eta/4\cdot F_2+ \sum_{i=1}^k |z_i^2-\hat{z}_i^2|.
\end{align*}
Denote $z_i=\hat{z}_i + v_i$ for $i\in[k]$.
We have $|v_i|\leq \gamma$.
Due to convexity, we have:
\begin{align*}
(1-\eta/4)\hat{z}_i^2-4v_i^2/\eta\leq (\hat{z}_i+v_i)^2 \leq \hat{z}_i^2/(1-\eta/4) + 4v_i^2/\eta.
\end{align*}
Since $\eta\in(0,0.5)$, $1/(1-\eta/4)\leq (1+\eta/2)$.
Therefore, $|z_i^2-\hat{z}_i^2|\leq \eta/2\cdot \hat{z}_i^2+4v_i^2/\eta$.
We have:
\begin{align*}
|F_2-Z|\leq \eta/4\cdot F_2 + \eta/2\cdot \hat{Z} + 4k\gamma^2/\eta\leq \eta/4\cdot F_2 + \eta/2 \cdot 3/2 \cdot F_2 + 4k\gamma^2/\eta\leq \eta F_2 + 4k\gamma^2/\eta.
\end{align*}
\end{proof}

\begin{comment}
\begin{remark}
To boost the success probability of the estimation in Lemma~\ref{lem:countsketch_frequency_momement} to $1-\xi$, we can set $k=300/\eta^2$ and repeat Algorithm~\ref{alg:countsketch} $O(\log(1/\xi))$ times and take the median of the estimations.
\end{remark}
\end{comment}

\begin{theorem}[Streaming continual release $\ell_2$ frequency estimators]\label{thm:count_sketch}
Let $\varepsilon>0,\eta\in(0,0.5),\xi\in(0,0.5)$.
There is an $\varepsilon$-DP algorithm in the streaming continual release model such that with probability at least $1-\xi$, it always outputs for every timestamp $t\in [T]$:
\begin{enumerate}
    \item $\hat{f}_a$ for every $a\in\mathcal{U}$ such that $|f_a-\hat{f}_a|\leq \eta \|\mathcal{S}\|_2+O\left(\frac{\log(T/\xi)+\log(|\mathcal{U}|)}{\varepsilon}\cdot \log^{2.5}(T)\cdot \log\left(\frac{\log(T/\xi)+\log(|\mathcal{U}|)}{\xi\eta}\right)\right)$, where $\mathcal{S}$ denotes the stream $(a_1,a_2,\cdots,a_t)$ and $f_a$ denotes the frequency of $a$ in $\mathcal{S}$,
    \item $\hat{F}_2$ such that $|\hat{F}_2-\|\mathcal{S}\|_2^2|\leq \eta \|\mathcal{S}\|_2^2+O\left(\frac{(\log(T/\xi)+\log(|\mathcal{U}|))^2}{\varepsilon^2\eta^3}\cdot \log^{5}(T)\cdot \log^2\left(\frac{\log(T/\xi)+\log(|\mathcal{U}|)}{\xi\eta}\right)\right)$
\end{enumerate}
The algorithm uses $O\left(\frac{\log(T/\xi)+\log(|\mathcal{U}|)}{\eta^2}\cdot \log(T)\right)$ space.
\end{theorem}
\begin{proof}
Suppose we set $k=400/\eta^2$.
Due to Lemma~\ref{lem:countsketch_frequency_elements} and Lemma~\ref{lem:countsketch_frequency_momement}, the approximation guarantees hold with probability at least $2/3$ for each particular timestamp $t\in [T]$ and $a\in\mathcal{U}$.
To boost the success probability to $1-\xi/2$ for the approximation guarantees and simultaneously for all $t\in [T]$ and all $a\in\mathcal{|U|}$, according to Lemma~\ref{lem:median_trick}, we run $\lceil 50(\log(2T/\xi) + \log(|\mathcal{U}|)) \rceil$ copies of Algorithm~\ref{alg:countsketch} and take the median of each estimator.

We apply Theorem~\ref{thm:dp_summing} for the summing problem of $\mathcal{S}_i$ for each $i\in [k]$.
Since we run $\lceil 50(\log(2T/\xi) + \log(|\mathcal{U}|)) \rceil$ copies of Algorithm~\ref{alg:countsketch}, if we desire $\varepsilon$-DP algorithm in the end, we need each summing subroutine to be $\varepsilon/ (2\cdot \lceil 50(\log(2T/\xi) + \log(|\mathcal{U}|)) \rceil)$-DP according to Lemma~\ref{lem:countsketch_dp_guarantee}.
To simultaneously make the call of each run of the summing subroutine succeeds with probability at least $1-\xi/2$, we need to apply union bound over all calls of summing and thus each run of the summing subroutine should success with probability at least $1-\xi/ (2\cdot \lceil 50(\log(2T/\xi) + \log(|\mathcal{U}|))\rceil \cdot k)=1-\xi/ (2\cdot \lceil 50(\log(2T/\xi) + \log(|\mathcal{U}|))\rceil \cdot 400/\eta^2)$.
Thus, according to Theorem~\ref{thm:dp_summing}, we have $\alpha=1$ and $\gamma = O\left(\frac{\log(T/\xi)+\log(|\mathcal{U}|)}{\varepsilon}\cdot \log^{2.5}(T)\cdot \log\left(\frac{\log(T/\xi)+\log(|\mathcal{U}|)}{\xi\eta^2}\right)\right)$.

Finally, let us consider the total space usage, since we call $(\lceil 50(\log(2T/\xi) + \log(|\mathcal{U}|))\rceil \cdot 400/\eta^2)$ times of summing subroutine, the space needed for executing them is $O\left(\frac{\log(T/\xi)+\log(|\mathcal{U}|)}{\eta^2}\cdot \log(T)\right)$.
Additional space needed is $O(1)$.
Thus, the total space required is $O\left(\frac{\log(T/\xi)+\log(|\mathcal{U}|)}{\eta^2}\cdot \log(T)\right)$
\end{proof}

\subsection{Continual Released $\ell_p$ Heavy Hitters} \label{sec:lp_heavyhitters}
By applying the CountSketch, we are able to develop $\ell_p$ heavy hitters.

\begin{algorithm}[H]\label{alg:lp_heavyhitters}
\small
	\KwIn{A stream $\mathcal{S}$ of elements $a_1,a_2,\cdots,a_T\in \mathcal{U}\cup \{\perp\}$, a parameter $k\in\mathbb{Z}_{\geq 1}$, an error parameter $\eta\in (0,0.5)$.}
	\Parameter{Additive error parameters $\gamma_1,\gamma_2\geq 0$ depending on the streaming continual release CountSketch algorithm. \hfill{//See Theorem~\ref{thm:count_sketch}}.}
	\KwOut{A set $H\subseteq \mathcal{U}$ of elements and their estimated frequencies $\hat{f}:H\rightarrow \mathbb{R}_{\geq 0}$ at every timestamp $t$.}
	Let $\phi \geq \max\left(|\mathcal{U}|^{1-2/p},1\right)$.
	Let $m = 10k^2$. \\
	Let $h:\mathcal{U}\rightarrow [m]$ be a pairwise independent hash function where $\forall a\in\mathcal{U},i\in[m],\Pr[h(a)=i]=1/m$. \\
	Initialize empty streams $\mathcal{S}_1,\mathcal{S}_2,\cdots,\mathcal{S}_m$.
	\For{ each $a_t$ in the stream $\mathcal{S}$ }{
	    \If{$a_t=\perp$}{
	        Append $\perp$ to the end of every stream $\mathcal{S}_1,\mathcal{S}_2,\cdots,\mathcal{S}_m$.
	    }
	    \Else{
	        Append $a_t$ to the end of the stream $\mathcal{S}_{h(a_t)}$ and append $\perp$ to the end of every stream $\mathcal{S}_i$ for $i\not =h(a_t)$.
	    }
	    For $i\in [m]$, compute $\hat{F}_{2,i}$ which is a $(1.1,\gamma_1)$-approximation to $\|\mathcal{S}_i\|_2^2$.\\
	    For $a\in \mathcal{U}$, compute $\hat{f}_a$ which is a $(1,(\eta/16)/(10\sqrt{\phi k})\cdot\|\mathcal{S}_{h(a)}\|_2+\gamma_2)$-approximation to $f_a$, the frequency of $a$ in $\mathcal{S}$ (or equivalently in $\mathcal{S}_{h(a)}$).\\
	        For $a\in\mathcal{U}$, if $\hat{f}_a^2\geq \frac{\hat{F}_{2,h(a)}+\gamma_1}{25\phi k}+\frac{512\gamma_2^2}{\eta^2} $, add $a$ into $\hat{H}$.\\
	    Let $H\subseteq \hat{H}$ only keep the elements $a$ such that $\hat{f}_a$ is one the top-$\left(\left(\frac{1+\eta}{1-\eta}\right)^p\cdot k\right)$ values among $\{\hat{f}_b\mid b\in \hat{H}\}$. 
	    For each $a\in H$, report $\hat{f}(a)\gets \hat{f}_a$.
	}
    \caption{Continual Released $\ell_p$ Heavy Hitters ($p\in [0,\infty)$)}
\end{algorithm}

\begin{lemma}[DP guarantee]\label{lem:lp_heavyhitter_dp_guarantee}
If $\forall i\in[m]$ the subroutine in Algorithm~\ref{alg:lp_heavyhitters} of continually releasing $\hat{F}_{2,i}$ and $\hat{f}_a$ for all $a$ satisfying $h(a)=i$ is $\varepsilon$-DP, Algorithm~\ref{alg:lp_heavyhitters} is $2\varepsilon$-DP.
\end{lemma}
\begin{proof}
Consider two neighboring streams $\mathcal{S}$ and $\mathcal{S}'$ where the only difference is the $t$-th element, i.e., $a_t\not=a'_t$.
If $a_t\not=\perp$, it only causes the difference of at most one element between $\mathcal{S}_{h(a_t)}$ and $\mathcal{S}'_{h(a_t)}$.
Similarly, if $a'_t\not=\perp$, it only causes the difference of at most one element between $\mathcal{S}_{h(a'_t)}$ and $\mathcal{S}'_{h(a'_t)}$.
Thus if for each $i\in[m]$, the continual release algorithm which releases $\hat{F}_{2,i}$ and $\hat{f}_a$ for every $a\in\mathcal{U}$ with $h(a)=i$ is $\varepsilon$-DP, the overall algorithm is $2\varepsilon$-DP.
\end{proof}

\begin{lemma}\label{lem:heavyhitter_goodapprox}
At any timestamp $t\in[T]$, $\forall a\in \mathcal{U}$, if $a\in \hat{H}$, $(1-\eta)f^2_a\leq \hat{f}^2(a)\leq (1+\eta)f^2_a$ where $f_a$ is the frequency of $a$ in $a_1,a_2,\cdots,a_t$.
\end{lemma}
\begin{proof}
Since $a\in \hat{H}$, we have: 
\begin{align*}
\hat{f}_a^2\geq \frac{\hat{F}_{2,h(a)}+\gamma_1}{25\phi k}+\frac{512\gamma_2^2}{\eta^2}.
\end{align*}
Thus:
\begin{align*}
\frac{\eta}{16}\cdot \hat{f}_a^2&\geq \frac{16}{\eta}\cdot \left(\frac{2\cdot(\eta/16)^2}{100\phi k}\cdot \left(2\hat{F}_{2,h(a)}+2\gamma_1\right)+2\gamma_2^2\right)\\
&\geq \frac{16}{\eta}\cdot \left(\frac{2\cdot(\eta/16)^2}{100 \phi k}\cdot \|\mathcal{S}_{h(a)}\|_2^2+2\gamma_2^2\right)\\
&= \frac{16}{\eta}\cdot \left(2\cdot\left(\frac{\eta/16}{10\sqrt{\phi k}}\cdot \|\mathcal{S}_{h(a)}\|_2 \right)^2+2\gamma_2^2\right)\\
&\geq \frac{16}{\eta}\cdot \left(\frac{\eta/16}{10\sqrt{\phi k}}\cdot \|\mathcal{S}_{h(a)}\|_2 + \gamma_2\right)^2
\end{align*}
By convexity:
\begin{align*}
\hat{f}_a^2\geq (1-\eta/16)\cdot f_a^2-\frac{16}{\eta}\cdot \left(\frac{\eta/16}{10\sqrt{\phi k}}\cdot \|\mathcal{S}_{h(a)}\|_2 + \gamma_2\right)^2\geq (1-\eta/16)\cdot f_a^2-\eta/16\cdot\hat{f}_a^2
\end{align*}
and
\begin{align*}
\hat{f}_a^2\leq 1/(1-\eta/16)\cdot f_a^2 + \frac{16}{\eta}\cdot \left(\frac{\eta/16}{10\sqrt{\phi k}}\cdot \|\mathcal{S}_{h(a)}\|_2 + \gamma_2\right)^2\leq 1/(1-\eta/16)\cdot f_a^2+\eta/16\cdot\hat{f}_a^2.
\end{align*}
Since $\eta\in(0,0.5)$, we have $(1-\eta)f_a^2\leq \hat{f}_a^2\leq (1+\eta)f_a^2$.
\end{proof}

\begin{lemma}\label{lem:lp_heavyhitter_size_output}
At any timestamp $t$, the output $H$ of Algorithm~\ref{alg:lp_heavyhitters} has size at most $\left(\frac{1+\eta}{1-\eta}\right)^p\cdot k$.
\end{lemma}
\begin{proof}
Note that $H$ only keeps the top-$\left(\left(\frac{1+\eta}{1-\eta}\right)^p\cdot k\right)$ values from $\hat{H}$.
\end{proof}

\begin{comment}
\begin{proof}
According to Lemma~\ref{lem:heavyhitter_goodapprox}, $\forall a\in H$, we have $\hat{f}_a^2\in (1\pm \eta)f_a^2$ and $\hat{f}_a^2\geq (\hat{F}_{2,h(a)}+\gamma_1)/(25\phi k)+512\gamma_2^2/\eta^2$.
Therefore, $\forall a\in H$,
\begin{align*}
f_a^2\geq \frac{1}{2}\cdot \hat{f}_a^2\geq \frac{\hat{F}_{2,h(a)}+\gamma_1}{50\phi k}\geq \frac{\|\mathcal{S}_{h(a)}\|_2^2}{100\phi k}.
\end{align*}
Since $\forall i\in[m],\|\mathcal{S}_{i}\|_2^2=\sum_{a\in\mathcal{U}:h(a)=i}f_a^2$, there are at most $100\phi k$ different $a\in H$ with $h(a)=i$.
Thus $H\leq m\cdot 100 \phi k\leq 1000 \phi k^3$.
\end{proof}
\end{comment}

\begin{lemma}\label{lem:acc_heavy_hitter}
At any timestamp $t$, consider any $a\in \mathcal{U}$.
Let $f_a$ be the frequency of $a$ in $a_1,a_2,\cdots,a_t$.
If $f_a\geq 4 \sqrt{\gamma_1/(\phi k)+512\gamma_2^2/\eta^2}$ and $f_a^p\geq \|\mathcal{S}\|_p^p/k$, with probability at least $0.9$, $a\in \hat{H}$.
\end{lemma}
\begin{proof}
In this proof, we consider all streams and variables at the timestamp $t$.
Suppose $f_a^p\geq \|\mathcal{S}\|_p^p/k$.
Let $B=\{b\in \mathcal{U}\mid f_b^p\geq \|\mathcal{S}\|_p^p/k\}$.
Then with probability at least $0.9$, $\forall b\in B\setminus\{a\},h(a)\not =h(b)$.
In the remaining of the proof, we condition on $\forall b\in B\setminus\{a\},h(a)\not =h(b)$.

\xhdr{Case 1, $p\leq 2$:}
In this case, we have $\phi=1$.
 $\forall x\in\mathcal{U}$ with $h(x)=h(a)$, we have $f_x\leq f_a$ which implies that $f_x^{p-2}\geq f_a^{p-2}$ since $p\leq 2$.
Since $\forall x\in\mathcal{U},\frac{f_x^p}{f_x^2}\leq 1$, we have
\begin{align*}
\frac{\|\mathcal{S}_{h(a)}\|_p^p}{\|\mathcal{S}_{h(a)}\|_2^2}=\frac{\sum_{x\in\mathcal{U}:h(x)=h(a)} f_x^p}{\sum_{x\in\mathcal{U}:h(x)=h(a)} f_x^2}\geq \min_{x\in\mathcal{U}:h(x)=h(a)}\frac{f_x^p}{f_x^2} =\frac{f_a^p}{f_a^2},
\end{align*}
which implies that $f_a^2/\|\mathcal{S}_{h(a)}\|_2^2\geq f_a^p/\|\mathcal{S}_{h(a)}\|_p^p\geq f_a^p/\|\mathcal{S}\|_p^p\geq 1/k$ and thus $f_a^2\geq \|\mathcal{S}_{h(a)}\|_2^2/(\phi k)$.

\xhdr{Case 2, $p>2$:}
In this case, we have $\phi = |\mathcal{U}|^{1-2/p}$.
Since $f_a^p\geq \|\mathcal{S}\|_p^p/k$, we have:
\begin{align*}
f_a\geq \|\mathcal{S}\|_p/k^{1/p}\geq \|\mathcal{S}\|_p/k^{1/2}\geq \|\mathcal{S}\|_2 / ( k^{1/2} \cdot |\mathcal{U}|^{1/2-1/p}),
\end{align*}
where the second inequality follows from $k^{1/2}\geq k^{1/p}$ for $p>2$, and the third inequality follows from Holder's inequality that $\|\mathcal{S}\|_2\leq |\mathcal{U}|^{1/2-1/p}\cdot\|\mathcal{S}\|_p$.
Therefore, $f_a^2\geq \|\mathcal{S}\|_2^2/(\phi k)$.

Therefore, in both above cases, we always have $f_a^2\geq \|\mathcal{S}\|_2^2/(\phi k)$.

By convexity, we have
\begin{align*}
\hat{f}_a^2&\geq (1-\eta/16)\cdot f_a^2 - \frac{16}{\eta}\cdot \left(\frac{\eta/16}{10\sqrt{\phi k}}\cdot \|\mathcal{S}_{h(a)}\|_2+\gamma_2\right)^2\\
&\geq (1-\eta/16)\cdot f_a^2 - \frac{16}{\eta}\cdot \left(\frac{2(\eta/16)^2}{100\phi k}\cdot \|\mathcal{S}_{h(a)}\|_2^2+2\gamma_2^2\right)\\
\end{align*}
Thus we have:
\begin{align}
\hat{f}_a^2 &\geq \frac{1}{2}\cdot \|\mathcal{S}_{h(a)}\|_2^2/(\phi k) - \frac{16}{\eta}\cdot \left(\frac{2(\eta/16)^2}{100\phi k}\cdot \|\mathcal{S}_{h(a)}\|_2^2+2\gamma_2^2\right)\notag\\
&\geq  \frac{1}{2}\cdot \frac{1}{2}\cdot (\hat{F}_{2,h(a)}-\gamma_1)/(\phi k) - \frac{16}{\eta}\cdot \left(\frac{2(\eta/16)^2}{100\phi k}\cdot 2\cdot (\hat{F}_{2,h(a)}+\gamma_1)+2\gamma_2^2\right)\notag\\
&= \left(\frac{1}{4}-\frac{\eta/16}{25}\right)\cdot \frac{\hat{F}_{2,h(a)}}{\phi k} - \left(\frac{1}{4}+\frac{\eta/16}{25}\right)\cdot \frac{\gamma_1}{\phi k} - \frac{32}{\eta} \cdot \gamma_2^2\notag\\
&\geq \frac{1}{5}\cdot \frac{\hat{F}_{2,h(a)}}{\phi k}-\frac{1}{3}\cdot \frac{\gamma_1}{\phi k} - \frac{32}{\eta^2}\cdot \gamma_2^2\notag\\
&\geq 2\cdot \left(\frac{\hat{F}_{2,h(a)}+\gamma_1}{25 \phi k}+\frac{512\gamma_2^2}{\eta^2}\right)-\left(\frac{2048\gamma_2^2}{\eta^2}+\frac{2}{3}\cdot \frac{\gamma_1}{\phi k}-\frac{3}{25\phi k}\cdot \hat{F}_{2,h(a)}\right)\label{eq:fhat_lb1}
\end{align}
On the other hand, by convexity, we have:
\begin{align*}
\hat{f}_a^2&\geq (1-\eta/16)f_a^2-\frac{16}{\eta}\cdot \left(\frac{\eta/16}{10\sqrt{\phi k}}\|\mathcal{S}_{h(a)}\|_2+\gamma_2\right)^2\\
&\geq (1-\eta/16)f_a^2-\frac{16}{\eta}\cdot \left(\frac{2(\eta/16)^2}{100\phi k}\|\mathcal{S}_{h(a)}\|_2^2+2\gamma_2^2\right)
\end{align*}
and thus
\begin{align}
\hat{f}_a^2&\geq \frac{1}{2}\cdot 16\cdot (\gamma_1/(\phi k)+512\gamma_2^2/\eta^2) - \frac{16}{\eta}\cdot \left(\frac{2(\eta/16)^2}{100\phi k}\cdot 2\cdot (\hat{F}_{2,h(a)}+\gamma_1)+2\gamma_2^2\right)\notag\\
&= \left(8 - \frac{\eta/16}{25k}\right)\cdot \frac{\gamma_1}{\phi k} + \left(\frac{4096}{\eta^2}-\frac{32}{\eta}\right)\gamma_2^2 - \frac{\eta/16}{25\phi k}\cdot \hat{F}_{2,h(a)}\notag\\
&\geq \frac{2}{3}\cdot\frac{\gamma_1}{\phi k} + \frac{2048}{\eta^2}\cdot\gamma_2^2 - \frac{3}{25\phi k}\cdot \hat{F}_{2,h(a)}\label{eq:fhat_lb2}
\end{align}
By looking at Equation~\eqref{eq:fhat_lb1} $+$ Equation~\eqref{eq:fhat_lb2}, we have $\hat{f}_a^2\geq \frac{\hat{F}_{2,h(a)}+\gamma_1}{25\phi k}+\frac{512\gamma_2^2}{\eta^2}$.
Thus, $a\in \hat{H}$.
\end{proof}

\begin{lemma}\label{lem:top_k_operation}
At any timestamp $t$, if $f_a^p\geq \|\mathcal{S}\|_p^p/k$ and $a\in\hat{H}$, then $a\in H$.
\end{lemma}
\begin{proof}
We prove the statement by contradiction.
Suppose $a\not\in H$, there is a subset $Q\subseteq H$ with $|Q|\geq \left(\frac{1+\eta}{1-\eta}\right)^p\cdot k$ such that $\forall b\in Q,\hat{f}_b\geq \hat{f}_a$.
According to Lemma~\ref{lem:heavyhitter_goodapprox}, we have $\forall b\in Q, f_b^p \geq \left(\frac{1-\eta}{1+\eta}\right)^p\cdot f_a^p\geq \left(\frac{1-\eta}{1+\eta}\right)^p\cdot \|\mathcal{S}\|_p^p/k$.
Then we have $\sum_{b\in Q}f_b^p\geq \|\mathcal{S}\|_p^p$ which leads to a contradiction.
\end{proof}

\begin{theorem}[$\ell_p$ Heavy hitters for all $p\in[0,\infty)$]\label{thm:lp_heavyhitters}
Let $\varepsilon>0,\eta\in(0,0.5),k\geq 1,\xi \in(0,0.5)$.
Let $\phi = \max(1,|\mathcal{U}|^{1-2/p})$.
There is an $\varepsilon$-DP algorithm in the streaming continual release model such that with probability at least $1-\xi$, it always outputs a set $H\subseteq \mathcal{U}$ and a function $\hat{f}: H\rightarrow \mathbb{R}$ for every timestamp $t\in [T]$:
such that
\begin{enumerate}
    \item $\forall a\in H$, $\hat{f}(a)\in (1\pm \eta)\cdot f_a$ where $f_a$ is the frequency of $a$ in the stream $\mathcal{S}=(a_1,a_2,\cdots,a_t)$,
    \item $\forall a\in \mathcal{U},$ if $f_a\geq \frac{1}{\varepsilon\eta}\cdot \log^C\left(\frac{T\cdot k\cdot |\mathcal{U}|}{\xi\eta}\right)$ for some sufficiently large constant $C>0$ and $f_a^p\geq \|\mathcal{S}\|_p^p/k$ then $a\in H$,
    \item The size of $H$ is at most $O\left((\log(T/\xi)+\log(|\mathcal{U}|))\cdot\left(\frac{1+\eta}{1-\eta}\right)^p\cdot k \right)$.
\end{enumerate}
The algorithm uses $\frac{\phi k^3}{\eta^2}\cdot \poly\left(\log\left(\frac{T\cdot k\cdot |\mathcal{U}|}{\xi}\right)\right)$ space.
\end{theorem}
\begin{proof}
The first property follows from Lemma~\ref{lem:heavyhitter_goodapprox}.

According to Lemma~\ref{lem:acc_heavy_hitter} and Lemma~\ref{lem:top_k_operation}, the guarantee holds with probability $0.9$ for any particular $t\in[T]$ and $a\in \mathcal{U}$.
To boost the probability to make the guarantee holds simultaneously for all $t\in [T]$ and $a\in\mathcal{U}$ with probability at least $1-\xi/2$, we need to repeat Algorithm~\ref{alg:lp_heavyhitters} $\lceil50(\log(2T/\xi)+\log(|\mathcal{U}|))\rceil$ times, and let final $H$ at timestamp $t$ be the union of all output $H$ at timestamp $t$, and let final $\hat{f}(a)$ at timestamp $t$ be any output $\hat{f}(a)$ at timestamp $t$.
According to Lemma~\ref{lem:lp_heavyhitter_size_output}, the output $H$ of a single running copy of Algorithm~\ref{alg:lp_heavyhitters} is at most $\left(\frac{1+\eta}{1-\eta}\right)^p\cdot k$.
Thus, the size of the final output $H$ is at most $O\left((\log(T/\xi)+\log(|\mathcal{U}|))\cdot\left(\frac{1+\eta}{1-\eta}\right)^p\cdot k \right)$ which proves the second property.

We apply Theorem~\ref{thm:count_sketch} for the frequency and $\ell_2$ frequency moment estimators of $\mathcal{S}_i$ for each $i\in [m]$.
Since we run $\lceil50(\log(2T/\xi)+\log(|\mathcal{U}|))  \rceil$ copies of Algorithm~\ref{alg:lp_heavyhitters}, if we desire $\varepsilon$-DP algorithm in the end, we need each frequency and $\ell_2$ frequency moment subroutine to be $\varepsilon / (4\cdot \lceil50(\log(2T/\xi)+\log(|\mathcal{U}|))  \rceil)$-DP according to Lemma~\ref{lem:lp_heavyhitter_dp_guarantee}.
To simultaneously make the call of each run of the frequency and $\ell_2$ frequency moment subroutine succeeds with probability at least $1-\xi/2$, we need to apply union bound over all calls of the subroutines and thus each run of the subroutine should succeed with probability at least $1-\xi/(4\cdot \lceil 50(\log(T/\xi)+\log(|\mathcal{U}|))\rceil\cdot m)=1-\xi/(4\cdot \lceil 50(\log(T/\xi)+\log(|\mathcal{U}|))\rceil\cdot 10k^2)$.
Notice that we also need to re-scale $\eta$ in Theorem~\ref{thm:count_sketch} to be $(\eta/16)/(10\sqrt{\phi k})$ used in Algorithm~\ref{alg:lp_heavyhitters} for estimation of the frequency of each element and set $\eta = 0.01$ in Theorem~\ref{thm:count_sketch} for the estimation of the $\ell_2$ frequency moment.
Thus, according to Theorem~\ref{thm:count_sketch}, we have 
\begin{align*}
\gamma_1 = \frac{1}{\varepsilon^2}\cdot \poly\left(\log\left(\frac{T\cdot k\cdot |\mathcal{U}|}{\xi}\right)\right)
\end{align*}
and 
\begin{align*}
\gamma_2 = \frac{1}{\varepsilon}\cdot \poly\left(\log\left(\frac{T\cdot k\cdot |\mathcal{U}|}{\xi\eta}\right)\right)
\end{align*}
Therefore, the second property follows from Lemma~\ref{lem:acc_heavy_hitter} and Lemma~\ref{lem:top_k_operation} and our probability boosting argument.

Finally, let us consider the space usage.
Since we run $O(\log(T/\xi)+\log(|\mathcal{U}|))$ copies of Algorithm~\ref{alg:lp_heavyhitters}, and each copy calls $O(k^2)$ frequency estimator and $\ell_2$ frequency moment estimator using the parameters discussed above.
Due to Theorem~\ref{thm:count_sketch}, the total space usage is at most $\frac{\phi k^3}{\eta^2}\cdot \poly\left(\log\left(\frac{T\cdot k\cdot |\mathcal{U}|}{\xi}\right)\right)$.
\end{proof}

\subsection{Differentially Private Continual Released Counting of Low Frequency Elements}\label{sec:low_freq_elements}
In this section, we show a differentially private continual released algorithm for counting the number of elements that have a certain (low) frequency.
Similar to our counting distinct elements algorithm, we first consider the case where the universe of the elements is small.

\subsubsection{Number of Low Frequency Elements for Small Universe}
\begin{algorithm}[H]\label{alg:low_freq_small_universe}
\small
	\KwIn{A stream $\mathcal{S}$ of elements $a_1,a_2,\cdots,a_T\in \mathcal{U}\cup \{\perp\}$ with gaurantee that $|\mathcal{U}|\leq m$, and a target frequency $k$.}
	\Parameter{Relative approximation factor $\alpha\geq 1$ and additive approximation factor $\gamma\geq 0$ depending on the streaming continual release summing algorithm.\\ \hfill{//See Theorem~\ref{thm:dp_summing}}. }
	\KwOut{Estimation of the number of elements with frequency exactly $i$ for each $i\in [k]$ at every timestamp $t$.}
	Initialize empty streams $\mathcal{C}_1,\mathcal{C}_2,\cdots,\mathcal{C}_k$.\\
	For each $a\in \mathcal{U}$, initialize frequency $f(a)\gets 0$\\
	\For {each $a_t$ in the stream $\mathcal{S}$}{
	    If $a_t\not=\perp$, $f(a_t)\gets f(a_t) + 1$. \\
	    \For { each $i\in [k]$}{
	        \If{$a_t\not=\perp$ and $f(a_t) = i + 1$}
	        {
	            Append $-1$ at the end of $\mathcal{C}_i$.
	        }
	        \ElseIf{$a_t\not=\perp$ and $f(a_t)=i$}
	        {
	            Append $1$ at the end of $\mathcal{C}_i$.
	        }
	        \Else{
	            Append $0$ at the end of $\mathcal{C}_i$.
	        }
	    }
	    For each $i\in[k]$, output $\hat{s}_i$ which is an $(\alpha,\gamma)$-approximation to the total counts of $\mathcal{C}_i$.
	}
    \caption{Number of Low Frequency Elements for Small Universe}
\end{algorithm}

\begin{lemma}
At the end of any time $t\in[T]$, $\forall i\in[k]$, the output $\hat{s}_i$ of Algorithm~\ref{alg:low_freq_small_universe} is an $(\alpha,\gamma)$-approximation to  $|\{a\in\mathcal{U}~|~f_a=i\}|,$ the size of the set of elements of which the frequency is exact $i$.
\end{lemma}
\begin{proof}
It is easy to observe that Algorithm~\ref{alg:low_freq_small_universe} maintains $f(a)$ such that $f(a)=f_a$ for any $a\in\mathcal{U}$.
For $i\in[k]$, if before adding $a_t$ the frequency of $a_t$ is $i$, we append $-1$ to $\mathcal{C}_i$. 
If after adding $a_t$ the frequency of $a_t$ becomes $i$, we append $1$ to $\mathcal{C}_i$.
Note that if the frequency of $a_t$ has the value which is not in above two cases, it does not affect the number of elements with frequency $i$.
Therefore, the total counts of $\mathcal{C}_i$ is always the number of elements of which frequency is $i$.
Thus, $\hat{s}_i$ is an $(\alpha,\gamma)$-approximation to the number of elements of which frequency is $i$.
\end{proof}

\begin{lemma}\label{lem:low_freq_dp_guarantee}
If the algorithm that continually release the approximate total counts of $\mathcal{C}_i$ for every  $i\in[k]$ is $\varepsilon$-DP, Algorithm~\ref{alg:low_freq_small_universe} is $(8k\varepsilon)$-DP in the continual release model.
\end{lemma}
\begin{proof}
Consider two neighboring stream $\mathcal{S}=(a_1,a_2,\cdots,a_T)$ and $\mathcal{S}'=(a_1',a_2',\cdots,a'_T)$ of elements in $\mathcal{U}$ where they only differ at timestamp $t$, i.e., $a_t\not=a'_t$.
Fix any $i\in[k]$, let us consider the differences between the corresponding count streams $\mathcal{C}_i=(c_1,c_2,\cdots,c_T)$ and $\mathcal{C}'_i=(c'_1,c'_2,\cdots,c'_T)$.

Consider an intermediate neighboring stream $\mathcal{S}'' = (a_1,a_2,\cdots,a_{t-1}, \perp,a_{t+1},\cdots,a_T)$. 
Let $\mathcal{C}''_i=(c''_1,c''_2,\cdots,c''_T)$ be the count stream corresponding to $\mathcal{S}''$.
The total difference between $\mathcal{C}_i$ and $\mathcal{C}'_i$ is bounded by the sum of the total difference between $\mathcal{C}_i$ and $\mathcal{C}''$ and the total difference between $\mathcal{C}'_i$ and $\mathcal{C}''_i$.

Consider the difference between $\mathcal{C}_i$ and $\mathcal{C}''_i$. 
If $a_t=\perp$, $\mathcal{C}''_i$ is exactly the same as $\mathcal{C}_i$.
Suppose $a_t=a\in\mathcal{U}$ is the $j$-th appearance of $a$ in $\mathcal{S}$.
If $j>i$, change $a_t$ to $\perp$ does not affect $\mathcal{C}_i$.
Suppose $j\leq i$.
Let $a_{t_1},a_{t_2},a_{t_3}$ be the $i$-th, $(i+1)$-th, $(i+2)$-th appearances of $a$ in $\mathcal{S}$ respectively.
Then it is easy to verify that $c_{t_1}=1,c''_{t_1}=0,c_{t_2}=-1,c''_{t_2}=1,c_{t_3}=0,c''_{t_3}=-1$, and for any other $p\not=t_1,t_2,t_3$, we have $c_{p}=c''_{p}$.
Thus, the total difference between $\mathcal{C}_i$ and $\mathcal{C}_i''$ is at most $4$.

Similarly, the total difference between $\mathcal{C}'_i$ and $\mathcal{C}_i''$ is at most $4$.
Thus, the total difference between $\mathcal{C}_i$ and $\mathcal{C}'_i$ is at most $8$.
Therefore, the total sensitivity of $\mathcal{C}_1,\mathcal{C}_2,\cdots,\mathcal{C}_k$ is at most $8k$.
If we use an $\varepsilon$-DP algorithm to continually release the total counts of $\mathcal{C}_i$ for every $i\in[k]$, the continually released outputs of Algorithm~\ref{alg:low_freq_small_universe} is $(8k\varepsilon)$-DP.
\end{proof}

\begin{theorem}[Streaming continual release count of low frequency elements for small universe]\label{thm:low_freq_small_universe}
Let $k\geq 1,\varepsilon\geq 0,\xi\in (0,0.5)$.
Suppose the universe $\mathcal{U}$ has size at most $m$.
There is an $\varepsilon$-DP algorithm in the streaming continual release model such that with probability at least $1-\xi$, it always outputs $k$ numbers $\hat{s}_1,\hat{s}_2,\cdots,\hat{s}_k$ for every timestamp $t$, such that $\forall i\in [k], \hat{s}_i$ is an approximation to $|\{a\in\mathcal{U}\mid f_a = i \}|$ with additive error $O\left(\frac{k}{\varepsilon}\cdot\log^{2.5}(T)\log\left(\frac{k}{\xi}\right)\right)$.
The algorithm uses $O(m+k\log(T))$ space.
\end{theorem}
\begin{proof}
Consider Algorithm~\ref{alg:low_freq_small_universe}, we use Theorem~\ref{thm:dp_summing} as the summing subroutine.
According to Lemma~\ref{lem:low_freq_dp_guarantee}, if we want the final algorithm to be $\varepsilon$-DP, then each subroutine must be $(\varepsilon/(8k))$-DP.
Furthermore, if we want the over success probability to be $1-\xi$, the success probability of each summing subroutine should be at least $1-\xi/k$.
By applying Theorem~\ref{thm:dp_summing}, we have $\alpha=1$ and $\gamma = O\left(\frac{k}{\varepsilon}\cdot\log^{2.5}(T)\log\left(\frac{k}{\xi}\right)\right)$.

We need $O(m)$ space to maintain the frequency of the elements.
We need $O(k\log(T))$ space to run $k$ summing subroutines.
Therefore, the total space usage is $O(m+k\log(T))$.
\end{proof}

\subsubsection{Number of Low Frequency Elements for General Universe}

\begin{algorithm}[H]\label{alg:low_freq_large_universe}
\small
	\KwIn{A stream $\mathcal{S}$ of elements $a_1,a_2,\cdots,a_T\in \mathcal{U}\cup \{\perp\}$, an error parameter $\eta\in (0,0.5)$, and a parameter $k\geq 1$.}
		\Parameter{Additive approximation factor $\gamma_1$ depending on the streaming continual release number of distinct elements, relative approximation factor $\alpha_2\geq 1$ and additive approximation factor $\gamma_2\geq 0$ depending on the streaming continual release count of low frequency elements for small universe.\\ \hfill{//See Corollary~\ref{cor:distinct_large_universe} and Theorem~\ref{thm:low_freq_small_universe}}. }
	\KwOut{Estimation of the number of elements with frequency exactly $i$ for each $i\in[k]$ at every timestamp $t$.}
	$L\gets\lceil \log \min(|\mathcal{U}|,T) \rceil,\lambda\gets 2\log (1000 k), m\gets 100\cdot(25600\lambda/\eta^2)^2$.\\
	Let $h:\mathcal{U}\rightarrow [m]$ be a pairwise independent hash function.\\
	Let $g:\mathcal{U}\rightarrow [L]\cup \{\perp\}$ be a $\lambda$-wise independent hash function and $\forall a\in \mathcal{U},i\in[L],\Pr[g(a)=i]=2^{-i},\Pr[g(a)=\perp]=2^{-L}$.\\
	Initialize empty streams $\mathcal{S}_1,\mathcal{S}_2,\cdots,\mathcal{S}_L$.\\
	\For{ each $a_t$ in the stream $\mathcal{S}$}
	{
	    \For {$ i \in L$}{
	    \If{$a_t\not=\perp$ and $g(a_t)=i$}{
	        Append $h(a_t)$ to the end of the stream $\mathcal{S}_{i}$.
	    }
	    \Else{
	        Append $\perp$ to the end of the stream $\mathcal{S}_i$.
	    }
	    }
	    Compute $\hat{d}$ which is a $(1.1,\gamma_1)$-approximation to the number of distinct elements in $\mathcal{S}$.\\
	    For $i\in [L]$, compute $\hat{s}_{i,1},\hat{s}_{i,2},\cdots,\hat{s}_{i,k}$ where $\hat{s}_{i,j}$ is an $(\alpha_2,\gamma_2)$-approximation to the number of elements in $\mathcal{S}_i$ where each element has frequency exact $j$.\\
	    If $\hat{d}\leq \max(3\gamma_1,64\lambda/\eta^2)$, set $\hat{s}_1=\hat{s}_2=\cdots=\hat{s}_k=0$.\\
	    Otherwise, find the largest $i^*\in[L]$ such that $2^{i^*}\cdot (64\lambda/\eta^2) \leq \hat{d}$, and $\forall j\in[k]$, set $\hat{s}_j=\hat{s}_{i^*,j}\cdot 2^i$.\\
	    Output $\hat{s}_1,\hat{s}_2,\cdots,\hat{s}_k$.
	}
    \caption{Number of Low Frequency Elements via Subsampling}
\end{algorithm}

\begin{lemma}\label{lem:dp_guarantee_low_freq_large_universe}
If the subroutine of continually releasing $\hat{d}$ in Algorithm~\ref{alg:low_freq_large_universe} is $\varepsilon$-DP and the subroutine of continually releasing $\{\hat{s}_{i,1},\hat{s}_{i,2},\cdots,\hat{s}_{i,k}\}$ for each $i\in[k]$ is also $\varepsilon$-DP, Algorithm~\ref{alg:low_freq_large_universe} is $3\varepsilon$-DP.
\end{lemma}
\begin{proof}
Consider two neighboring streams $\mathcal{S}$ and $\mathcal{S}'$ where the only difference is the $t$-th element, i.e., $a_t\not=a_t'$.
Consider their corresponding streams $\mathcal{S}_1,\mathcal{S}_2,\cdots,\mathcal{S}_L$ and $\mathcal{S}'_1,\mathcal{S}'_2,\cdots,\mathcal{S}'_L$ in Algorithm~\ref{alg:low_freq_large_universe}.
For $i\in[L]$ and $j\not = t$, the $j$-th element in $\mathcal{S}_i$ should be the same as the $j$-th element in $\mathcal{S}'_i$.
Since $a_t$ can only make the $t$-th element of $\mathcal{S}_{h(a_t)}$ be different from the $t$-th element of $\mathcal{S}'_{h(a_t)}$, and $a'_t$ can only make the $t$-th element of $\mathcal{S}_{h(a'_t)}$ be different from the $t$-th element of $\mathcal{S}'_{h(a'_t)}$, the total sensitivity of $\{\mathcal{S}_1,\mathcal{S}_2,\cdots,\mathcal{S}_L\}$  is at most $2$.
Thus, if the continual release of $\hat{d}$ is $\varepsilon$-DP and for every $i\in[L]$, the continual release of $\hat{s}_{i,1},\hat{s}_{i,2},\cdots,\hat{s}_{i,k}$ is also $\varepsilon$-DP, the continual release of $\hat{s}_1,\hat{s}_2,\cdots,\hat{s}_k$ is $(3\varepsilon)$-DP.
\end{proof}

\begin{lemma}\label{lem:accuracy_low_freq_members}
Consider an arbitrary timestamp $t\in[T]$. 
$\forall a\in\mathcal{U}$, let $f_a$ denote the frequency of $a$ in $a_1,a_2,\cdots,a_t$.
With probability at least $0.9$, $\forall j\in [k]$, the output $\hat{s}_j$ of Algorithm~\ref{alg:low_freq_large_universe} is a $\left(\alpha_2,\left(\alpha_2\eta+\frac{\eta^2\gamma_2}{32\lambda}\right)\cdot \|\mathcal{S}\|_0+6\gamma_1+128\lambda/\eta^2\right)$-approximation to $\left|\{a\in\mathcal{U}\mid f_a = j\}\right|$, where $\|\mathcal{S}\|_0$ is the number of distinct elements in $a_1,a_2,\cdots,a_t$.
\end{lemma}

To prove Lemma~\ref{lem:accuracy_low_freq_members}, we need following intermediate statements.
Consider timestamp $t\in [T]$.
$\forall a\in\mathcal{U}$, let $f_a$ denote the frequency of $a$ in $a_1,a_2,\cdots,a_t$.
If $i^*$ exists in Algorithm~\ref{alg:low_freq_large_universe}, we define $G=\{a_j\mid g(a_j)=i^*,j\in[t]\}$.
Let $\|\mathcal{S}\|_0$ denote the number of distinct elements in $a_1,a_2,\cdots,a_t$.
For $l\in [k],$ define $G_l=\{a\in G\mid f_a = l\}$.

\begin{claim}\label{cla:collision_low_freq}
With probability at least $0.98$, $\forall a\not=a'\in G,h(a)\not=h(a')$.
\end{claim}
\begin{proof}
Since $E\left[|G|\right]=\|\mathcal{S}\|_0/2^{i^*}$, we have $|G|\leq 100\cdot \|\mathcal{S}\|_0/2^{i^*}$ with probability at least $0.99$ by Markov's inequality.
Since $i^*$ can be found by Algorithm~\ref{alg:low_freq_large_universe}, we have $\|\mathcal{S}\|_0\leq 1.1\cdot  (\hat{d}+\gamma_1)\leq 2\hat{d}$.
By the choice of $i^*$, we have $2^{i^*}\geq \hat{d}/(128\lambda/\eta^2)$.
Therefore, $\|\mathcal{S}\|_0/2^{i^*}\leq 256\lambda/\eta^2$.
With probability at least $0.99$, $|G|\leq 25600\lambda/\eta^2$.
By union bound and Markov's inequality, $\Pr[\forall a\not=a'\in G,h(a)\not=h(a') \mid |G|\leq 25600\lambda/\eta^2]\geq 1-|G|^2/m\geq 0.99$.
\end{proof}

\begin{claim}\label{cla:approx_low_freq}
$\forall l\in [k], \Pr[2^{i^*}\cdot|G_l|\in |\{a\in U\mid f_a=l\}|\pm \eta\|\mathcal{S}\|_0]\geq 1-0.01/k$.
\end{claim}
\begin{proof}
Let $s_l=|\{a\in U\mid f_a=l\}|$.
Since $i^*$ can be found, we have $\hat{d}\geq 3\gamma_1$.
Note that $s_l\leq \|\mathcal{S}\|_0\leq 1.1(\hat{d}+\gamma_1)\leq 2\hat{d}\leq 2^{i^*}\cdot 4\cdot 64\lambda/\eta^2$.
By applying Lemma~\ref{lem:concentration}, $\forall l\in [k]$ we have:
\begin{align*}
&\Pr\left[\left||G_l|-s_l/2^{i^*}\right|>32\lambda/\eta\right]\\
\leq & 8\cdot \left(\frac{\lambda\cdot s_l/2^{i^*}+\lambda^2}{(32\lambda/\eta)^2}\right)^{\lambda/2}\\
\leq & 0.01/k,
\end{align*}
where the last inequality follows from $s_l/2^{i^*}\leq 4\cdot 64\lambda/\eta^2$ and $\lambda = 2\cdot \log(1000k)$.
Since $\|\mathcal{S}\|_0\geq (\hat{d}-\gamma_1)/1.1\geq \hat{d}/2\geq 2^{i^*}\cdot 32\lambda/\eta^2$, we have 
\begin{align*}
\Pr\left[\left||G_l|-s_l/2^{i^*}\right|>\eta \|\mathcal{S}\|_0/2^{i^*}\right]\leq 0.01/k.
\end{align*}
\end{proof}
Next we show the proof of Lemma~\ref{lem:accuracy_low_freq_members}
\begin{proof}[Proof of Lemma~\ref{lem:accuracy_low_freq_members}]
Let $\mathcal{E}$ denote the event that $\forall a\not=a'\in G,h(a)\not=h(a')$.
Let $\forall l\in[k], s_l=|\{a\in U\mid f_a=l\}|$.
Let $\mathcal{E}'$ denote the event that $\forall l\in [k],2^{i^*}\cdot|G_l|\in s_l\pm \eta\|\mathcal{S}\|_0$.
According to Claim~\ref{cla:collision_low_freq} and Claim~\ref{cla:approx_low_freq}, the probability that both $\mathcal{E}$ and $\mathcal{E}'$ happen is at least $0.97$. 
In the remaining of the proof, we condition on both events $\mathcal{E}$ and $\mathcal{E}'$.

First, consider the case that $\hat{d}<\max(3\gamma_1,64\lambda/\eta^2)$.
Since $\hat{d}$ is a $(1.1,\gamma_1)$-approximation to $\|\mathcal{S}\|_0$, we have $\|\mathcal{S}\|_0\leq 6\gamma_1+128\lambda/\eta^2$.
In this case, $\forall l\in[k],\hat{s}_l=0$ and $\left|\{a\in\mathcal{U}\mid f_a = l\}\right|\leq \|\mathcal{S}\|_0\leq 6\gamma_1+128\lambda/\eta^2$.

Next, consider the case that $\hat{d}\geq \max(3\gamma_1,64\lambda/\eta^2)$.
For each $j\in [k]$, let $s_{i^*,l}$ denote the number of elements in $\mathcal{S}_l$ where each element has frequency exact $j$.
According to event $\mathcal{E}$, we have $\forall l\in[k],s_{i^*,l}=|G_{l}|$.
According to event $\mathcal{E}'$, $\forall l\in[k]$, we have:
\begin{align*}
\hat{s}_l&\geq 2^{i^*}\cdot \left(\frac{1}{\alpha_2}\cdot |G_l| - \gamma_2\right)\geq \frac{1}{\alpha_2}\cdot \left(s_l-\eta\|\mathcal{S}\|_0\right) - 2^{i^*}\cdot \gamma_2\\
&\geq \frac{1}{\alpha_2}\cdot s_l -   \left(\eta+\frac{\eta^2\gamma_2}{32\lambda}\right)\cdot \|\mathcal{S}\|_0,
\end{align*}
where the first inequality follows from that $\hat{s}_{i^*,l}$ is an $(\alpha_2,\gamma_2)$-approximation to $s_{i^*,l}$ and $s_{i^*,l}=|G_l|$, the second inequality follows from event $\mathcal{E}'$, and the third inequality follows from $\alpha_2\geq 1$ and $ \|\mathcal{S}\|_0\geq (\hat{d}-\gamma_1)/1.1\geq \hat{d}/2\geq 2^{i^*}\cdot 32\lambda/\eta^2$.
Similarly, $\forall l\in[k]$, we have:
\begin{align*}
\hat{s}_l&\leq 2^{i^*}\cdot \left(\alpha_2\cdot |G_l|+\gamma_2\right)
\leq \alpha_2\cdot (s_l+\eta\|\mathcal{S}\|_0) + 2^{i^*}\cdot \gamma_2\\
&\leq \alpha_2\cdot s_l +   \left(\alpha_2\eta+\frac{\eta^2\gamma_2}{32\lambda}\right)\cdot \|\mathcal{S}\|_0.
\end{align*}
\end{proof}

\begin{theorem}[Streaming continual release of count of low frequency elements]\label{thm:count_low_freq_large}
Let $k\geq 1,\varepsilon\geq 0,\xi\in(0,0.5),\eta\in(0,0.5)$.
There is an $\varepsilon$-DP algorithm in the streaming continual release model such that with probability at least $1-\xi$, it always outputs $k$ numbers $\hat{s}_1,\hat{s}_2,\cdots,\hat{s}_k$ for every timestamp $t$ such that $\forall i\in [k]$, $\hat{s}_i$ is an approximation to $|\{a\in \mathcal{U} \mid f_a = i\}|$ with additive error:
\begin{align*}
\left(\eta+\eta^2\cdot \frac{k}{\varepsilon}\cdot \poly\left(\log\left(\frac{Tk}{\xi}\right)\right)\right)\cdot \|\mathcal{S}\|_0+\frac{1}{\varepsilon} \cdot \poly\left(\log\left(\frac{T}{\xi}\right)\right)+O\left(\frac{\log k}{\eta^2}\right)
\end{align*}
The algorithm uses $\frac{1}{\eta^4}\cdot \poly\left(\frac{\log(T\cdot k/\xi)}{\min(\varepsilon,1)}\right)$ space.
\end{theorem}
\begin{proof}
According to Lemma~\ref{lem:accuracy_low_freq_members}, the approximation guarantee holds with probability at least $0.9$.
To boost the success probability to $1-\xi/3$ for all $t\in [T]$, we need to run $\lceil 50\log(3T/\xi)\rceil$ independent copies of Algorithm~\ref{alg:low_freq_large_universe}.
Since we run $\lceil 50\log(3T/\xi)\rceil$ copies of Algorithm~\ref{alg:low_freq_large_universe} and according to Lemma~\ref{lem:dp_guarantee_low_freq_large_universe}, if we want the final algorithm to be $\varepsilon$-DP, we need each subroutine of the streaming continual release of $\{\hat{s}_{i,1},\hat{s}_{i,2},\cdots,\hat{s}_{i,k}\}$ for each $i\in [k]$ to be $(\varepsilon/(3\cdot \lceil50\log(2T/\xi))\rceil)$-DP and we need the subroutine of the streaming continual release of $\hat{d}$ to be also $(\varepsilon/(3\cdot \lceil50\log(3T/\xi))\rceil)$-DP.
To simultaneously make the call of each subroutine of the streaming continual release of $\{\hat{s}_{i,1},\hat{s}_{i,2},\cdots,\hat{s}_{i,k}\}$ over all independent copies of Algorithm~\ref{alg:low_freq_large_universe} satisfy the desired $(\alpha_2,\gamma_2)$-approximation with probability at least $1-\xi/3$, we need to make $\{\hat{s}_{i,1},\hat{s}_{i,2},\cdots,\hat{s}_{i,k}\}$ satisfy the approximation guarantee for each particular $i\in[L]$ and a particular copy of Algorithm~\ref{alg:low_freq_large_universe} with probability at least $1-\xi/(3\cdot \lceil50\log(3T/\xi))\rceil\cdot L)$.
To simultaneously make the call of each subroutine of the streaming continual release of $\hat{d}$ over all independent copies of Algorithm~\ref{alg:low_freq_large_universe} satisfy the desired $(1.1,\gamma_1)$-approximation with probability at least $1-\xi/3$, we need to make $\hat{d}$ satisfy the approximation guarantee for each particular copy of Algorithm~\ref{alg:low_freq_large_universe} with probability at least $1-\xi/(3\cdot \lceil50\log(3T/\xi))\rceil)$.
Then, according to Corollary~\ref{cor:distinct_large_universe}, we have 
$
\gamma_1 = \frac{1}{\varepsilon}\cdot \poly\left(\log\left(\frac{T}{\xi}\right)\right).
$
According to Theorem~\ref{thm:low_freq_small_universe}, we have $\alpha_2=1$ and 
$
\gamma_2= \frac{k}{\varepsilon}\cdot \poly\left(\log\left(\frac{Tk
}{\xi}\right)\right).
$
By plugging above parameters into Lemma~\ref{lem:accuracy_low_freq_members}, we have $\forall i\in [k],\hat{s}_i$ is an approximation to $|\{a\in\mathcal{U}\mid f_a=i\}|$ with additive error:
\begin{align*}
\left(\eta+\eta^2\cdot \frac{k}{\varepsilon}\cdot \poly\left(\log\left(\frac{Tk}{\xi}\right)\right)\right)\cdot \|\mathcal{S}\|_0+\frac{1}{\varepsilon} \cdot \poly\left(\log\left(\frac{T}{\xi}\right)\right)+O\left(\frac{\log k}{\eta^2}\right).
\end{align*}

Next, consider the space usage. 
According to Corollary~\ref{cor:distinct_large_universe}, the total space needed to release $\hat{d}$ for all copies of Algorithm~\ref{alg:low_freq_large_universe} needs $\poly\left(\frac{\log(T/\xi)}{\min(\varepsilon,1)}\right)$.
According to Theorem~\ref{thm:low_freq_small_universe}, the total space needed to release $\{\hat{s}_{i,1},\hat{s}_{i,2},\cdots,\hat{s}_{i,k}\}$ for all $i\in [L]$ over all copies of Algorithm~\ref{alg:low_freq_large_universe} needs $O(\log(T/\xi)\cdot L \cdot (m+\log(T)))=\frac{1}{\eta^4}\cdot \poly\left(\log\left(\frac{T\cdot k}{\xi}\right)\right)$.
Thus, the overall space is at most $\frac{1}{\eta^4}\cdot \poly\left(\frac{\log(T\cdot k/\xi)}{\min(\varepsilon,1)}\right)$.
\end{proof}

\subsection{$\ell_p$ Moment Estimation}\label{sec:lp_moment}
In this section, we show how to use our $\ell_p$ heavy hitters and the estimator of number of low frequency elements to solve $\ell_p$ moment estimation problem.

\begin{algorithm}[H]\label{alg:lp_moment}
\small
	\KwIn{A stream $\mathcal{S}$ of elements $a_1,a_2,\cdots,a_T\in \mathcal{U}\cup \{\perp\}$, an error parameter $\eta\in (0,0.5)$.}
	\Parameter{A threshold parameter $\tau$ depending on the heavy hitter algorithm, a relative approximation factor $\alpha$ and an additive approximation factor $\gamma$ depending on the algorithm of estimating count of low frequency elements. \hfill{//See Theorem~\ref{thm:lp_heavyhitters}} and Theorem~\ref{thm:count_low_freq_large}.}
	\KwOut{Estimation of the $\ell_p$ frequency moment at every timestamp $t$.}
	Let $\beta'$ be drawn uniformly at random from $[1/2,1]$ and let $\beta\in \beta'\pm (\eta/T)^C$ for some sufficiently large constant $C>0$. \hfill{//Thus $\beta$ can be represented by $\Theta(\log(T/\eta))$ bits.}\\
	Let $q^*_1$ be the smallest integer that $\beta(1+\eta)^{q^*_1}>\tau $ and let $q^*_2$ be the smallest integer that $\beta(1+\eta)^{q^*_2+1}\geq T$, and for any $q\in[q^*_1,q^*_2],$ define the interval $I_q=(\beta(1+\eta)^q, \beta(1+\eta)^{q+1}]$.\\
	$\lceil L\gets \log(|\mathcal{U}|) \rceil,\lambda\gets 2\cdot \log (1000 (L+1)\log(4T)/\eta)$.\\
	Let $k$ be the largest integer such that $k\leq  \beta(1+\eta)^{q^*_1}$. Let $B\gets \left(\frac{\log(4T)}{\eta}\right)\cdot 100(L+1)\cdot \frac{32\lambda}{\eta^3}\cdot (1+\eta)^p$.\\
	Let $g:\mathcal{U}\rightarrow [L]\cup \{\perp\}$ be a $\lambda$-wise independent hash function and $\forall a\in \mathcal{U},i\in[L],\Pr[g(a)=i]=2^{-i},\Pr[g(a)=\perp]=2^{-L}$.\\
	Initialize empty streams $\mathcal{S}_0,\mathcal{S}_1,\mathcal{S}_2,\cdots,\mathcal{S}_L$.\\
	\For{ each $a_t$ in the stream $\mathcal{S}$}
	{
	    \If{ $a_t\not=\perp$ }
	    {
	        Append $a_t$ at the end of $\mathcal{S}_0$.\\
	        For $i\in[L]$, if $g(a_t)=i$, append $a_t$ to $\mathcal{S}_i$, otherwise append $\perp$ to $\mathcal{S}_i$.
	    }
	    \Else{
	    For $i\in[L]\cup \{0\}$, append $\perp$ at the end of $\mathcal{S}_i$.
	    }
	    For each $i\in [L]\cup \{0\}$, compute a set $H_i\subseteq\mathcal{U}$ together with a function $\hat{f}_i:H_i\rightarrow \mathbb{R}_{\geq 0}$ satisfying:
	    \begin{enumerate}
	        \item $\forall a\in H_i, (1-\eta')\cdot f_a\leq \hat{f}_i(a)\leq (1+\eta')\cdot f_a$, where $f_a$ is the frequency of $a$ in $a_1,a_2,\cdots,a_t$, and $\eta'$ satisfies $\eta' \leq \frac{\eta}{10000(L+1)|H_i|}$.
	        \item $\forall a \in \mathcal{U}$ that appears in $\mathcal{S}_i$, if $f_a\geq \tau$ and $f_a^p\geq \|\mathcal{S}_i\|_p^p/B$, $a\in H_i$.
	    \end{enumerate}
	    Compute $\hat{s}_1,\hat{s}_2,\cdots,\hat{s}_k$ where $\forall l\in[k]$, $\hat{s}_l$ is an $(\alpha,\gamma)$-approximation to $|\{a\in\mathcal{U}\mid f_a = l\}|$.\\
	    \For{$q\in[q_1^*,q_2^*]$}{
	        Initialize $\hat{z}_q=0$.\\
	        \For{$i\in [L]\cup\{0\}$}{
	            \If{$|\{a\in H_i\mid \hat{f}_i(a) \in I_q\}|\geq 8\lambda/\eta^2$ or $i=0$}{
	                $\hat{z}_q\gets \max(\hat{z}_q, |\{a\in H_i\mid \hat{f}_i(a) \in I_q\}|\cdot 2^i)$.
	            }
	        }
	    }
	    Output $\hat{F}_p =\sum_{l\in [k]} \hat{s}_l\cdot l^p + \sum_{q\in[q_1^*,q_2^*]} \hat{z}_q \cdot (\beta(1+\eta)^{q})^p$
	}
    \caption{$\ell_p$ Frequency Moment Estimation}
\end{algorithm}

\begin{lemma}\label{lem:lp_moment_dp_guarantee}
Consider the subroutines in Algorithm~\ref{alg:lp_moment}. 
If the algorithm that continually release $\hat{s}_1,\hat{s}_2,\cdots,\hat{s}_k$ is $\varepsilon$-DP and for every $i\in [L]\cup\{0\}$ the algorithm that continually release $(H_i,\hat{f}_i)$ is $\varepsilon$-DP.
Algorithm~\ref{alg:lp_moment} is $4\varepsilon$-DP in the continual release model.
\end{lemma}
\begin{proof}
Consider two neighboring stream $\mathcal{S}=(a_1,a_2,\cdots,a_T)$ and $\mathcal{S}'=(a'_1,a'_2,\cdots,a'_T)$ where they only differ at timestamp $t$, i.e., $a_t\not=a'_t$.
Consider the corresponding streams $\mathcal{S}_0,\mathcal{S}_1,\cdots,\mathcal{S}_L$ and $\mathcal{S}_0',\mathcal{S}_1',\cdots,\mathcal{S}_L'$.
$\forall i\in [L]\cup\{0\},j\not =t$, the $j$-th element of $\mathcal{S}_i$ should be the same as the $j$-th element of $\mathcal{S}'_i$.
Furthermore, $\forall i\in [L]$ with $i\not=g(a_t),i\not=g(a'_t)$, the $t$-th element of $\mathcal{S}_i$ is also the same as the $t$-th element of $\mathcal{S}'_i$.
Thus, at most $3$ streams $\mathcal{S}_0,\mathcal{S}_{g(a_t)}$ and $\mathcal{S}_{g(a_t')}$ might be different from $\mathcal{S}'_0,\mathcal{S}'_{g(a_t)}$ and $\mathcal{S}'_{g(a_t')}$ respectively.
Thus, the output $\{(H_0,\hat{f}_0),(H_1,\hat{f}_1),\cdots,(H_L,\hat{f}_L)\}$ is $3\varepsilon$-DP.
Since $(\hat{s}_1,\hat{s}_2,\cdots,\hat{s}_k)$ is $\varepsilon$-DP and the final output only depends on $(\hat{s}_1,\hat{s}_2,\cdots,\hat{s}_k)$ and $\{(H_0,\hat{f}_0),(H_1,\hat{f}_1),\cdots,(H_L,\hat{f}_L)\}$, the final output is $4\varepsilon$-DP.
\end{proof}

In the remaining of the section, let us analyze the approximation guarantee of Algorithm~\ref{alg:lp_moment}.
Let us consider a timestamp $t\in[T]$ and $\forall a\in\mathcal{U}$, let $f_a$ denote the frequency of $a$ in $a_1,a_2,\cdots,a_t$.
Let $\mathcal{S},\mathcal{S}_0,\mathcal{S}_1, \cdots, \mathcal{S}_L$ denote the streams up to timestamp $t$.
Let $\mathcal{I}=\{[1,1],[2,2],\cdots,[k,k],I_{q_1^*},I_{q_1^*+1},\cdots, I_{q_2^*}\}$. 
By our choice of $k,q_1^*,q_2^*$, we have the following observation:
\begin{observation}
$\sum_{I\in\mathcal{I}}\sum_{a\in\mathcal{U}:f_a\in I} f_a^p = \|\mathcal{S}\|_p^p$.
\end{observation}

\begin{definition}\label{def:contributing}
For any $I\in \mathcal{I}$, if $\sum_{a\in\mathcal{U}:f_a\in I} f_a^p\geq \eta\|\mathcal{S}\|_p^p/(q_2^*-q_1^*+1)$ or $I$ is $\{i\}$ for some $i\in[k]$, then interval $I$ is contributing.
\end{definition}

\subsubsection{Analysis of High Frequency Elements}
In this section, we show that $\sum_{q\in[q_1^*,q_2^*]} \hat{z}_q \cdot (\beta(1+\eta)^{q})^p$ is a good approximation to $\sum_{q\in[q_1^*,q_2^*]} \sum_{a\in \mathcal{U}:f_a\in I_q} f_a^p$ and $\sum_{q\in[q_1^*,q_2^*]:I_q\text{ is contributing}} \sum_{a\in \mathcal{U}:f_a\in I_q} f_a^p$.

The following lemma says that the frequency of any particular sampled element should be sufficiently far away from the boundary of intervals $I_{q^*_1}, I_{q^*_1+1},\cdots, I_{q^*_2}$ with a good probability.
\begin{lemma}[\cite{indyk2005optimal}]\label{lem:far_away_from_boundary}
Consider $\beta$ in Algorithm~\ref{alg:lp_moment}.
Consider any $f\in [T]$ and any $r\geq 2(\eta/T)^{C-1}$.
\begin{align*}
\Pr_{\beta'}\left[\min_{q\in \{q^*_1,q^*_1+1,\cdots,q^*_2,q^*_2+1\}} |f - \beta(1+\eta)^q| < r\right] \leq \frac{100r}{\eta \cdot f}
\end{align*}
\end{lemma}
By applying above lemma, we show that with high probability, we can use $\hat{f}_i(a)$ to correctly classify $a$ into right interval in $I_{q^*_1},I_{q^*_1+1},I_{q^*_1+2},\cdots,I_{q^*_2}$.
\begin{lemma}
With probability at least $0.99$, $\forall i\in [L]\cup\{0\}$, $\forall a \in H_i$, and $\forall q\in[q^*_1,q^*_2]$, $\hat{f}_i(a)\in I_q$ if and only if $f_a\in I_q$.
\end{lemma}
\begin{proof}
According to Lemma~\ref{lem:far_away_from_boundary}, for any $i\in [L]\cup \{0\}$ and any $a\in H_i$, we have:
\begin{align*}
\Pr\left[\min_{q\in \left\{q^*_1,q^*_1+1,\cdots,q^*_2+1\right\}} |f_a - \beta(1+\eta)^q| <  \eta' f_a\right] \leq \frac{1}{100\cdot (L+1)\cdot |H_i|},
\end{align*}
where the inequality follows from $\eta'\leq \frac{\eta}{10000(L+1)|H_i|}$.
By taking a union bound over all $i\in [L]\cup\{0\}$ and all $a\in H_i$, with probability at least $0.99$, the following event happens:
$\forall i\in [L]\cup \{0\},\forall a \in H_i,\forall q\in  \left\{q^*_1,q^*_1+1,\cdots,q^*_2+1\right\}$, \begin{enumerate}
    \item if $f_a\leq \beta(1+\eta)^q$, then $\hat{f}_i(a)\leq f_a+\eta'f_a\leq \beta(1+\eta)^q$,
    \item if $f_a> \beta(1+\eta)^q$, then $\hat{f}_i(a)\geq f_a - \eta'f_a >\beta(1+\eta)^q$.
\end{enumerate}
Therefore, with probability at least $0.99$, $\forall i\in [L]\cup\{0\}$, $\forall a \in H_i$, and $\forall q\in[q^*_1,q^*_2]$, $\hat{f}_i(a)\in I_q$ if and only if $f_a\in I_q$.
\end{proof}

Let $G_i$ denote the set of elements that appears in the stream $\mathcal{S}_i$ and let $G_{i,q}$ denote the set of elements that appear in the stream $\mathcal{S}_i$ and whose frequency is in the interval $I_q$.
Formally, $\forall i \in [L]$, let $G_i = \{a_j \mid g(a_j) = i, j\leq t\},G_0=\{a_j\mid j\leq t\}$, and $\forall q\in \{q^*_1,q^*_1+1,\cdots,q^*_2\},\forall i\in [L]\cup\{0\},$ let $G_{i,q}=\{a\in G_i\mid f_a\in I_q\}$.
For $q\in[q^*_1,q^*_2]$, let $z_q=|G_{0,q}|$, i.e., the number of elements that are in the stream $\mathcal{S}$ and have frequency in the range $I_q$.

\begin{lemma}\label{lem:freq_concentration}
$\forall i \in [L]\cup\{0\},q\in[q^*_1,q^*_2]$, if $i=0$ or $z_q\geq 2^i\cdot 4\lambda/\eta^2, \Pr[|G_{i,q}| \in (1\pm \eta) \cdot z_q/2^i]\geq 1 - 0.01/\left((L+1)\cdot \log(4T)/\eta\right)$.
Otherwise, $\Pr[||G_{i,q}|-z_q/2^i|\leq4\lambda/\eta] \geq 1 - 0.01/\left((L+1)\cdot \log(4T)/\eta\right)$.
\end{lemma}
\begin{proof}
Consider any $i\in [L]\cup \{0\}$ and $q\in [q^*_1,q^*_2]$. If $i=0$, by definition $z_q=|G_{0,q}|$.

Suppose $z_q\geq 2^i\cdot 4\lambda/\eta^2$.
Due to Lemma~\ref{lem:concentration}, we have:
\begin{align*}
&\Pr\left[\left||G_{i,q}| - z_q / 2^i\right| > \eta\cdot z_q / 2^i\right]\\
\leq & 8\cdot \left(\frac{\lambda\cdot z_q/2^i + \lambda^2}{\left(\eta\cdot z_q/2^i\right)^2}\right)^{\lambda/2}\\
\leq & 0.01/\left((L+1)\cdot \log(4T)/\eta\right),
\end{align*}
where the last inequality follows from that $\lambda = 2\cdot \log(1000(L+1)\log(4T)/\eta)$ and $z_q\geq 2^i\cdot 4\lambda / \eta^2$.

Suppose $z_q \leq 2^i \cdot4\lambda/\eta^2$, By applying Lemma~\ref{lem:concentration} again, we have:
\begin{align*}
&\Pr\left[\left||G_{i,q}|- z_q/2^i\right| >4\lambda/\eta\right]\\
\leq & 8\cdot \left(\frac{\lambda\cdot z_q/2^i + \lambda^2}{(4\lambda/\eta)^2}\right)^{\lambda/2}\\
\leq & 0.01/\left((L+1)\cdot \log(4T)/\eta\right),
\end{align*}
where the last inequality follows from $z_q/2^i \leq 4\lambda / \eta^2$ and $\lambda = 2\cdot \log (1000 (L+1)\log(4T)/\eta)$.
\end{proof}

We define events $\mathcal{E}_1$ and $\mathcal{E}_2$ as the following:
     $\mathcal{E}_1$ denotes the event $\forall i\in[L]\cup \{0\},\forall a\in H_i,\forall q \in [q^*_1,q^*_2],\hat{f}_i(a)\in I_q$ if and only if $f_a\in I_q$.
     $\mathcal{E}_2$ denotes the event:
\begin{enumerate}
     \item $\forall i\in[L]\cup \{0\},\forall q\in [q^*_1,q^*_2],$ if $z_q\geq 2^i\cdot 4\lambda/\eta^2$ or $i=0$, $|G_{i,q}|\cdot 2^i\in (1\pm \eta) z_q$.
     \item $\forall i\in[L],\forall q\in [q^*_1,q^*_2],$ if $z_q< 2^i\cdot 4\lambda/\eta^2$, $|G_{i,q}|\in z_q/2^i\pm 4\lambda/\eta$.
\end{enumerate}
According to Lemma~\ref{lem:far_away_from_boundary}, $\mathcal{E}_1$ happens with probability at least $0.99$.
According to Lemma~\ref{lem:freq_concentration}, since $\forall q \in [q^*_1,q^*_2], z_q = |G_{0,q}|$ and $q^*_2-q^*_1+1 \leq \log(4T)/\eta$, by taking a union bound over all $i\in [L]\cup \{0\}$ and all $q\in[q^*_1,q^*_2]$, $\mathcal{E}_2$ happens with probability at least $0.99$.

\begin{lemma}[Upper bound of estimation of high frequency moment]\label{lem:ub_high_freq}
Condition on $\mathcal{E}_1$ and $\mathcal{E}_2$,
\begin{align*}
\sum_{q\in[q^*_1,q^*_2]} \hat{z}_q \cdot (\beta(1+\eta)^q)^p \leq (1+\eta) \cdot \sum_{q\in[q^*_1,q^*_2]} \sum_{a\in\mathcal{U},f_a\in I_q}f_a^p
\end{align*}
\end{lemma}
\begin{proof}
Due to event $\mathcal{E}_1$, we have $\forall i \in [L]\cup \{0\}, \forall q \in [q^*_1,q^*_2], \{a\in H_i \mid \hat{f}_i(a)\in I_q\} \subseteq  G_{i,q}$.
Note that for $q\in[q^*_1,q^*_2],$ $\hat{z}_q$ is either $0$, or $\hat{z}_q = 2^{i'}\cdot |\{a\in H_{i'} \mid \hat{f}_{i'}(a)\in I_q\}|$ for some $i'$ satisfying $i'=0$ or $|\{a\in H_{i'} \mid \hat{f}_{i'}(a)\in I_q\}|\geq 8\lambda/\eta^2$.
Since $\forall q\in[q^*_1,q^*_2],\forall i\in [L]\cup \{0\}$, $|\{a\in H_i \mid \hat{f}_i(a)\in I_q\}|\leq |G_{i,q}|$, 
$\forall q\in [q^*_1,q^*_2],$ if $\hat{z}_q\not = 0$, there is some $i'\in [L]\cup \{0\}$ such that:
\begin{align*}
\hat{z}_q&= 2^{i'}\cdot |\{a\in H_{i'} \mid \hat{f}_{i'}(a)\in I_q\}|\\
&\leq 2^{i'}\cdot |G_{i',q}|\\
&\leq (1+\eta)\cdot z_q,
\end{align*}
where the second inequality follows from that $|G_{i',q}|\geq 8\lambda/\eta^2$ which implies that $|G_{i',q}|\cdot 2^{i'}\in (1\pm \eta)z_q$ according to event $\mathcal{E}_2$.

Therefore,
\begin{align*}
&\sum_{q\in[q^*_1,q^*_2]} \hat{z}_q \cdot (\beta(1+\eta)^q)^p \\
\leq & (1+\eta) \sum_{q\in[q^*_1,q^*_2]}  z_q \cdot (\beta(1+\eta)^q)^p \\
\leq & (1+\eta) \sum_{q\in[q^*_1,q^*_2]}  \sum_{a\in \mathcal{U}:f_a \in I_q} (\beta(1+\eta)^q)^p \\
\leq & (1+\eta) \sum_{q\in[q^*_1,q^*_2]} \sum_{a\in \mathcal{U}:f_a \in I_q} f_a^p, \\
\end{align*}
where the first inequality follows from $\hat{z}_q\leq (1+\eta) z_q$, the second inequality follows from the definition of $z_q$, i.e., $z_q=|\{a\in\mathcal{U}\mid f_a\in I_q\}|$, and the last inequality follows from that $f_a\geq \beta(1+\eta)^q$ if $f_a\in I_q$.
\end{proof}

Define the event $\mathcal{E}_3$ as: $\forall i\in [L]\cap \{0\}$, $\|\mathcal{S}_i\|_p^p\leq 100(L+1)\cdot\|\mathcal{S}\|_p^p/2^i$.
\begin{lemma}
$\mathcal{E}_3$ happens with probability at least $0.99$.
\end{lemma}
\begin{proof}
Consider $i\in [L]\cup\{0\}$.
We have $E\left[\|\mathcal{S}_i\|_p^p\right] = \sum_{a\in \mathcal{U}} \Pr[g(a)=i]\cdot f_a^p=\|\mathcal{S}\|_p^p/2^i$.
By Markov's inequality, with probability at least $1-1/(100(L+1))$, $\|\mathcal{S}_i\|_p^p\leq 100(L+1)\cdot \|\mathcal{S}\|_p^p/2^i$.
By taking a union bound over all $i\in[L]\cup \{0\}$, $\mathcal{E}_3$ happens with probability at least $0.99$.
\end{proof}

In the remaining of the analysis, we condition on $\mathcal{E}_3$ as well.

\begin{lemma}[Lower bound of estimation of contributing high frequency moment]\label{lem:lb_high_freq}
Condition on $\mathcal{E}_1,\mathcal{E}_2,\mathcal{E}_3$.
\begin{align*}
\sum_{q\in[q^*_1,q^*_2]} \hat{z}_q \cdot (\beta(1+\eta)^q)^p \geq (1-\eta)^{p+1} \cdot \sum_{q\in[q^*_1,q^*_2]:I_q\text{ is contributing}} \sum_{a\in\mathcal{U},f_a\in I_q}f_a^p
\end{align*}
\end{lemma}
\begin{proof}
Consider any contributing $q\in[q^*_1,q^*_2]$.

\xhdr{Case 1: $z_q\leq 16\lambda/\eta^2$.}
According to Definition~\ref{def:contributing}, we have $\sum_{a\in\mathcal{U}:f_a\in I_q} f_a^p\geq \eta \|\mathcal{S}\|_p^p/(q^*_2-q^*_1+1)$ which implies that $\eta \|\mathcal{S}_0\|_p^p/(q^*_2-q^*_1+1)\leq z_q \cdot (\beta(1+\eta)^{q+1})^p\leq 16\lambda/\eta^2\cdot (1+\eta)^p \cdot (\beta(1+\eta)^q)^p$.
Since $B\geq (q^*_2-q^*_1+1)\cdot 16\lambda\cdot (1+\eta)^p/\eta^3$, $\forall a \in\mathcal{U}$ with $f_a\in I_q$, $f_a^p\geq \|\mathcal{S}_0\|_p^p/B$ .
Therefore, $\forall a \in\mathcal{U}$ with $f_a\in I_q$, we have $a\in H_0$.
According to event $\mathcal{E}_1$, we have $|\{a\in H_0\mid \hat{f}_0(a)\in I_q\}| = |G_{0,q}|=z_q$.
Thus, we have $\hat{z}_q\geq |\{a\in H_0\mid \hat{f}_0(a)\in I_q\}|\geq z_q$.

\xhdr{Case 2: $z_q>16\lambda/\eta^2$.}
Let $i^*\in\{0\}\cup [L]$ be the largest value such that $z_q/2^{i^*}\geq 16\lambda/\eta^2$.
According to event $\mathcal{E}_2$, we have $|G_{i^*,q}|\geq 8\lambda/\eta^2$.

Since $q$ is contributing, we have:
\begin{align*}
&z_q \cdot (\beta(1+\eta)^{q+1})^p\\
\geq & \sum_{a\in\mathcal{U}: f_a\in I_q} f_a^p\\
\geq & \eta\|\mathcal{S}\|_p^p/(q^*_2-q^*_1+1)\\
\geq & 2^{i^*}\cdot \eta\|\mathcal{S}_{i^*}\|_p^p/((q^*_2-q^*_1+1)\cdot 100(L+1)),
\end{align*}
where the last inequality follows from event $\mathcal{E}_3$.
Therefore, we have 
\begin{align*}
&(\beta(1+\eta)^q)^p\\
\geq & \eta\|\mathcal{S}_{i^*}\|_p^p/((q^*_2-q^*_1+1)\cdot 100(L+1)\cdot (z_q/2^{i^*}) \cdot (1+\eta)^p)\\
\geq & \eta\|\mathcal{S}_{i^*}\|_p^p/((q^*_2-q^*_1+1)\cdot 100(L+1)\cdot (32\lambda/\eta^2) \cdot (1+\eta)^p),
\end{align*}
where the last inequality follows from $z_q/2^{i^*}\leq 32\lambda/\eta^2$.
Since $B\geq ((q^*_2-q^*_1+1)\cdot 100(L+1)\cdot (32\lambda/\eta^2) \cdot (1+\eta)^p)/\eta$, we have $\forall a\in\mathcal{U}$ with $a\in G_{i^*,q}$, $f_a^p\geq \|\mathcal{S}_{i^*}\|_p^p/B$ and $f_a>\beta(1+\eta)^q>\tau$ which implies that $a\in H_{i^*}$.
According to event $\mathcal{E}_1$, we have $\{a\in H_{i^*}\mid \hat{f}_{i^*}(a)\in I_q\}=G_{i^*,q}$.
Therefore, $|\{a\in H_{i^*}\mid \hat{f}_{i^*}(a)\in I_q\}|\geq 8\lambda/\eta^2$.
Finally, in addition, according to event $\mathcal{E}_2$, we have $\hat{z}_q\geq |\{a\in H_{i^*}\mid \hat{f}_{i^*}(a)\in I_q\}|\cdot 2^{i^*}\geq (1-\eta)z_q$.

Therefore, in any case, we have $\hat{z}_q\geq (1-\eta)z_q$, and we have:
\begin{align*}
&\sum_{q\in[q^*_1,q^*_2]} \hat{z}_q \cdot (\beta(1+\eta)^q)^p \\
\geq & \sum_{q\in[q^*_1,q^*_2]:I_q\text{ is contributing}} \hat{z}_q \cdot (\beta(1+\eta)^q)^p \\
\geq & \sum_{q\in[q^*_1,q^*_2]:I_q\text{ is contributing}} (1-\eta)z_q \cdot (\beta(1+\eta)^q)^p \\
= & \frac{1-\eta}{(1+\eta)^p}\sum_{q\in[q^*_1,q^*_2]:I_q\text{ is contributing}} \sum_{a\in\mathcal{U}:f_a\in I_q} (\beta(1+\eta)^{q+1})^p \\
\geq & \frac{1-\eta}{(1+\eta)^p}\sum_{q\in[q^*_1,q^*_2]:I_q\text{ is contributing}} 
\sum_{a\in\mathcal{U}:f_a\in I_q} f_a^p \\
\geq & (1-\eta)^{p+1}\sum_{q\in[q^*_1,q^*_2]:I_q\text{ is contributing}} \sum_{a\in\mathcal{U}:f_a\in I_q} f_a^p, \\
\end{align*}
where the equality follows from the definition that $z_q = |\{a\in\mathcal{U}\mid f_a\in I_q\}|$.
\end{proof}

\subsubsection{Analysis of Low Frequency Elements}
In this section, we show that $\sum_{l\in [k]}\hat{s}_l\cdot l^p$ is a good approximation to $\sum_{l\in [k]} \sum_{a\in\mathcal{U}:f_a=l} f_a^p$.

\begin{lemma}[Approximation of low frequency moments]\label{lem:approx_low_freq}
$\sum_{l\in [k]}\hat{s}_l\cdot l^p$ is a $(\alpha,\gamma\cdot (2\tau)^{p+1})$-approximation to $\sum_{l\in [k]} \sum_{a\in\mathcal{U}:f_a=l} f_a^p$.
\end{lemma}
\begin{proof}
For $l\in [k]$, let $s_l=|\{a\in\mathcal{U}\mid f_a=l\}|$.
We have:
\begin{align*}
&\sum_{l\in [k]}\hat{s}_l\cdot l^p\\
\geq & \sum_{l\in [k]}(\frac{1}{\alpha}\cdot s_l - \gamma) \cdot l^p \\
\geq  & \frac{1}{\alpha}\cdot \sum_{l\in [k]} \sum_{a\in \mathcal{U}:f_a = l} f_a^p - k\cdot \gamma \cdot k^p \\
\geq & \frac{1}{\alpha}\cdot \sum_{l\in [k]} \sum_{a\in \mathcal{U}:f_a = l} f_a^p - \gamma \cdot (2\tau)^{p+1},
\end{align*}
where the last inequality follows from that our choice of $k$ implies that $k\leq 2\tau$.
Similarly, we have:
\begin{align*}
&\sum_{l\in [k]}\hat{s}_l\cdot l^p\\
\leq & \sum_{l\in [k]}(\alpha\cdot s_l + \gamma) \cdot l^p \\
\leq & \alpha\cdot \sum_{l\in [k]} \sum_{a\in \mathcal{U}:f_a = l} f_a^p - \gamma \cdot (2\tau)^{p+1}.
\end{align*}
\end{proof}

\subsubsection{Putting High Frequency Moments and Low Frequency Moments Together}

\begin{lemma}\label{lem:contributing_cost}
$\sum_{\text{contributing }I\in\mathcal{I}} \sum_{a\in\mathcal{U}:f_a\in I}f_a^p\geq (1-\eta)\cdot\|\mathcal{S}\|_p^p$.
\end{lemma}
\begin{proof}
\begin{align*}
&\sum_{\text{contributing }I\in\mathcal{I}} \sum_{a\in\mathcal{U}:f_a\in I}f_a^p\\
= & \sum_{I\in\mathcal{I}} \sum_{a\in\mathcal{U}:f_a\in I}f_a^p - \sum_{\text{non-contributing }I\in\mathcal{I}} \sum_{a\in\mathcal{U}:f_a\in I}f_a^p\\
\geq & \|\mathcal{S}\|_p^p- (q^*_2-q^*_1+1)\cdot \eta \cdot \|\mathcal{S}\|_p^p/(q^*_2-q^*_1+1)\\
\geq & (1-\eta) \cdot \|\mathcal{S}\|_p^p.
\end{align*}
\end{proof}

\begin{lemma}\label{lem:acc_lp_moment}
Consider any timestamp $t\in [T]$.
$\forall a\in \mathcal{U}$, let $f_a$ denote the frequency of $a$ in $a_1,a_2,\cdots,a_t$.
With probability at least $0.9$, the output $\hat{F}_p$ of Algorithm~\ref{alg:lp_moment} is a $\left(\max\left(\frac{\alpha}{1-\eta},(1+2\eta)^{p+2}\right),\gamma\cdot (2\tau)^{p+1}\right)$-approximation to $\|\mathcal{S}\|_p^p$.
\end{lemma}
\begin{proof}
According to Lemma~\ref{lem:ub_high_freq} and Lemma~\ref{lem:lb_high_freq}, we have:
\begin{align*}
(1-\eta)^{p+1} \sum_{q\in[q^*_1,q^*_2]:I_q\text{ is contributing}} \sum_{a\in\mathcal{U},f_a\in I_q}f_a^p\leq \sum_{q\in[q^*_1,q^*_2]} \hat{z}_q \cdot (\beta(1+\eta)^q)^p \leq (1+\eta)  \sum_{q\in[q^*_1,q^*_2]} \sum_{a\in\mathcal{U},f_a\in I_q}f_a^p.
\end{align*}
According to Lemma~\ref{lem:approx_low_freq}, we have:
\begin{align*}
\frac{1}{\alpha} \cdot \sum_{l\in [k]}\sum_{a\in\mathcal{U}:f_a=l} f_a^p - \gamma\cdot (2\tau)^{p+1}\leq \sum_{l\in [k]}\hat{s}_l\cdot l^p \leq \alpha \cdot \sum_{l\in [k]}\sum_{a\in\mathcal{U}:f_a=l} f_a^p + \gamma\cdot (2\tau)^{p+1}
\end{align*}
Therefore, we have:
\begin{align*}
\hat{F}_p&\geq \min\left(\frac{1}{\alpha},(1-\eta)^{p+1}\right)\sum_{I\in\mathcal{I}:I\text{ is contributing}} \sum_{a\in\mathcal{U}:f_a\in I} f_a^p-\gamma\cdot (2\tau)^{p+1}\\
&\geq \min\left(\frac{1-\eta}{\alpha},(1-\eta)^{p+2}\right)\|\mathcal{S}\|_p^p - \gamma\cdot (2\tau)^{p+1}\\
&\geq \min\left(\frac{1-\eta}{\alpha},\frac{1}{(1+2\eta)^{p+2}}\right)\|\mathcal{S}\|_p^p - \gamma\cdot (2\tau)^{p+1}
\end{align*}
where the second step follows from Lemma~\ref{lem:contributing_cost}.
Similarly, we have:
\begin{align*}
\hat{F}_p & \leq  \max(\alpha,1+\eta)\sum_{I\in \mathcal{I}}\sum_{a\in\mathcal{U}:f_a\in I} f_a^p + \gamma\cdot (2\tau)^{p+1}\\
& = \max(\alpha,1+\eta)\cdot \|\mathcal{S}\|_p^p + \gamma\cdot (2\tau)^{p+1}
\end{align*}
\end{proof}

\begin{theorem}[Streaming continual release $\ell_p$ frequency moment estimation]\label{thm:lp_moment}
Let $p>0,\varepsilon\geq 0,\xi \in (0,0.5),\eta\in(0,0.5)$.
There is an $\varepsilon$-DP algorithm in the streaming continual release model such that with probability at least $1-\xi$, it always outputs an $\left(1+\eta, \left(\frac{\log(T|\mathcal{U}|/\xi)}{\eta\varepsilon}\right)^{O(\max(1,p))}\right)$-approximation to $\|\mathcal{S}\|_p^p$.
The algorithm uses space at most
\begin{align*}
\phi\cdot \left(\frac{\log(T|\mathcal{U}|/\xi)}{\eta\varepsilon}\right)^{O(\max(1,p))},
\end{align*}
where $\phi = \max(1,|\mathcal{U}|^{1-2/p})$.
\end{theorem}
\begin{proof}
To boost the probability of the approximation guarantee of Lemma~\ref{lem:acc_lp_moment} to $1-\xi/3$ and simultaneously for all timestamps $t\in T$, we run $\lceil 50\log(3T/\xi) \rceil$ independent copies of Algorithm~\ref{alg:lp_moment} and take the median of the outputs.
Since we run $\lceil 50\log(3T/\xi) \rceil$ independent copies of Algorithm~\ref{alg:lp_moment} and according to Lemma~\ref{lem:lp_moment_dp_guarantee}, if we want the final algorithm to be $\varepsilon$-DP, we need each subroutine of the streaming continual release of $(H_i,\hat{f}_i)$ for $i\in [L]\cup \{0\}$ to be $(\varepsilon/(4\cdot \lceil 50\log(3T/\xi) \rceil))$-DP and we need the subroutine of the streaming continual release of $\{\hat{s}_1,\hat{s}_2,\cdots,\hat{s}_k\}$ to be $(\varepsilon/(4\cdot \lceil 50\log(3T/\xi) \rceil))$-DP as well.
To simultaneously make the call of each subroutine of the streaming continual release of $(H_i,\hat{f}_i)$ over all independent copies of Algorithm~\ref{alg:lp_moment} satisfy the desired properties stated in Algorithm~\ref{alg:lp_moment} with probability at least $1-\xi/3$, we need to make $(H_i,\hat{f}_i)$ satisfy the property for each particular $i\in[L]\cup \{0\}$ and a particular copy of Algorithm~\ref{alg:lp_moment} with probability at least $1-\xi/(4\cdot \lceil50\log(3T/\xi)\rceil\cdot (L+1))$.
To simultaneously make the call of each subroutine of the streaming continual release of $\{\hat{s}_1,\hat{s}_2,\cdots,\hat{s}_k\}$ over all independent copies of Algorithm~\ref{alg:lp_moment} satisfy the desired property stated in Algorithm~\ref{alg:lp_moment} with probability at least $1-\xi/3$, we need to make $\{\hat{s}_1,\hat{s}_2,\cdots,\hat{s}_k\}$ satisfy the desired property for each particular $i\in[L]\cup \{0\}$ and a particular copy of Algorithm~\ref{alg:lp_moment} with probability at least $1-\xi/(3\cdot \lceil 50\log(3T/\xi)\rceil)$.
According to Algorithm~\ref{alg:lp_moment}, we have 
\begin{align*}
B=\Theta\left(\frac{\log(T)\log(|\mathcal{U}|)\log(\log(T|\mathcal{U}|)/\eta)\cdot (1+\eta)^p}{\eta^4}\right).
\end{align*}
Thus, according to Theorem~\ref{thm:lp_heavyhitters}, the size of $|H_i|$ in Algorithm~\ref{alg:lp_moment} for $i\in[L]\cup \{0\}$ is at most $\poly\left(\frac{\log(T|\mathcal{U}|/\xi)}{\eta}\right)\cdot 2^{O(p)}$.
Thus, we choose $\eta'=1/\left(\poly\left(\frac{\log(T|\mathcal{U}|/\xi)}{\eta}\right)\cdot 2^{O(p)}\right)$.
Then according to Theorem~\ref{thm:lp_heavyhitters}, we have 
\begin{align*}
\tau = \frac{1}{\varepsilon}\cdot \poly\left(\frac{\log(T\cdot |\mathcal{U}|/\xi)}{\eta}\right)\cdot 2^{O(p)}.
\end{align*}
According to Theorem~\ref{thm:count_low_freq_large}, we have $\alpha = 1$, and we choose $\gamma$ to be
\begin{align*}
\left(\eta''+{\eta''}^2\cdot \frac{\tau}{\varepsilon}\cdot \poly\left(\log\left(\frac{T\tau}{\xi}\right)\right)\right)\cdot \|\mathcal{S}\|_0+\frac{1}{\varepsilon} \cdot \poly\left(\log\left(\frac{T}{\xi}\right)\right)+O\left(\frac{\log \tau}{{\eta''}^2}\right),
\end{align*}
where we choose $\eta''$ to be 
\begin{align*}
\frac{\varepsilon\eta}{\tau^{O(\max(1,p))}\poly\left(\frac{\log(T|\mathcal{U}|/\xi)}{\eta}\right)} =\frac{1}{\left(\frac{\log(T|\mathcal{U}|/\xi)}{\eta\varepsilon}\right)^{O(\max(1,p))}}
\end{align*}
such that 
\begin{align*}
\gamma\cdot (2\tau)^{p+1} &\leq  \eta \|\mathcal{S}\|_0 + \left(\frac{\log(T|\mathcal{U}|/\xi)}{\eta\varepsilon}\right)^{O(\max(1,p))}\\
& \leq \eta\|\mathcal{S}\|_p^p + \left(\frac{\log(T|\mathcal{U}|/\xi)}{\eta\varepsilon}\right)^{O(\max(1,p))}.
\end{align*}
Note that $\eta\|\mathcal{S}\|_p^p$ becomes the relative error.
Thus, according to Lemma~\ref{lem:acc_lp_moment}, the output $\hat{F}_p$ is an $\left((1+\eta)^{O(\max(1,p))}, \left(\frac{\log(T|\mathcal{U}|/\xi)}{\eta\varepsilon}\right)^{O(\max(1,p))}\right)$-approximation.

Next, consider the space usage.
According to Theorem~\ref{thm:lp_heavyhitters}, the total space needed to run all heavy hitters subroutines is at most 
\begin{align*}
\phi\cdot 2^{O(\max(1,p))}\cdot \poly\left(\frac{\log(T|\mathcal{U}|/\xi)}{\eta}\right).
\end{align*}
According to Theorem~\ref{thm:count_low_freq_large}, the total space needed for computing $\{\hat{s}_1,\hat{s}_2,\cdots,\hat{s}_k\}$ for all running copies of Algorithm~\ref{alg:lp_moment} is at most 
\begin{align*}
\left(\frac{\log(T|\mathcal{U}|/\xi)}{\eta\varepsilon}\right)^{O(\max(1,p))}.
\end{align*}
Therefore, the overall space needed is at most
\begin{align*}
\phi\cdot \left(\frac{\log(T|\mathcal{U}|/\xi)}{\eta\varepsilon}\right)^{O(\max(1,p))}.
\end{align*}
\end{proof}

\section{Extension to Sliding Window Continual Release Algorithms}
In this section, we briefly review the smooth histogram~\cite{BravermanO07} technique which converts any (non-private) streaming algorithm into (non-private) sliding window algorithm.
The original framework only supports the approximation algorithm which only has relative error and no additive error.
In this section, we show how to extend it to support the additive error as well.

Suppose there are two streams $\mathcal{A}=(a_1,a_2,\cdots,a_{t_1})$ and $\mathcal{B}=(b_1,b_2,\cdots,b_{t_2})$.
We use $\mathcal{A}\cup\mathcal{B}$ to denote the concatenation of two streams, i.e., $\mathcal{A}\cup \mathcal{B}=(a_1,a_2,\cdots,a_{t_1},b_1,b_2,\cdots,b_{t_2})$.
If $\mathcal{B}$ is a suffix of $\mathcal{A}$, i.e., $\exists i \in [t_1]$ such that $a_i=b_1,a_{i+1}=b_2,\cdots,a_{t_1}=b_{t_2}$, then we denote it as $\mathcal{B}\subseteq_r \mathcal{A}$.

\begin{definition}[Smooth function~\cite{BravermanO07}]\label{def:smooth_function}
Let $g(\cdot)$ be a function over streams.
Function $g(\cdot)$ is $(\zeta,\beta)$-smooth if:
\begin{enumerate}
    \item $\forall \mathcal{A},0\leq g(\mathcal{A})\leq \poly(T)$.
    \item $\forall \mathcal{A},\mathcal{B}$ with $\mathcal{B}\subseteq_r\mathcal{A}$, $g(\mathcal{A})\geq g(\mathcal{B})$
    \item For any $\eta\in (0,1)$, there exists $\zeta(\eta,g)$ and $\beta(\eta,g)$ such that
    \begin{enumerate}
        \item $0<\beta\leq \zeta < 1$.
        \item If $\mathcal{B} \subseteq_r \mathcal{A}$ and $(1 - \beta)g(\mathcal{A}) \leq g(\mathcal{B})$ then $(1 - \zeta)g(\mathcal{A} \cup \mathcal{C}) \leq g(\mathcal{B} \cup \mathcal{C})$ for any stream $\mathcal{C}$.
    \end{enumerate}
\end{enumerate}
\end{definition}

\begin{lemma}[Smoothness of frequency moments~\cite{BravermanO07}]\label{lem:smoothness_moment}
For $p>1$, $\|\mathcal{S}\|_p^p$ is $\left(\eta,\left(\frac{\eta}{p}\right)^p\right)$-smooth.
For $0<p\leq 1$, $\|\mathcal{S}\|_p^p$ is $(\eta,\eta)$-smooth. 
$\|\mathcal{S}\|_0$ is $(\eta,\eta)$-smooth.
\end{lemma}

\begin{lemma}\label{lem:smoothness_additive_moment}
Let $0\leq Z\leq \poly(T)$.
If $g(\mathcal{S}):=\|\mathcal{S}\|_p^p+Z$, then $g(\cdot)$ is $\left(\eta,\left(\frac{\eta}{p}\right)^p\right)$-smooth if $p>1$ and $g(\cdot)$ is $(\eta,\eta)$-smooth if $0<p\leq 1$.
If $g(\mathcal{S})=\|\mathcal{S}\|_0+Z$, $g(\cdot)$ is $\left(\eta,\eta\right)$-smooth.
\end{lemma}
\begin{proof}
For $g(\cdot)$ stated in the lemma statement, the first two requirements in Definition~\ref{def:smooth_function} are satisfied obviously. 
Therefore, we only need to justify the third requirement of Definition~\ref{def:smooth_function}.

Let us consider $\mathcal{A}$ and $\mathcal{B}$ such that $\mathcal{B}\subseteq_r \mathcal{A}$ and $(1-\beta)g(\mathcal{A})\leq g(\mathcal{B})$.
Consider any stream $\mathcal{C}$.
We are going to prove $(1-\zeta)g(\mathcal{A}\cup \mathcal{C})\leq g(\mathcal{B}\cup \mathcal{C})$.
Consider the case $g(\mathcal{S})=\|\mathcal{S}\|_p^p+Z$.
We construct an auxiliary stream $\mathcal{X}$ such that the elements of $\mathcal{X}$ do not appear in any of $\mathcal{A},\mathcal{B},\mathcal{C}$ and $\|\mathcal{X}\|_p^p=Z$.
Then, we have
$(1-\beta)\|\mathcal{A}\cup\mathcal{X}\|_p^p = (1-\beta)g(\mathcal{A})\leq g(\mathcal{B})=\|\mathcal{B}\cup\mathcal{X}\|_p^p$.
Due to the smoothness of $\|\mathcal{S}\|_p^p$ (Lemma~\ref{lem:smoothness_moment}), we have $(1-\zeta)g(\mathcal{A}\cup\mathcal{C})=(1-\zeta)\|\mathcal{A}\cup \mathcal{X}\cup\mathcal{C}\|_p^p\leq \|\mathcal{B}\cup\mathcal{X}\cup \mathcal{C}\|_p^p=g(\mathcal{B}\cup\mathcal{C})$.
According to Lemma~\ref{lem:smoothness_moment},  $g(\cdot)$ is $\left(\eta,\left(\frac{\eta}{p}\right)^p\right)$-smooth if $p>1$ and $g(\cdot)$ is $(\eta,\eta)$-smooth if $0<p\leq 1$.

Similarly, using the similar argument by constructing $\|\mathcal{X}\|_0=Z$, we can show that $g(\mathcal{S})=\|\mathcal{S}\|_0+Z$ is $\left(\eta,\eta\right)$-smooth.

\end{proof}

\begin{lemma}\label{lem:additive_to_relative_approx}
Let $\alpha\geq 1,\gamma\geq 0$.
If $g'$ is an $(\alpha,\gamma)$-approximation to $g$, then $g'+Z$ is an $\alpha$-approximation to $g+Z$ if $Z\geq \frac{\alpha}{\alpha-1}\cdot \gamma$.
\end{lemma}
\begin{proof}
We have:
\begin{align*}
g'+Z&\geq \frac{1}{\alpha}\cdot g - \gamma + Z \\
&\geq \frac{1}{\alpha}\cdot g -\frac{\alpha-1}{\alpha}\cdot Z + Z\\
&\geq \frac{1}{\alpha}\cdot (g+ Z).
\end{align*}
On the other hand,
\begin{align*}
g'+Z&\leq \alpha\cdot g + \gamma + Z\\
& \leq \alpha \cdot g + (\alpha-1) Z + Z\\
& \leq \alpha\cdot (g+Z).
\end{align*}
\end{proof}

\begin{theorem}[Smooth histogram algorithmic framework~\cite{BravermanO07}]\label{thm:smooth_histogram_original}
Let $\eta\in(0,0.5)$.
Let $g(\cdot)$ be an $(\zeta,\beta)$-smooth function.
If there exists a streaming algorithm $\Lambda$ which maintains an $(\frac{1}{1-\eta})$-approximation of $g(\cdot)$  simultaneously for all timestamps $t\in [T]$ with probability at least $1-\xi$, using space $h(\eta,\xi)$, then there is a sliding window algorithm $\Lambda'$ that maintains a $\left(\frac{1}{1-\eta-\zeta}\right)$-approximation of $g(\cdot)$ over sliding windows simultaneously for all timestamps $t\in[T]$ with probability at least $1-\xi$ and uses space $O\left(\frac{\log T}{\beta}\cdot h(\eta,\xi\beta/\log(T))\right)$.

Furthermore, at any timestamp $t\in[T]$, $\Lambda'$ starts a new instance of $\Lambda$ which regards the $t$-th element in the stream as the beginning of the stream, and $\Lambda'$ only keeps at most $O(\log(T)/\beta)$ past instances of $\Lambda$ (started from different timestamps).
The output of $\Lambda'$ at any timestamp $t$ only depends on the outputs of its maintained instances $\Lambda$, and the decision of whether keeping an instance $\Lambda$ to timestamp $t+1$ only depends on the outputs of its maintained instances $\Lambda$ at timestamp $t$ as well.
\end{theorem}

We are able to extend the above smooth histogram framework to the differentially private continual release setting.

\begin{theorem}[Smooth histogram for differentially private continual release model]\label{thm:continual_release_smooth_histogram}
Let $g(\cdot)$ be an $(\zeta,\beta)$-smooth function.
If there exists a $\varepsilon'$-DP streaming continual release algorithm $\Lambda$ which maintains an $\left(\frac{1}{1-\eta}\right)$-approximation of $g(\cdot)$ simultaneously for all timestamps $t\in [T]$ with probability at least $1-\xi$, using space $h(\eta,\xi)$, then there is a $\varepsilon$-DP sliding window continual release algorithm $\Lambda'$ with $\varepsilon = O(\varepsilon'\beta/\log(T))$ which maintains a $\left(\frac{1}{1-\eta-\zeta}\right)$-approximation of $g(\cdot)$ over sliding windows for all timestamps $t\in [T]$ with probability at least $1-\xi$ and uses space $O\left(\frac{\log T}{\beta}\cdot h(\eta,\xi\beta/\log(T))\right)$.
\end{theorem}
\begin{proof}
The approximation guarantee, space guarantee and success probability follows from Theorem~\ref{thm:smooth_histogram_original} directly.
In the remaining of the proof, we prove the DP guarantee.

Let $\Lambda_1,\Lambda_2,\cdots,\Lambda_T$ be instances of $\Lambda$ where $\Lambda_t$ is started at timestamp $t$.
Let $o_{t,1},o_{t,2},\cdots,o_{t,T}$ be the outputs of $\Lambda_t$, if at timestamp $j$, $\Lambda_t$ is not started yet or is already kicked out by $\Lambda'$, then $o_{t,j}=\perp$.
According to Theorem~\ref{thm:smooth_histogram_original}, the outputs of $\Lambda'$ over all timestamps $t\in [T]$ is determined by $\{o_{i,j}\mid i,j\in[T]\}$.
Thus, we only need to show that $\{o_{i,j}\mid i,j\in[T]\}$ is $\varepsilon$-DP.
Consider two neighboring streams $\mathcal{S}$ and $\mathcal{S}'$ where only the $r$-th elements are different.
Let us fix a possible configuration $\{o_{i,j}\mid i,j\in[T]\}$.
According to Theorem~\ref{thm:smooth_histogram_original}, there are at most $O(\log(T)/\beta)$ different $i\in[T]$ such that $o_{i,r}\not=\perp$.
Let such set of $i$ to be $I$.
Since each $\Lambda_i$ for $i\in I$ is $\varepsilon'$-DP in the streaming continual releasing setting,
we have:
\begin{align*}
\Pr\left[\forall i\in I,j\in [T], o_{i,j}(\mathcal{S})=o_{i,j}\right]&\leq \exp(\varepsilon'\cdot |I|)\cdot \Pr\left[\forall i\in I,j\in [T], o_{i,j}(\mathcal{S}')=o_{i,j}\right]\\
&\leq \exp(\varepsilon)\cdot \Pr\left[\forall i\in I,j\in [T], o_{i,j}(\mathcal{S}')=o_{i,j}\right].
\end{align*}
On the other hand, we have:
\begin{align*}
&\Pr\left[\forall i\not\in I,j\in [T],o_{i,j}(\mathcal{S})=o_{i,j}\mid o_{i,j}(\mathcal{S})=o_{i,j}\forall i\in I,j\in[T] \right]\\
=& \Pr\left[\forall i\not\in I,j\in [T],o_{i,j}(\mathcal{S}')=o_{i,j}\mid o_{i,j}(\mathcal{S}')=o_{i,j}\forall i\in I,j\in[T] \right]
\end{align*}
Therefore, the algorithm is $\varepsilon$-DP.
\end{proof}

By combining Theorem~\ref{thm:summing_non_negative} with Lemma~\ref{lem:smoothness_additive_moment}, Lemma~\ref{lem:additive_to_relative_approx} and Theorem~\ref{thm:continual_release_smooth_histogram}, we are able to obtain the following sliding window continual release algorithm for summing non-negative numbers.
\begin{corollary}[Sliding window summing of a non-negative numbers]\label{cor:sliding_window_summing}
Let $\eta\in(0,0.5),\varepsilon\geq 0,\xi\in(0,0.5)$, there is an $\varepsilon$-DP algorithm for summing in the sliding window continual release model.
If the input numbers are guaranteed to be non-negative, with probability at least $1-\xi$, the output is always a $\left(1+\eta, O\left(\frac{\log(T/(\eta\xi))\log(T)}{\varepsilon\eta^3}\right)\right)$-approximation to the summing problem at any timestamp $t\in [T]$.
The algorithm uses space $O(\log(T)/\eta)$.
\end{corollary}

By combining Corollary~\ref{cor:distinct_large_universe_better_additive} with Lemma~\ref{lem:smoothness_additive_moment}, Lemma~\ref{lem:additive_to_relative_approx} and Theorem~\ref{thm:continual_release_smooth_histogram}, we are able to obtain the following sliding window continual release algorithm for number of distinct elements.

\begin{corollary}[Sliding window continual release distinct elements]\label{cor:sliding_window_distinct_elements}
For $\eta\in(0,0.5),\varepsilon\geq 0,\xi\in (0,0.5)$ there is an $\varepsilon$-DP algorithm for the number of distinct elements of streams with element universe $\mathcal{U}$ in the sliding window continual release model.
With probability at least $1-\xi$, the output is always a $(1+\eta,O\left(\frac{\log^2(T/(\eta\xi))\log(T)}{\eta^4\varepsilon}\right))$-approximation for every timestamp $t\in [T]$. 
The algorithm uses $\poly\left(\frac{\log(T/\xi)}{\eta\min(\varepsilon,1)}\right)$ space.
\end{corollary}

By combining Theorem~\ref{thm:count_sketch} with Lemma~\ref{lem:smoothness_additive_moment}, Lemma~\ref{lem:additive_to_relative_approx} and Theorem~\ref{thm:continual_release_smooth_histogram}, we are able to obtain the following sliding window continual release algorithm for $\ell_2$ frequency moment.

\begin{corollary}[Sliding window continual release $\ell_2$ frequency moments]\label{cor:sliding_window_l2}
Let $\varepsilon>0,\eta\in(0,0.5),\xi\in(0,0.5)$.
There is an $\varepsilon$-DP algorithm in the sliding window continual release model such that with probability at least $1-\xi$, it always outputs $\hat{F}_2$ for every timestamp $t\in [T]$
     such that $|\hat{F}_2-\|\mathcal{S}\|_2^2|\leq \eta \|\mathcal{S}\|_2^2+O\left(\frac{(\log(T/(\xi\eta))+\log(|\mathcal{U}|))^2\log^2(T)}{\varepsilon^2\eta^8}\cdot \log^{5}(T)\cdot \log^2\left(\frac{\log(T/\xi)+\log(|\mathcal{U}|)}{\xi\eta}\right)\right)$, where $\mathcal{S}$ denotes the sub-stream corresponding to the latest $W$ elements at timestamp $t$.
The algorithm uses $O\left(\frac{\log(T/(\xi\eta))+\log(|\mathcal{U}|)}{\eta^4}\cdot \log^2(T)\right)$ space.
\end{corollary}

By combining Theorem~\ref{thm:lp_moment} with Lemma~\ref{lem:smoothness_additive_moment}, Lemma~\ref{lem:additive_to_relative_approx} and Theorem~\ref{thm:continual_release_smooth_histogram}, we are able to obtain the following sliding window continual release algorithm for $\ell_p$ frequency moment.

\begin{corollary}[Sliding window continual release $\ell_p$ frequency moments]\label{cor:sliding_window_lp_moment}
Let $p>0,\varepsilon\geq 0,\xi \in (0,0.5),\eta\in(0,0.5)$.
There is an $\varepsilon$-DP algorithm in the sliding window continual release model such that with probability at least $1-\xi$, it always outputs an $\left(1+\eta, \left(\frac{\log(T|\mathcal{U}|/\xi)}{\eta\varepsilon}\right)^{O(p)}\right)$-approximation to $\|\mathcal{S}\|_p^p$, where $\mathcal{S}$ denotes the sub-stream corresponding to the latest $W$ elements at timestamp $t$.
The algorithm uses space at most
\begin{align*}
\phi\cdot \left(\frac{\log(T|\mathcal{U}|/\xi)}{\eta\varepsilon}\right)^{O(p)},
\end{align*}
where $\phi = \max(1,|\mathcal{U}|^{1-2/p})$.
\end{corollary}

\bibliographystyle{plainnat}
\bibliography{bib}

\appendix
\section{Missing Details of Section~\ref{sec:count}}
\subsection{Proof of Lemma~\ref{lem:dp_output_stream}}\label{sec:proof_of_grouping_dp}
\begin{proof}
%\jm{For this proof, can we directly cite the theorems from the book, or we are doing something slightly different?}
%\pz{The only difference is that we need to repeatedly apply the sparse vector technique.} \jm{I think the sparse vector technique itself is applying the AboveThreshold repeatedly, see Algorithm 2 in the book. The difference here is that here when we start a new run of the AboveThreshold algorithm, we don't use the data that was used in the previous AboveThreshold run. So we don't lose a factor $c$, where $c$ is the upper bound on the sparsity. } 
We first show that the output groups $G_1,G_2,\cdots,G_m$ is $\varepsilon_0$-DP.
Let $O_1,O_2,\cdots,O_m$ be any fixed grouping.
Let $c'_1,c'_2,\cdots,c'_T$ be any neighboring stream, i.e., $\exists q\in [T]$ such that $|c_q-c'_q|\leq 1$ and $\forall j\not=q,c_j=c'_j$.
Let $G'_1,G'_2,\cdots,G'_{m'}$ be the output groups of the neighboring stream.
Suppose $q\in O_r$ for some $r\in[m]$.
Let us consider $\Pr\left[(G_1,G_2,\cdots,G_m)=(O_1,O_2,\cdots,O_m)\right]$ and $\Pr\left[(G'_1,G'_2,\cdots,G'_{m'})=(O_1,O_2,\cdots,O_m)\right]$ (in this case $m'=m$).
We have:
\begin{align*}
&\frac{\Pr\left[(G_1,G_2,\cdots,G_m)=(O_1,O_2,\cdots,O_m)\right]}{\Pr\left[(G'_1,G'_2,\cdots,G'_m)=(O_1,O_2,\cdots,O_m)\right]}\\
=&\frac{\Pr\left[(G_1,G_2,\cdots,G_{r-1})=(O_1,O_2,\cdots,O_{r-1})\right]}{\Pr\left[(G'_1,G'_2,\cdots,G'_{r-1})=(O_1,O_2,\cdots,O_{r-1})\right]} \cdot \frac{\Pr\left[G_r=O_r\mid(G_1,G_2,\cdots,G_{r-1})=(O_1,O_2,\cdots,O_{r-1})\right]}{\Pr\left[G'_r=O_r\mid(G'_1,G'_2,\cdots,G'_{r-1})=(O_1,O_2,\cdots,O_{r-1})\right]}\\
&\cdot \frac{\Pr\left[(G_{r+1},G_{r+2},\cdots,G_m)=(O_{r+1},O_{r+2},\cdots,O_{m})\mid(G_1,G_2,\cdots,G_{r})=(O_1,O_2,\cdots,O_{r})\right]}{\Pr\left[(G'_{r+1},G'_{r+2},\cdots,G'_m)=(O_{r+1},O_{r+2},\cdots,O_m)\mid(G'_1,G'_2,\cdots,G'_{r})=(O_1,O_2,\cdots,O_{r})\right]}
\end{align*}
Since $\forall j\in O_1\cup O_2\cup \cdots \cup O_{r-1}, c_j = c'_j$, the behavior of running Algorithm~\ref{alg:DPgroup} on the prefix $O_1\cup O_2\cup \cdots \cup O_{r-1}$ of $c_1,c_2,\cdots,c_T$ is the same as the behavior of running it on the prefix $O_1\cup O_2\cup \cdots \cup O_{r-1}$ of $c'_1,c'_2,\cdots,c'_T$.
Therefore, we have $\frac{\Pr\left[(G_1,G_2,\cdots,G_{r-1})=(O_1,O_2,\cdots,O_{r-1})\right]}{\Pr\left[(G'_1,G'_2,\cdots,G'_{r-1})=(O_1,O_2,\cdots,O_{r-1})\right]}=1$.
Similarly, since $\forall j \in O_{r+1}\cup O_{r+2}\cup\cdots\cup O_m,c_j=c'_j$, when $G_r=G'_r=O_r$, the behavior of running Algorithm~\ref{alg:DPgroup} on the suffix $O_{r+1}\cup O_{r+2}\cup\cdots\cup O_m$ of $c_1,c_2,\cdots,c_T$ is the same as the behabior of running it on the same suffix of $c'_1,c'_2,\cdots,c'_T$.
Therefore, we have $\frac{\Pr\left[(G_{r+1},G_{r+2},\cdots,G_m)=(O_{r+1},O_{r+2},\cdots,O_{m})\mid(G_1,G_2,\cdots,G_{r})=(O_1,O_2,\cdots,O_{r})\right]}{\Pr\left[(G'_{r+1},G'_{r+2},\cdots,G'_m)=(O_{r+1},O_{r+2},\cdots,O_m)\mid(G'_1,G'_2,\cdots,G'_{r})=(O_1,O_2,\cdots,O_{r})\right]}$.
Thus, we have:
\begin{align*}
&\frac{\Pr\left[(G_1,G_2,\cdots,G_m)=(O_1,O_2,\cdots,O_m)\right]}{\Pr\left[(G'_1,G'_2,\cdots,G'_m)=(O_1,O_2,\cdots,O_m)\right]}\\
=&\frac{\Pr\left[G_r=O_r\mid(G_1,G_2,\cdots,G_{r-1})=(O_1,O_2,\cdots,O_{r-1})\right]}{\Pr\left[G'_r=O_r\mid(G'_1,G'_2,\cdots,G'_{r-1})=(O_1,O_2,\cdots,O_{r-1})\right]}.
\end{align*}
Suppose $O_r=\{x+1,\cdots,x+k\}$.

Consider the first case where $r\not=m$.
In this case, we have 
\begin{align*}
&\Pr\left[G_r=O_r\mid (G_1,G_2,\cdots,G_{r-1})=(O_1,O_2,\cdots,O_{r-1})\right]\\
=&\Pr\left[G_r=O_r\mid G_{r-1}=O_{r-1}\right]\\
=&\Pr\left[\left(\forall j\in [k-1], \nu_{x+j}+\sum_{b=1}^j c_{x+b} < \tau_r \right)\bigwedge \left(\nu_{x+k}+\sum_{b=1}^k c_{x+b}\geq \tau_r\right) \right].
\end{align*}
Now, let us fix $\nu_{x+1},\nu_{x+2},\cdots,\nu_{x+k-1}$ and let $g=\max_{j\in[k-1]}\nu_{x+j}+\sum_{b=1}^j c_{x+b}$.
Then,
\begin{align}
&\Pr_{\nu_{x+k},\tau_r}\left[G_r = O_r\mid G_{r-1}=O_{r-1}\right]\notag\\
=&\Pr_{\nu_{x+k},\tau_r}\left[\tau_r\in(g,\nu_{x+k}+\sum_{b=1}^k c_{x+b}]\right]\notag\\
=&\int_{-\infty}^{\infty}\int_{-\infty}^{\infty} p_{\nu_{x+k}}(v)\cdot p_{\tau_r}(\tau)\cdot \mathbf{1}\left(\tau \in (g, v + \sum_{b=1}^k c_{x+b})\right) \mathrm{d}v\mathrm{d}\tau \label{eq:sparse_vec_tech}
\end{align}
where $p_{\nu_{x+k}}(\cdot)$ and $p_{\tau_r}(\cdot)$ are density functions of $\nu_{x+k}$ and $\tau_r$ respectively.
Let $g' = \max_{j\in[k-1]}\nu_{x+j}+\sum_{b=1}^j c'_{x+b}$.
Let $v' = v + g - g'+\sum_{b=1}^k c'_{x+b} - \sum_{b=1}^k c_{x+b}$.
Let $\tau' = \tau + g - g'$.
Since $|c_q-c'_q|\leq 1$, it is easy to see that $|v'-v|\leq 2$ and $|\tau-\tau'|\leq 1$.
Note that $\mathrm{d}v'=\mathrm{d}v$ and $\mathrm{d}\tau'=\mathrm{d}\tau$.
Therefore, Equation~\eqref{eq:sparse_vec_tech} is equal to the following:
\begin{align*}
&\int_{-\infty}^{\infty}\int_{-\infty}^{\infty}p_{\nu_{x+k}}(v')\cdot p_{\tau_r}(\tau')\cdot \mathbf{1}\left(\tau+g-g'\in\left(g,v+g-g'+\sum_{b=1}^k c'_{x+b}\right)\right)\mathrm{d}v\mathrm{d}\tau\\
=&\int_{-\infty}^{\infty}\int_{-\infty}^{\infty}p_{\nu_{x+k}}(v')\cdot p_{\tau_r}(\tau')\cdot \mathbf{1}\left(\tau\in\left(g',v+\sum_{b=1}^k c'_{x+b}\right)\right)\mathrm{d}v\mathrm{d}\tau\\
\leq & \int_{-\infty}^{\infty}\int_{-\infty}^{\infty}\exp(\varepsilon_0/2)\cdot p_{\nu_{x+k}}(v)\cdot\exp(\varepsilon_0/2)\cdot p_{\tau_r}(\tau)\cdot \mathbf{1}\left(\tau\in\left(g',v+\sum_{b=1}^k c'_{x+b}\right)\right)\mathrm{d}v\mathrm{d}\tau\\
=&\exp(\varepsilon_0)\cdot\Pr_{\nu_{x+k},\tau_r}\left[\tau_r\in(g',\nu_{x+k}+\sum_{b=1}^k c'_{x+b}]\right]\\
=&\exp(\varepsilon_0)\cdot \Pr_{\nu_{x+k},\tau_r}\left[G'_r=O_r\mid G'_{r-1}=O_{r-1}\right]\\
=&\exp(\varepsilon_0)\cdot \Pr_{\nu_{x+k},\tau_r}\left[G'_r=O_r\mid (G'_1,G'_2,\cdots,G'_{r-1})=(O_1,O_2,\cdots,O_{r-1})\right].
\end{align*}

Next, consider the second case where $r=m$. 
In this case, we have
\begin{align*}
&\Pr\left[G_r=O_r\mid (G_1,G_2,\cdots,G_{r-1})=(O_1,O_2,\cdots,O_{r-1})\right]\\
=&\Pr\left[G_r=O_r\mid G_{r-1}=O_{r-1}\right]\\
=&\Pr\left[\forall j\in [k],\nu_{x+j}+\sum_{b=1}^j c_{x+b}<\tau_r\right].
\end{align*}
Now, let us fix $\nu_{x+1},\nu_{x+2},\cdots \nu_{x+k}$ and let $g=\max_{j\in[k]} \nu_{x+j} + \sum_{b=1}^j c_{x+b}$.
Let $g'=\max_{j\in[k]}\nu_{x+j} +\sum_{b=1}^j c'_{x+b}$.
Since $|c_q-c'_q|\leq 1$, we have $|g-g'|\leq 1$.
Then,
\begin{align*}
\Pr_{\tau_r}\left[\tau_r>g\right]\leq \exp(\varepsilon_0)\cdot \Pr_{\tau_r}\left[\tau_r>g'\right].
\end{align*}
Thus, we have 
\begin{align*}
&\Pr\left[G_r=O_r\mid (G_1,G_2,\cdots,G_{r-1})=(O_1,O_2,\cdots,O_{r-1})\right]\\
\leq& \exp(\varepsilon_0)\cdot \Pr\left[G'_r=O_r\mid (G'_1,G'_2,\cdots,G'_{r-1})=(O_1,O_2,\cdots,O_{r-1})\right].
\end{align*}
Therefore, we can conclude that $(G_1,G_2,\cdots,G_m)$ is always $\varepsilon_0$ DP.

Notice that, given any fixed $(O_1,O_2,\cdots,O_m)$ and condition on $(G_1,G_2,\cdots,G_m)=(O_1,O_2,\cdots,O_m)$, $\hat{c}_1,\hat{c}_2,\cdots,\hat{c}_T$ is $\varepsilon_0$-DP by Laplace mechanism.
Thus, by composition theorem, the output stream $\hat{c}_1,\hat{c}_2,\cdots,\hat{c}_T$ is $\varepsilon$-DP.
\end{proof}

\subsection{Proof of Lemma~\ref{lem:acc_output_stream}}\label{sec:proof_of_acc_grouping}

Let $G_1,G_2,\cdots,G_m$ be the groups produced during Algorithm~\ref{alg:DPgroup}.
Let $\tilde{c}_1,\tilde{c}_2,\cdots,\tilde{c}_m$ be that $\forall i\in [m],\tilde{c}_i=\hat{c}_{\max(G_i)}$, i.e., $\tilde{c}_i$ is the noisy count of the group $G_i$.

\begin{lemma}\label{lem:group_auc}
With probability at least $1-\xi$, the output stream of Algorithm~\ref{alg:DPgroup} satisfies the following properties:
\begin{enumerate}
%    \item $\forall i\in [m-1],\sum_{j\in G_i} c_j\geq \frac{7}{\eta\varepsilon_0}\cdot \ln\left(3\cdot T/\xi\right)$.\label{it:prop_lb}
    \item $\forall i\in [m-1],\sum_{j\in G_i\setminus \{\max_{j'\in G_i}j'\}} c_j \leq \frac{7}{\eta\varepsilon_0}\cdot \ln\left(3\cdot T/\xi\right)+\frac{13}{\varepsilon_0}\cdot \ln(3\cdot T/\xi)$.\label{it:prop_ub}
    \item $\sum_{j\in G_m} c_j \leq \frac{7}{\eta\varepsilon_0}\cdot \ln\left(3\cdot T/\xi\right)+\frac{13}{\varepsilon_0}\cdot \ln(3\cdot T/\xi)$. \label{it:prop_last_group}
    \item $\forall i\in[m-1],(1-\eta)\sum_{j\in G_i}c_j\leq \tilde{c}_i\leq (1+\eta)\sum_{j\in G_i}c_j$. \label{it:prop_relative_approx}
%    \item $|\tilde{c}_m - \sum_{j\in G_m} c_j|\leq \frac{1}{\varepsilon_0}\cdot \ln(3\cdot T/\xi)$. 
\end{enumerate}
\end{lemma}
\begin{proof}

Let $\mathcal{E}$ denote the event that 
\begin{enumerate}
    \item $\forall i\in [m], \left|\tau_i - \left(\frac{1}{\eta}+1\right)\cdot \frac{7}{\varepsilon_0}\cdot\ln\left(3\cdot T/\xi\right)\right|\leq \frac{2}{\varepsilon_0}\cdot \ln\left(3\cdot T/\xi\right)$.
    \item $\forall t \in [T],|\nu_t|\leq \frac{4}{\varepsilon_0}\cdot \ln\left(3\cdot T/\xi\right)$.
    \item $\forall i \in [m-1], |\tilde{c}_i - \sum_{j\in G_i} c_j|\leq \frac{1}{\varepsilon_0}\cdot \ln(3\cdot T/\xi)$
\end{enumerate}
According to the CDF of Laplace noise, it is easy to show that $\mathcal{E}$ happens with probability at least $1-\xi$ by a union bound over all $i\in [m],t\in[T]$.
In the remaining of the proof we condition on the event $\mathcal{E}$.

%Property~\ref{it:prop_last_est} directly follows from event $\mathcal{E}$.

%Consider property~\ref{it:prop_lb}.

Consider property~\ref{it:prop_ub}.
Consider any $i\in[m-1]$.
Let $t=\max_{j\in G_i} j - 1$.
Then we have $\nu_t + \sum_{j\in G_i,j\leq t} c_j\leq \tau_i$.
Since $|\nu_t|\leq \frac{4}{\varepsilon_0}\cdot \ln\left(3\cdot T/\xi\right)$ and $\tau_i\leq \frac{7}{\eta\cdot \varepsilon_0}\cdot \ln\left(3\cdot T/\xi\right)+\frac{9}{\varepsilon_0}\cdot \ln(3\cdot T/\xi)$, we have $\sum_{j\in G_i,j\leq t} c_j\leq \frac{7}{\eta\varepsilon_0}\cdot \ln(3\cdot T/\xi)+\frac{13}{\varepsilon_0}\cdot \ln(3\cdot T/\xi)$.

The proof of property~\ref{it:prop_last_group} is the same as the proof of property~\ref{it:prop_ub}.

Consider property~\ref{it:prop_relative_approx}.
Consider any $i\in[m-1]$.
Let $t=\max_{j\in G_i} j$.
Then we have $\nu_t + \sum_{j\in G_i} c_j\geq \tau_i$.
Since $|\nu_t|\leq \frac{4}{\varepsilon_0}\cdot \ln\left(3\cdot T/\xi\right)$ and $\tau_i\geq \frac{7}{\eta\cdot \varepsilon_0}\cdot \ln\left(3\cdot T/\xi\right)+\frac{5}{\varepsilon_0}\cdot \ln(3\cdot T/\xi)$, we have $\sum_{j\in G_i} c_j\geq \frac{7}{\eta\varepsilon_0}\cdot \ln(3\cdot T/\xi)$.
Note that $|\tilde{c}_i-\sum_{j\in G_i} c_j| \leq \frac{1}{\varepsilon_0}\cdot \ln\left(3\cdot T/\xi\right)$.
Thus, we know that $\sum_{j\in G_i}c_j \geq \frac{7}{\eta\varepsilon_0}\cdot \ln\left(3\cdot T/\xi\right)$.
Thus, $\forall i\in [m-1]$ we have $(1-\eta)\sum_{j\in G_i}c_j\leq \tilde{c}_i\leq (1+\eta)\sum_{j\in G_i}c_j$.
\end{proof}

%According to above lemma, we can reduce the sliding window counting to sliding window group counting.

\begin{comment}
\begin{lemma}
Let $(G_1,\hat{c}_1),(G_2,\hat{c}_2),\cdots,(G_m,\hat{c}_m)$ be the output stream of Algorithm~\ref{alg:DPgroup}.
Then with probability at least $1-\xi$, $\forall l,r$ satisfying $1\leq l\leq r\leq T$,
\begin{align*}
(1-\eta)    \sum_{j=l}^r c_j - \left(\frac{1}{\eta}+2\right)\cdot \frac{7}{\varepsilon_0}\cdot \ln\left(3\cdot T/\xi\right)\leq \sum_{j'=l'}^{r'} \hat{c}_{j'} \leq (1+\eta) \sum_{j=l}^r c_j + \left(\frac{1}{\eta}+2\right)\cdot \frac{7}{\varepsilon_0}\cdot \ln\left(3\cdot T/\xi\right),
\end{align*}
where $l'\in[m]$ is the smallest index such that $\max_{j\in G_{l'}} j \geq l$ and $r'\in [m]$ is the largest index such that $\max_{j \in G_{r'}} j \leq r$.
\end{lemma}
\end{comment}

Now, we are able to prove Lemma~\ref{lem:acc_output_stream}.
\begin{proof}[Proof of Lemma~\ref{lem:acc_output_stream}]
With probability at least $1-\xi$, the properies listed in Lemma~\ref{lem:group_auc} hold.

Let $l'\in[m]$ be the smallest index such that $\max_{j\in G_{l'}} j \geq l$ and $r'\in [m]$ be the largest index such that $\max_{j \in G_{r'}} j \leq r$.

We have:
\begin{align*}
\sum_{j=l}^r \hat{c}_j
= &\sum_{j'=l'}^{r'} \tilde{c}_{j'}\\
\geq & (1-\eta)\sum_{j'=l'}^{r'} \sum_{j\in G_{j'}}c_j  - \left(\frac{7}{\eta\varepsilon_0}\cdot \ln\left(3\cdot T/\xi\right)+\frac{13}{\varepsilon_0}\cdot \ln(3\cdot T/\xi)\right)\\
\geq & (1-\eta)\sum_{j = l}^{\max_{b\in G_{r'}} b} c_j -\left(\frac{7}{\eta\varepsilon_0}\cdot \ln\left(3\cdot T/\xi\right)+\frac{13}{\varepsilon_0}\cdot \ln(3\cdot T/\xi)\right)\\
\geq & (1-\eta)\sum_{j = l}^{r} c_j  - \frac{7}{\eta\varepsilon_0}\cdot \ln\left(3\cdot T/\xi\right) - \frac{26}{\varepsilon_0}\cdot \ln\left(3\cdot T/\xi\right),
\end{align*}
where the first inequality follows from property~\ref{it:prop_relative_approx} and property~\ref{it:prop_last_group} of Lemma~\ref{lem:group_auc}, the second inequality follows from the choice of $l'$, and the third inequality follows from property~\ref{it:prop_ub} of Lemma~\ref{lem:group_auc}.

Similarly, we can show
\begin{align*}
\sum_{j=l}^r \hat{c}_j 
=&\sum_{j'=l'}^{r'} \tilde{c}_{j'}\\
\leq & (1+\eta)\sum_{j'=l'}^{r'} \sum_{j\in G_{j'}}c_j  + \left(\frac{7}{\eta\varepsilon_0}\cdot \ln\left(3\cdot T/\xi\right)+\frac{13}{\varepsilon_0}\cdot \ln(3\cdot T/\xi)\right)\\
\leq &  (1+\eta)\sum_{j = \min_{b\in G_{l'}} b}^{r} c_j + \left(\frac{7}{\eta\varepsilon_0}\cdot \ln\left(3\cdot T/\xi\right)+\frac{13}{\varepsilon_0}\cdot \ln(3\cdot T/\xi)\right)\\
\leq & (1+\eta)\sum_{j = l}^{r} c_j + \frac{7}{\eta\varepsilon_0}\cdot \ln\left(3\cdot T/\xi\right) + \frac{26}{\varepsilon_0}\cdot \ln\left(3\cdot T/\xi\right),
\end{align*}
where the first inequality follows from property~\ref{it:prop_relative_approx} and property~\ref{it:prop_last_group} of Lemma~\ref{lem:group_auc}, the second inequality follows from the choice of $r'$, and the third inequality follows from property~\ref{it:prop_ub} of Lemma~\ref{lem:group_auc}.
\end{proof}

\end{document}